\documentclass[11pt]{report}         
\pdfoutput=1
\usepackage{thesis}              
\usepackage{url}
\usepackage{graphicx}
\usepackage{subfigure}
\usepackage{subfloat}
\usepackage[usenames]{color}
\usepackage{listings}
\usepackage{tabulary} 
\usepackage{array}
\usepackage{amsmath}
\usepackage{amssymb}

\usepackage{algorithm2e}
\usepackage[table]{xcolor}    
\usepackage[toc,page]{appendix}
\usepackage{rotating}

\phdthesis                         

\doublespace                       

\renewcommand{\thesisauthor}{Noah Manus Daniels}    

\renewcommand{\thesismonth}{April}   

\renewcommand{\thesisyear}{2013}      

\renewcommand{\thesistitle}{Remote Homology Detection in Proteins Using Graphical Models}

\renewcommand{\thesisauthorpreviousdegrees}{B.S., M.S.}

\renewcommand{\thesissupervisor}{Prof. Lenore Cowen}

\renewcommand{\thesiscommitteemembera}{Prof. Donna Slonim}
\renewcommand{\thesiscommitteememberb}{Prof. Benjamin Hescott}
\renewcommand{\thesiscommitteememberc}{Prof. Bonnie Berger}
\renewcommand{\thesiscommitteememberd}{Prof. Yu-Shan Lin}

\renewcommand{\thesisauthoraddress}{Halligan Hall\\
Tufts University\\
161 College Ave.
Medford, MA 02155}

\renewcommand{\thesisdedication}{{\em For my grandmother}}

\newtheorem{theorem}{Theorem}[section]
\newtheorem{lemma}[theorem]{Lemma}

\newenvironment{proof}[1][Proof]{\begin{trivlist}
\item[\hskip \labelsep {\bfseries #1}]}{\end{trivlist}}
\newenvironment{definition}[1][Definition]{\begin{trivlist}
\item[\hskip \labelsep {\bfseries #1}]}{\end{trivlist}}

  \long\def\remark#1{%
      \ifvmode
         \marginpar{\raggedright\hbadness=10000
         \parindent=8pt \parskip=2pt
         \def\baselinestretch{0.8}\tiny
         \itshape\noindent #1\par}%
      \else
          \unskip\raisebox{-3.5pt}{\rlap{$\scriptstyle\diamond$}}%
          \marginpar{\raggedright\hbadness=10000
         \parindent=8pt \parskip=2pt
         \def\baselinestretch{0.8}\tiny
         \itshape\noindent #1\par}%
      \fi}

\newcommand{\qed}{\nobreak \ifvmode \relax \else
      \ifdim\lastskip<1.5em \hskip-\lastskip
      \hskip1.5em plus0em minus0.5em \fi \nobreak
      \vrule height0.75em width0.5em depth0.25em\fi}

\begin{document}

\thesiscopyrightpage                 

\thesiscertificationpage             

\thesistitlepage                     


\thesisdedicationpage                

\begin{thesisacknowledgments}        
The work described in this dissertation is the result of the advice, 
contributions, and collaborations of many friends and colleagues.

First, I would like to thank my advisor, Professor Lenore Cowen, for many
years of mentorship.
Her encouragement, patience, and motivation have been essential to this work.
I thank my dissertation committee members, Professor
Donna Slonim, Professor Benjamin Hescott, Professor Yu-Shan Lin, and Professor
Bonnie Berger.

My collaborators during the course of my doctoral work have also been inspiring.
I thank Anoop Kumar, Matt Menke, Raghavendra Hosur, Shilpa Nadimpalli, Andrew 
Gallant, Po-Ru Loh, Michael Baym, Jian Peng, Jisoo Park, and Mengfei Cao for 
their contributions, be they in code, conversation, or collegiality.

I thank Professor Norman Ramsey, Professor Sinaia Nathanson, and 
Dean Lynne Pepall, along with my GIFT colleagues, for teaching me how to teach.
I also thank my teaching assistants, Sarah Nolet, Joel Greenberg, Andrew 
Pellegrini, and Michael Pietras, for helping me teach, Nathan Ricci for the
constant feedback and guest lectures, and Professors Carla
Brodley and Diane Souvaine for the opportunities.

I thank the Tufts University Computer Science Department, especially Gail 
Fitzgerald, Jeannine Vangelist, and Donna Cirelli, for so much support over so 
many years.

I owe a debt of gratitude to Michael Bauer, Erik Patton, Jon Frederick, George
Preble, and Eric Berg for keeping our systems running despite my best efforts
to the contrary. 

Much of the material in Chapter 2 of this dissertation has been published 
as ``Touring Protein Space with Matt'', with Anoop Kumar, Matt Menke, and
Lenore Cowen,
in the journal \emph{ACM Transactions on Computational Biology and 
Bioinformatics.}
Much of the material in Chapter 3 of this dissertation has appeared 
as ``SMURFLite: combining simplified Markov random fields with simulated 
evolution improves remote homology detection for beta-structural proteins into 
the twilight zone'', with Raghavendra Hosur, Bonnie Berger, and Lenore Cowen,
in the journal \emph{Bioinformatics}.
Some of the material in Chapter 4 of this dissertation has appeared as an
experience report, ``Experience Report: Haskell in Computational Biology'',
with Andrew Gallant and Norman Ramsey, in 
the Proceedings of the \emph{International Conference on Functional 
Programming}.

Finally, I would not have reached this point without the love and support of my
parents, Anne and Norman Daniels, and my wife, Rachel Daniels.
\end{thesisacknowledgments}

\begin{thesisabstract}

Given the amino acid sequence of a protein, researchers often infer its
structure and function by finding homologous, or evolutionarily-related, 
proteins of known structure and function.
Since structure is typically more conserved than sequence
over long evolutionary distances, recognizing remote protein homologs
from their sequence poses a challenge.

We first consider all proteins of known three-dimensional structure, and 
explore how they cluster according to different levels of homology. 
An automatic computational method reasonably approximates
a human-curated hierarchical organization of proteins according to their
degree of homology.

Next, we return to homology prediction, based only on the one-dimensional
amino acid sequence of a protein. 
Menke, Berger, and Cowen proposed a Markov random field model to predict 
remote homology for beta-structural proteins, but their formulation was
computationally intractable on many
beta-strand topologies.

We show two different approaches
to approximate this random field, both of which make it computationally
tractable, for the first time, on all protein folds. 
One method simplifies the random field itself, while the other retains 
the full random field, but approximates the solution through 
stochastic search. 
Both methods achieve
improvements over the state of the art in remote homology detection
for beta-structural protein folds.

\nopagebreak
\end{thesisabstract}   

\tableofcontents                   
\listoftables                      
\listoffigures                     
\chapter{Introduction}

\label{chapter:introduction}

\section{Proteins}

Proteins are the molecular machines that are essential to the process of life.
For example, transmembrane proteins allow molecules to move into and out of the 
cell.
Hemoglobin ferries iron through the blood, while immunoglobulin provides
for defense against pathogens.
Actin contracts our muscles, and myelin insulates our nerves.

It is well known that DNA encodes the genetic information that determines how
we develop and function.
Portions of this DNA are transcribed into RNA, and then 
a complex piece of cellular machinery called the ribosome translates this RNA 
into amino acids, the building blocks of proteins.
Proteins are the machines for which the
DNA is the blueprint.
Chains of amino acids fold into intricate, low-energy
forms, and these structures {\emph do things}.

It is the structure of
a protein that allows it to perform its function, and while this structure 
is determined by the
amino acid sequence that derives from DNA, the relationship between sequence
and structure is not simple.


\subsection{Primary Structure}

Proteins are composed of linear chains of molecules called \textit{amino acids}.
An amino acid is a molecule comprising an \textit{amine} group, a 
\textit{carboxyl} group, and one of twenty possible \textit{sidechains} (see
Figure~\ref{amino_acid}).
Each of these components is attached to a carbon atom, known as the 
$\alpha$-carbon.

\begin{figure}[htb!]
\begin{center}
  \fbox{\includegraphics[width=5in]{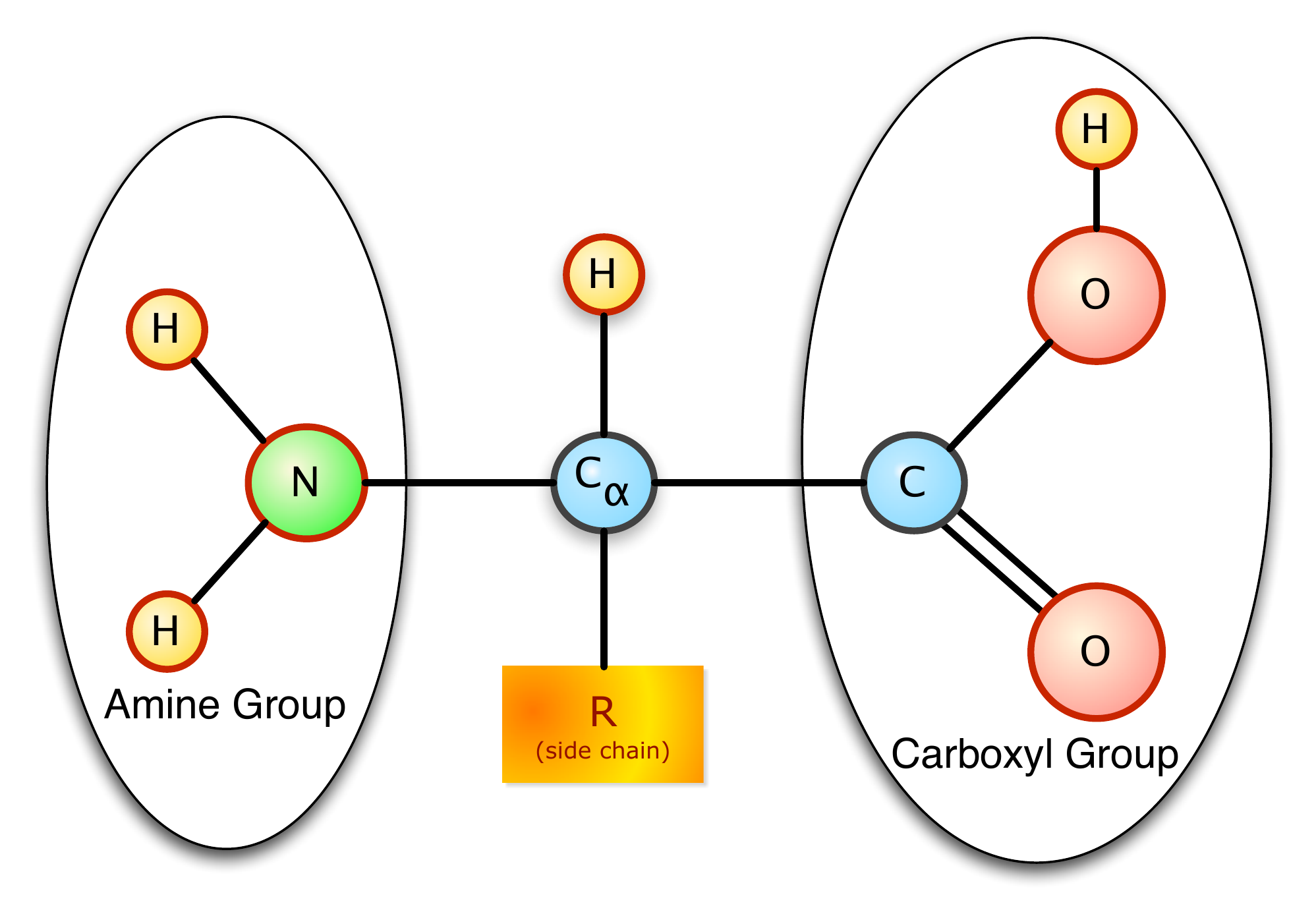}}
   \caption[The general structure of  an amino acid showing the hydrogen(H), 
   nitrogen(N), oxygen(O) and carbon(C,$\text{C}_{\alpha}$) atoms.]{The general 
   structure of  an amino acid showing the hydrogen(H), nitrogen(N), oxygen(O) 
   and carbon(C,$\text{C}_{\alpha}$) atoms.
   The sidechain is one of twenty possible ``decorations;'' amino acids differ
   only in their sidechains. }
   \label{amino_acid}
 \end{center}
\end{figure}

Amino acids bind to one another via a \textit{peptide bond}, which forms when
the carboxyl group of one amino acid gives up an oxygen and hydrogen to bind
with the amine group of another amino acid, which gives up a hydrogen. 
This results in a free water molecule. 
In addition, as multiple amino acids form \textit{polypeptide chains}, the 
unbound amine group at one end is known as the \textit{N-terminal} end of the
resulting protein, while the unbound carboxyl group at the other end is known
as the \textit{C-terminal} end (Figure~\ref{polypeptide}).

\begin{figure}[htb!]
\advance\leftskip-0.3in
\begin{center}
  \fbox{\includegraphics[width=5in] {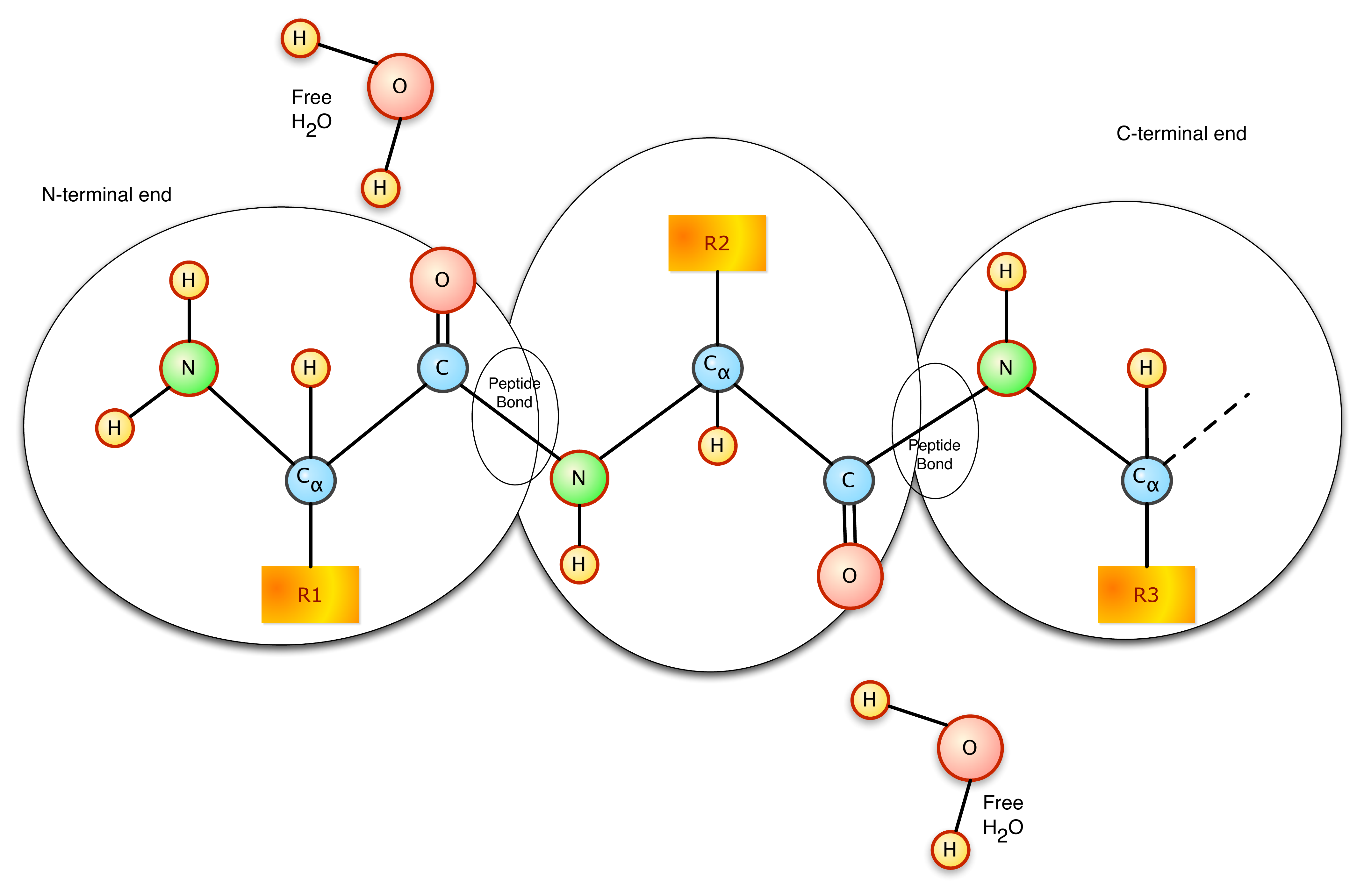}}
   \caption[Peptide bonds and the protein backbone]{Amino acids join by peptide 
   bond to form the backbone.  
   Individual amino acids are highlighted with ovals.
   The sidechain connected to the $\text{C}_{\alpha}$  is different for each 
   amino acid. 
    When a peptide bond forms, a water molecule forms from the hydrogen given
    up by the nitrogen end of one amino acid, and the oxygen and hydrogen given
    up by the carbon end of the other.}
   \label{polypeptide}
 \end{center}
\end{figure}

The peptide linkages, along with the $\alpha$-carbon atoms, form the 
\textit{backbone} of the protein.
Ultimately, the protein folds into a globular form, generally representing a
lowest-energy conformation.
It is useful to describe the structure of proteins at several levels of
organization.

The \textit{primary structure} of a protein is simply its sequence of amino
acids.
In principal, any of the 20 standard amino acids can occur in any
position of an amino acid chain; for a protein of length $n$, there are $20^{n}$
possible protein sequences.
Of course, the subset of those sequences that will fold into a compact, 
three-dimensional structure is much smaller; the subset of \emph{those} that 
would fold into a compact, three-dimensional structure that exists in nature is
smaller still.
However, determining which protein sequences nature allows is not 
trivial.

\subsection{Secondary Structure}

Local interactions among amine and carbonyl groups result in \textit{hydrogen
bonds} between amino acids that are not immediately adjacent in sequence.
A hydrogen bond is the electrostatic attraction between a hydrogen atom in
one amino acid and an oxygen or nitrogen atom in another.
In particular, we can describe the \textit{secondary structure} of a protein
according to the shape of the angles of the backbone.
The most common type of secondary structure is the $\alpha$-helix, in which the
protein backbone coils into a twisted shape, stabilized by hydrogen bonds 
(Figure~\ref{helix}).
The most common $\alpha$-helices have hydrogen bonds between residues
four positions apart in sequence.
Other, less common helical structures include the $3_{10}$ helix, in which
residues three apart in sequence form hydrogen bonds, and the $\pi$ helix, in 
which residues five apart in sequence form hydrogen bonds.

Another secondary structure is the $\beta$-strand, which in combination form 
$\beta$-sheets.
$\beta$-strands occur when the backbone is stretched out; typically, this 
conformation is stabilized by hydrogen bonds between adjacent strands
(Figure~\ref{sheet}), resulting in $\beta$-sheets.
The hydrogen bonds in $\beta$-sheets may occur between residues that are very
far apart from each other in the amino acid sequence.
$\beta$-strands in a sheet may be parallel or anti-parallel to one another with
respect to the direction of the amino acid sequence.

The remainder of local backbone conformations, consisting of turns, bulges,
loops, bridges, etc., have been classified into several different subcategories,
but is often grouped together into a third category of secondary structure,
commonly referred to as a \emph{coil}.

\begin{figure}[htb!]
\advance\leftskip-0.3in
\begin{center}
  \fbox{\includegraphics[width=5in] {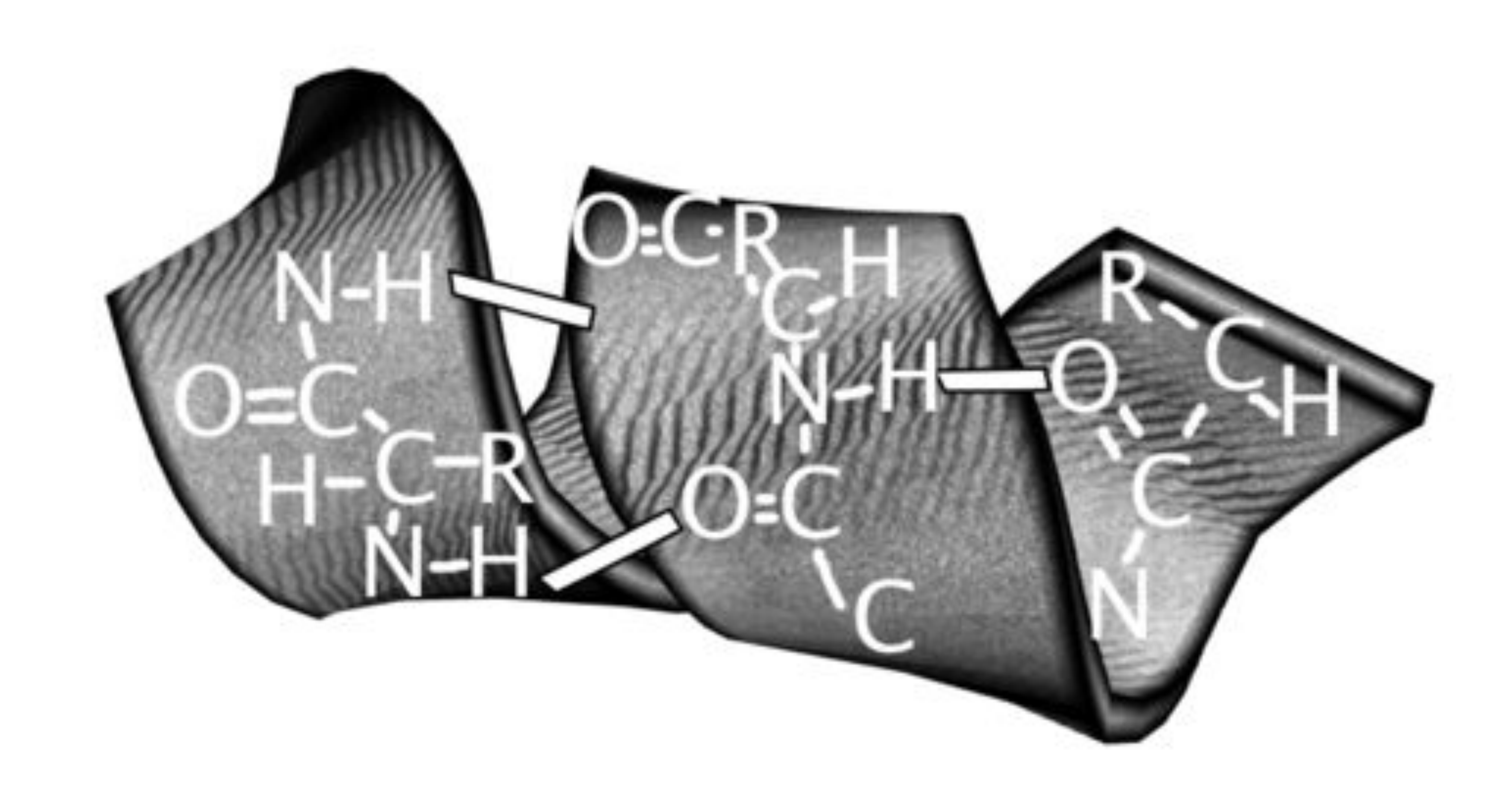}}
   \caption[$\alpha$-helix secondary structure.]{$\alpha$-helix secondary 
   structure.
   Hydrogen bonds between residues 4 positions apart in sequence cause the
   helical shape.
   Other, less common helix structures include the $3_{10}$ helix, in which
   residues 3 apart in sequence form hydrogen bonds, and the $\pi$ helix,
   in which residues 5 apart in sequence form hydrogen bonds.}
   \label{helix}
 \end{center}
\end{figure}

\begin{figure}[htb!]
\advance\leftskip-0.3in
\begin{center}
  \fbox{\includegraphics[width=5in] {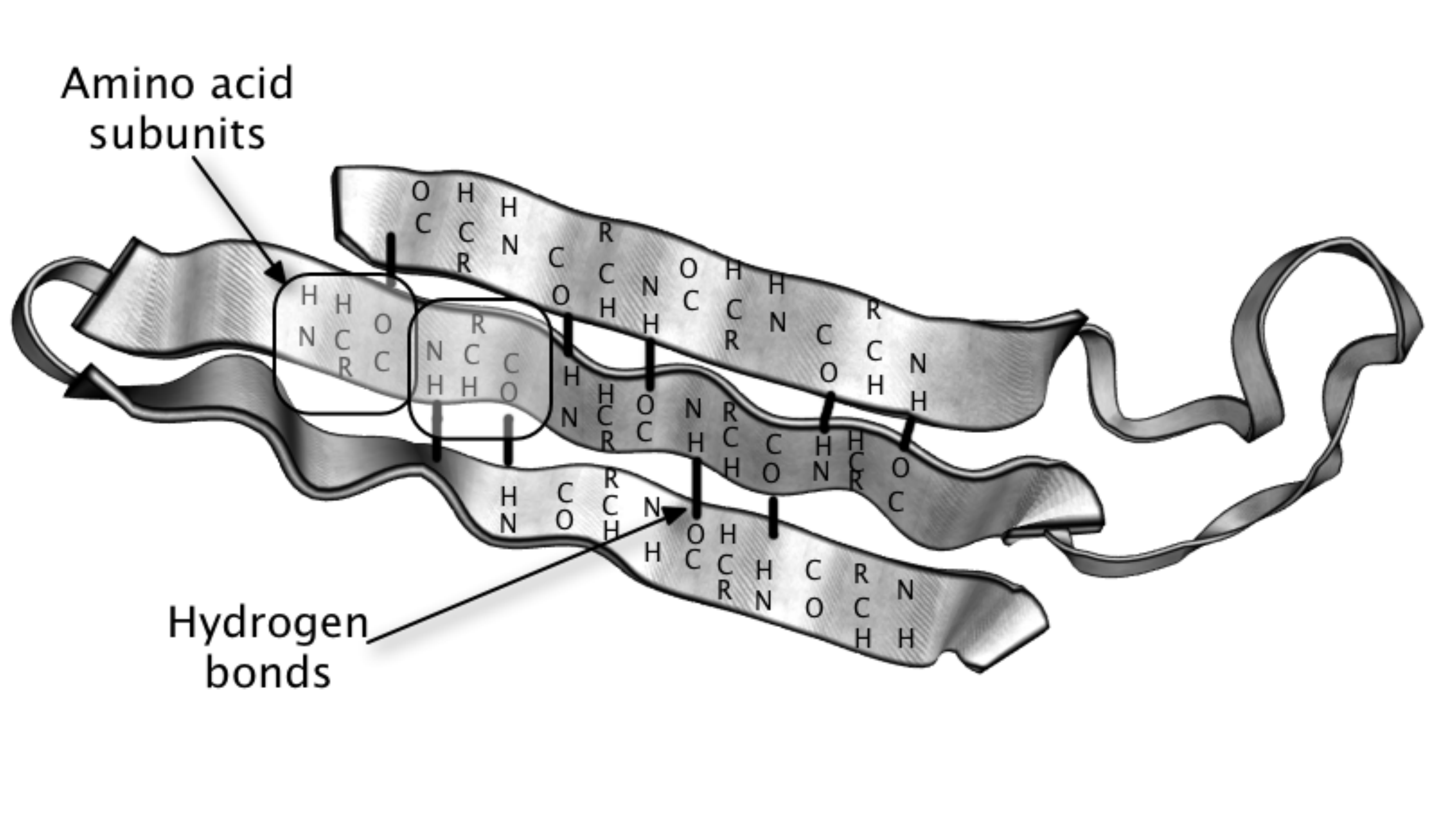}}
   \caption[$\beta$-sheet secondary structure.]{$\beta$-sheet secondary 
   structure.
   Hydrogen bonds between residues that may be quite far apart in sequence
   cause this pleated, sheet-like shape.
   Antiparallel $\beta$-strands are shown here; parallel $\beta$-strands also 
   exist.
   }
   \label{sheet}
 \end{center}
\end{figure}

\subsection{Supersecondary and Tertiary Structure}

We can mark secondary structural elements of the complete structure of a protein
backbone as it is folded in three-dimensional space, and consider the pattern
of where the $\alpha$-helices and $\beta$-strands lie.
For example, $\beta$-strands can be organized into $\beta$-barrels 
(Figure~\ref{barwin}), sandwiches,
or propellers; $\alpha$-helices can be organized into 2- or 4-helix bundles, and
there are other patterns of strand topologies that involve mixed collections of
$\alpha$-helices and $\beta$-strands.
The topologies of the various strand positions are known as 
\emph{super-secondary structure}.

\begin{figure}[ht!]
\begin{center}
  \fbox{\includegraphics[width=5in] {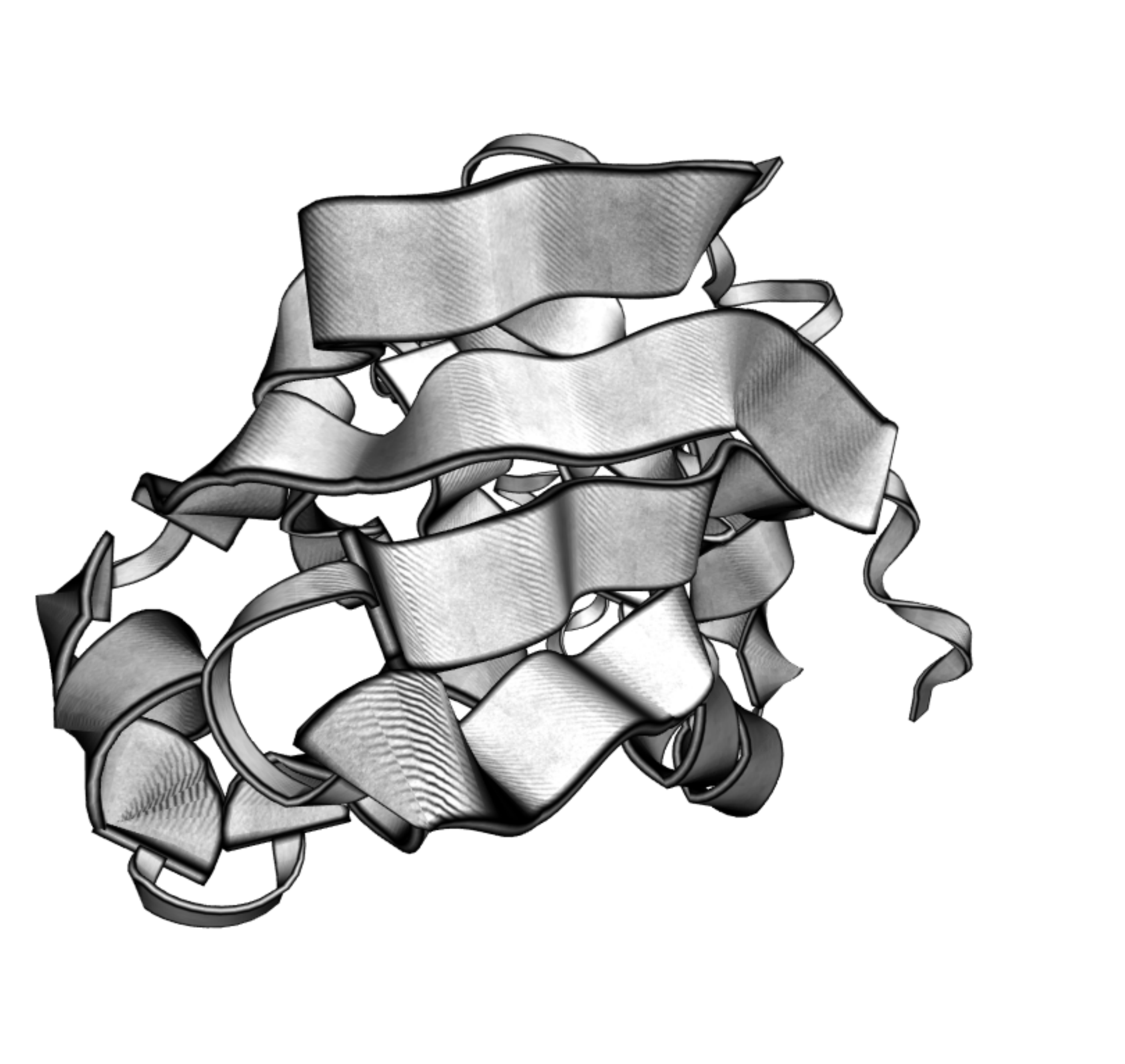}}
   \caption[Super-secondary structure ``cartoon'' of Barwin (PDB ID 
   1BW3)]{Super-secondary structure ``cartoon'' of Barwin (PDB ID 1BW3).
   Barwin, an endoglucanase, has eight $\beta$-strands forming a closed 
   ``barrel''
   shape, as well as four $\alpha$-helices.}
   \label{barwin}
 \end{center}
\end{figure}

The \textit{tertiary structure} of a protein is the fully-specified
three\--dimensional position of every atom.
The orientation of the backbone atoms in three-dimensional space forms
three distinguishing dihedral angles:
$\phi$ between the carbon-1-nitrogen and $\alpha$-carbon-carbon-1
atoms in an amino acid, $\psi$ between the nitrogen-$\alpha$-carbon and carbon-1-nitrogen atoms,
and $\omega$ between the $\alpha$-carbon-carbon-1 and the nitrogen and 
$\alpha$-carbon of the next amino acid (see
Figure~\ref{angles}).
The angle $\omega$ is usually $0^\circ$, and occasionally $180^\circ$.
The
side chain of each amino acid must then pack into a low-energy state
in such a way that it does not interfere with the other amino acids in the
protein.
The tertiary structure represents, in most cases, a global minimum energy state,
also known as the \textit{native state}.
Many proteins have now had their tertiary structure determined by X-ray 
crystallography, or by nuclear magnetic resonance (NMR) spectroscopy.
A protein whose structure has been determined experimentally is said to have a
\emph{solved} structure.
However, the difficulty of experimentally solving the structure of any 
particular protein of interest can vary.
X-ray crystallography's limiting factor is that not all proteins can be
put into solution and crystallized, while NMR's limiting factor is primarily 
computational.
The Protein Data Bank (PDB)~\cite{Bernstein:1977un, Berman:2000hl} is a publicly available 
database that contains the atomic coordinates of all proteins whose tertiary 
structure has been solved.

\begin{figure}[ht!]
\begin{center}
  \fbox{\includegraphics[width=5in] {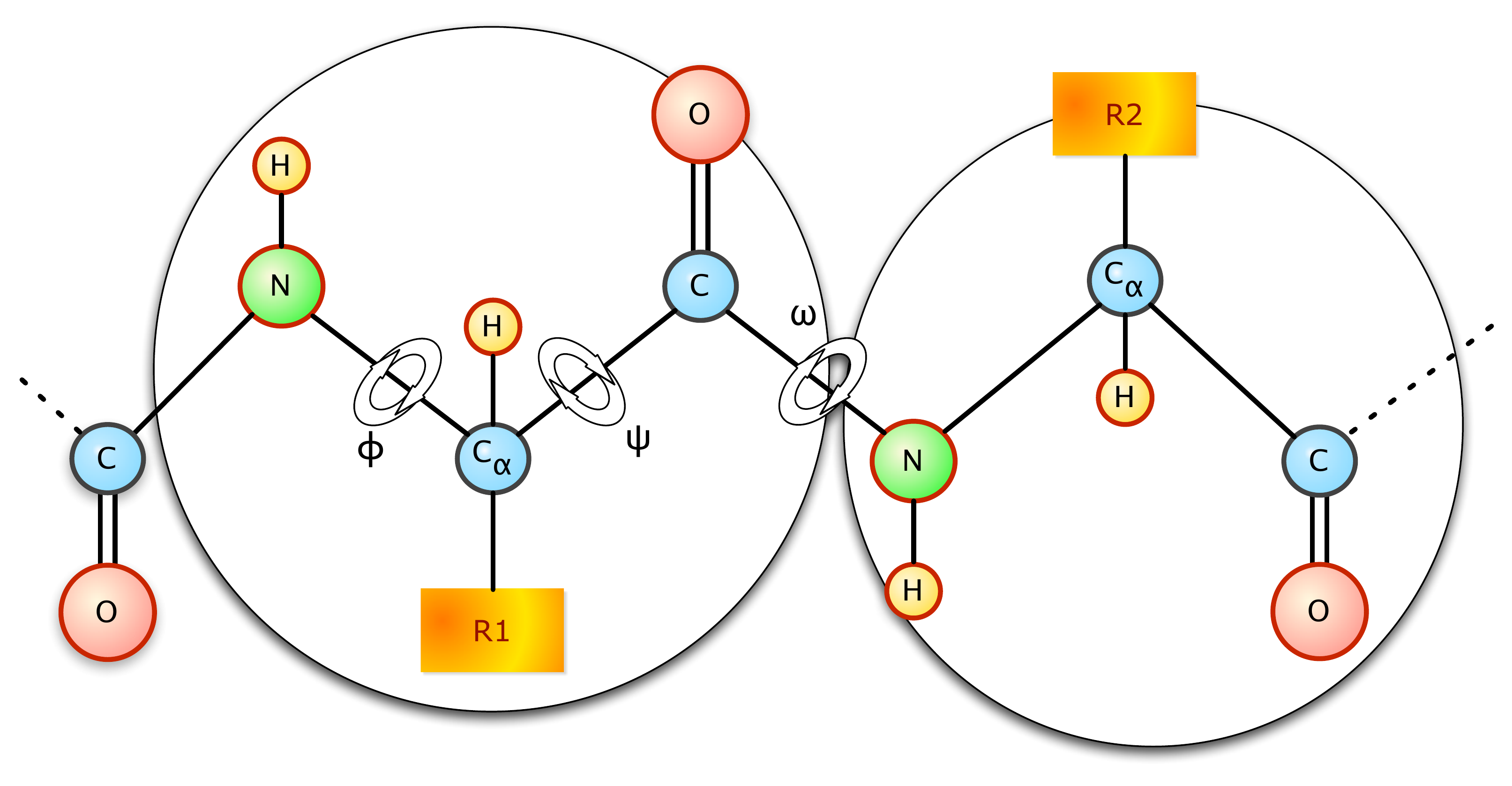}}
   \caption[Backbone angles]{The orientation of the backbone atoms in 
   three-dimensional space 
   forms three dihedral angles:
     $\phi$ between the carbon-1-nitrogen and $\alpha$-carbon-carbon-1
    atoms, $\psi$ between the nitrogen-$\alpha$-carbon and carbon-1-nitrogen 
    atoms,
    and $\omega$ between the $\alpha$-carbon-carbon-1 and the nitrogen and 
    $\alpha$-carbon of the next amino acid.}
   \label{angles}
 \end{center}
\end{figure}

Finally, \textit{quaternary structure} describes how multiple tertiary 
structures interact; these may be multiple duplicate protein chains (for 
example, a homodimer is a complex of two identical protein chains, while a 
heterodimer is a complex comprising two different protein chains).
In this work, we focus on individual chains, rather than quaternary structures.

\subsection{Protein Data Sets}

In order to make sense of the evolutionary, structural, and functional 
relationships among proteins, biologists have created several organizational
schemes.
Structural Classification Of Proteins 
(SCOP)~\cite{Murzin:1995uh, Andreeva:2004ic}
and CATH (which stands for Class, Architecture, Topology,
and Homologous superfamily)~\cite{Orengo:1997vy, Pearl:2003wb, Greene:2007iu}
are hierarchical schemes that
place proteins in a tree based primarily on structural, but also on evolutionary
and functional similarities.
In this work, we will primarily rely on SCOP, since it has been used in many
homology detection 
studies~\cite{Elofsson:1999tj, Wistrand:2004ia, Soding:2005ff}.

SCOP organizes all protein sequences of known structure (with some time delay)
into a four-level hierarchy.
The top level of the SCOP hierarchy is \emph{class}, which distinguishes the
primary secondary-structural composition of proteins: mainly-$\alpha$,
mainly-$\beta$, mixed $\alpha$ and $\beta$, cellular-membrane proteins, among
others.
The second level of the SCOP hierarchy is \emph{fold}, which organizes proteins
by overall structural motif, or supersecondary structure.
Proteins in the same fold are not necessarily evolutionarily related.
Below fold is \emph{superfamily}, which organizes proteins that share
evolutionary relationships, as well as similar structure and function.
Below the superfamily level is the \emph{family} level of SCOP.
Proteins in the same family have clear evolutionary relationships, and a
significant level of sequence similarity.
Figure~\ref{scop-hierarchy} illustrates the SCOP hierarchy.

\begin{figure}[htb!]
\advance\leftskip-0.3in
\begin{center}
  \fbox{\includegraphics[width=5in] {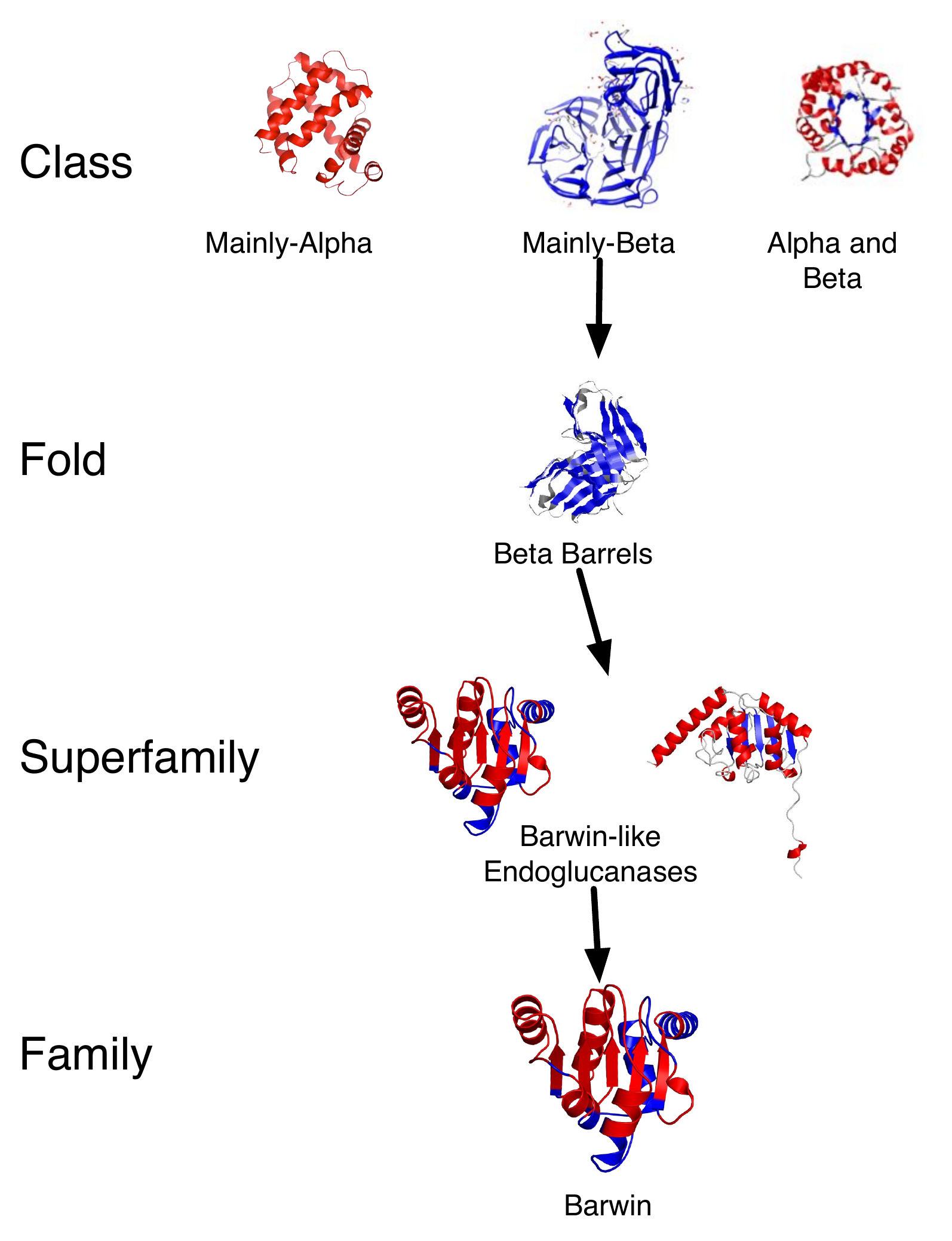}}
   \caption[The SCOP hierarchy of protein structure.]{The SCOP hierarchy of 
   protein structure.
   \emph{Class} organizes proteins in large part according to supersecondary
   structural content.
   \emph{Fold} organizes proteins by supersecondary structural motifs.
   \emph{Superfamily} organizes proteins by structural, functional, and
   evolutionary similarity, while \emph{Family} organizes them by sequence
   similarity as well.
   }
   \label{scop-hierarchy}
 \end{center}
\end{figure}

\subsection{Protein Folding}

The process by which a protein, as its amino acid sequence is emitted by the
ribosome, forms its stable tertiary structure is called \textit{folding}.
In 1969, molecular biologist Cyrus Levinthal noted~\cite{Levinthal:1969vj} that 
even given a coarse
(tripartite) discretization of bond angles, a protein of merely 100 amino acids
(a short chain by most standards) could take $3^{300}$ distinct 
three-dimensional conformations (tertiary structures).
Given the accepted view that proteins
typically fold to globally minimum energy states, Levinthal noted that a
protein would take the lifetime of the universe to find a minimum energy state
by sampling the entire fold space.
However, in practice, proteins fold in
microseconds or milliseconds.
This apparent paradox became known as Levinthal's
Paradox; the solution to the paradox must be that nature does not explore the
entire fold space.

The rapidity of protein folding is thought to be because proteins fold along
\textit{folding funnels}, which prune much of the 
possible fold space very quickly~\cite{Dill:1997vn, Matagne:1998uu, Tsai:1999we, Dinner:2000vu}. 
It is even conjectured~\cite{Rost:2002wx,Cuff:2009cz} that only those proteins that exhibit fold 
funnels that allow them to fold quickly have evolved; protein sequences that
would not quickly find stable native states would be selected against during
the course of evolution.

Physics-based approaches to computationally solving the protein folding problem 
try to solve the various force field equations 
(including hydrophobic, electrostatic, and van der Waals forces) to find a
minimum-energy state~\cite{BornbergBauer:1999wu, Heun:1999wr}.
However, even the most simplified models can prove computationally intractable.
In 1998, Berger and Leighton~\cite{Berger:1998vv} proved that the seemingly simple
HP (hydrophobic-hydrophilic) lattice model of protein folding is NP-hard.

For some purposes, such as understanding the molecular motion of proteins such
as ion channels (which control the flow of ions through a cellular membrane) or
flagellin (which forms the moving filament in bacterial flagella), just knowing 
the native state is not enough, and molecular dynamics simulations are 
necessary.

The current state of the art in full tertiary structure prediction via
molecular dynamics modeling relies on huge computational infrastructures; we 
describe two of them.
The first, Anton, is a
supercomputer purpose-built for protein simulations by the D.E.
Shaw Research~\cite{Shaw:2007cr}.
The second, Folding@home, is a worldwide distributed-computing system developed
at Stanford~\cite{Jayachandran:2006ur}; 
it uses spare CPU and GPU cycles on desktop, laptop, and video
game systems around the world.
Both of these systems can compute a few
milliseconds of simulation time per day.

Fortunately, however, it is not always necessary to determine tertiary 
structure to the level of precision achieved by experimental methods.
Computational biology methods that use statistical energy functions and
secondary or supersecondary structure prediction have made significant progress
in the last ten years~\cite{Moult:2006tn}.
In particular, \emph{approximately} predicting the tertiary structure--or
predicting the supersecondary structure--may be adequate when the end goal
is function prediction or homology detection.

\section{Protein Homology}

An alternative to experimentally predicting the structure of a protein is to try
to determine, based on \emph{sequence} similarity, that a protein of interest
is sufficiently closely related, in evolutionary terms, to some other protein
of solved structure that it is likely to fold into a similar shape.
However, while the task gets easier the closer it becomes to that of 
determining
sequence, the quality of the results worsens; protein sequence is less well
\textit{conserved} than structure; that is, fairly significantly different
protein sequences may nonetheless share quite similar structures and functions~\cite{DunbrackJr:2006dx}.

Biologists say that two proteins are \emph{homologous} when they are derived
from a common ancestor.
Often, homologous proteins share common structure.
When two protein sequences are similar, it is relatively easy to determine that
they are homologous.
However, homologous proteins may differ significantly in terms of sequence
identity.
Sequence analysis methods have long allowed for the detection of homologous
proteins, provided sequence divergence is not too great.
The problem of detecting homologous proteins when sequence similarity is low is
known as \emph{remote} homolog detection.
The purpose of this thesis is to develop novel methods for remote
homology detection.
Now, we will survey existing methods for homology detection.

\subsubsection{BLAST}

Altschul, et al. developed the Basic Local Alignment Search Tool (BLAST)~\cite{Altschul:1990dw}
algorithm as a faster alternative to dynamic programming-based methods such
as the Smith-Waterman~\cite{Smith:1981up} algorithm.
BLAST uses a number of heuristics to reduce the time required to perform an
alignment, at the possible expense of some accuracy.
BLAST also relies on an indexed database of sequences to be searched.

BLAST allows for fast search through databases to find potential homologs.
The protein-specific version of BLAST is called BLASTP.
BLASTP uses a \emph{substitution matrix} to score alignments; the most commonly
used substitution matrix is BLOcks of amino acid SUbstitution Matrix
(BLOSUM)~\cite{Henikoff:1992tk}.
A BLOSUM score $s(i,j)$ for two residues $i$ and $j$ is given by:
\begin{equation}
  s(i,j) = \frac{log \frac{P_{i,j}}{f_{i}f_{j}}}{\lambda}
\end{equation}
where $P_{i,j}$ is the probability of observing residues $i$ and $j$ aligned
in homologous sequences, and $f_{i}$ is the observed background frequency of
residue $i$, and $\lambda$ is a scaling factor chosen to produce integer
values for the scores~\cite{Henikoff:1992tk}.

Different variants of the BLOSUM matrices exist; for a chosen threshold $L$,
only sequences within a sequence identity threshold of $L\%$ are clustered into
a single representative sequence; those sequences are then aligned and the
alignment used to compute the BLOSUM$L$ matrix.
Thus, BLOSUM80 is intended for use in less divergent sequence alignments,
while BLOSUM50 is intended for use in more divergent sequence alignments.
BLOSUM62 is a commonly used default scoring matrix for protein sequence
alignment tools such as BLAST~\cite{Altschul:1997tl}.

There are newer BLASTP variants,
such as PSI-BLAST~\cite{Altschul:1997tl} and DELTA-BLAST~\cite{Boratyn:2012er} 
that improve
sensitivity by replacing BLOSUM with a \emph{Position-specific scoring matrix}
(PSSM) that scores mismatches differently depending on where they occur in the
alignment. 
PSI-BLAST determines its PSSM by iterative search: First, it performs a
standard BLASTP search, and computes a PSSM from the resulting alignment.
It then repeats this process, searching with the PSSM created by the previous
iteration, and computing a new PSSM.
In contrast, DELTA-BLAST uses pre-determined PSSMs derived from the Conserved
Domains Database (CDD)~\cite{MarchlerBauer:2005uv}, essentially groups of proteins already
determined to be homologous.

BLAST and its derivatives, such as PSI-BLAST and DELTA-BLAST, are effective at
identifying homologous protein sequences for a query sequence when those
homologous sequences share a reasonable amount of sequence identity with the
query sequence~\cite{Rost:1999taa}.
However, we wish to be able to identify homologous proteins--those that share
structural, functional, and evolutionary relationships--even when they do not
share a great deal of sequence similarity.
Since protein structure is more highly conserved than 
sequence~\cite{Dalal:1997wl}, we would like to incorporate
information that is not simply derived from sequence alignments.

\subsection{Structural Alignment}

Just as we can align the sequences of two or more proteins in order to compare
them, we can also align the \emph{structures} of two or more proteins.
Clearly, protein structure alignment requires knowing the tertiary 
structure--the three dimensional coordinates of all the atoms, or at very least
the backbone atoms--of the proteins to be aligned.

Structural alignment can be used to measure structural similarity, and from 
there infer functional and evolutionary relationships.
Structural alignment can also be used to measure the quality of a computational
protein structure \emph{prediction} versus a known, solved structure.
In general, protein structural alignment relies on some form of geometric
superposition, though a wide variety of algorithms exist for efficiently
computing this superposition.
In fact, computing the \emph{optimal} geometric superposition is known to be
NP-hard~\cite{WANG:1994jq}.
Several heuristic approaches have been developed for practical protein
structural alignment.
DALI~\cite{Shindyalov:1998wn}, for example, breaks the structures into hexapeptide fragments and 
calculates a distance matrix by evaluating the contact patterns between
fragments.
DALI then compares these distance matrices and applies a score-maximization
search to compute an alignment.
This approach is called ``aligned fragment pair'' alignment, and is also used by
MAMMOTH~\cite{Ortiz:2009wx}, which instead applies dynamic programming to 
compute the alignment.
Matt~\cite{Menke:2008wu} also uses aligned fragment pairs, but allows ``impossible''
translations and twists in order to better capture structural similarities at
the superfamily or even fold levels.
Hybrid aligners, which use both sequence and structure information, also exist.
DeepAlign~\cite{Wang:2012wq} incorporates not just atomic coordinates but also
secondary structural annotation and sequence information.
Our own Formatt~\cite{Daniels:2012jr} also combines sequence and structure 
information, to try to avoid ``register errors'' that trade significant 
sequence alignment errors for small structural gains.
Most of these methods are fundamentally solving a bi-criterion optimization
problem: we wish to align as much of the input proteins' structure as possible,
while at the same time minimizing the root mean square distance (RMSD) of the
resulting alignment, where RMSD is defined as the square root of the average 
distance between corresponding $\alpha$-carbon atoms between the backbones of 
the proteins in alignment~\cite{Kabsch:1976tp}.

When aligning more distantly-related proteins, structural alignment methods
often outperform purely sequence-based methods~\cite{Chothia:1986tm}.
For this reason, protein structural alignment is routinely used to produce
alignments of homologous proteins to form training sets for remote homology
detection techniques.

\section{Hidden Markov Models}

A Markov model represents a series of observations via a probabilistic 
finite-state automaton.
A Markov model on an alphabet $A$ is a triplet

\begin{equation}
  M=(Q,\pi,\alpha),
\end{equation}  
   where $Q$ is
a finite set of states, each state generates a character from $A$, $\pi$ is the
set of initial state probabilities, and $\alpha$ is the set of state transition
probabilities.
Markov models are so named because they uphold the \emph{Markov property},
which states that future states depend only on the current state of the system.
In other words, first-order Markov models are ``memoryless.''
A Markov model may take into account a \emph{fixed} number of past states (a 
$k_{th}$-order Markov model make take into account $k$ past states).

A hidden Markov model (HMM) represents a series (sometimes a time series) of
observations by a ``hidden'' stochastic process.
HMMs were originally developed for speech recognition~\cite{Viterbi:1967hq, Rabiner:1989vx}.
A HMM is on an alphabet $A$ is a 5-tuple
\begin{equation}
 M=(Q,V,\pi,\alpha,\beta),
 \end{equation}
  where $Q$ is again a finite set
of states, $V$ is a finite set of observations per state, $\pi$ is the set of
initial state probabilities, $\alpha$ is again the set of state transition
probabilities, and $\beta$ is the finite set of emission probabilities over the
alphabet $A$.
Hidden Markov models differ from ordinary Markov models in that,
while the emissions are observable, the states occupied by the finite state 
machine are not themselves observable.
There are three problems to be solved regarding HMMs, and correspondingly, three
algorithms to solve them.

The first problem is: \emph{Given an HMM $M$ and an observed sequence $S$, 
compute the most probable path through $M$ that generates $S$.}

This problem is solved by the \emph{Viterbi algorithm}~\cite{Viterbi:1967hq}, which
is typically implemented using dynamic programming.
The Viterbi algorithm solves the recurrence relation:
\begin{equation}
  \begin{array}{rcl}
  V_{1,k} &=& \mathrm{P}\big( s_1 \ | \ k \big) \cdot \pi_k \\
  V_{t,k} &=& \mathrm{P}\big( s_t \ | \ k \big) \cdot \max_{x \in Q} \left( a_{x,k} \cdot V_{t-1,x}\right)
  \end{array}
\end{equation}

where $V_{t,k}$ is the probability of the most probable state sequence emitting 
the first $t$ observations with $k$ as its final state, and $s_i \in S$ is the
$i^{th}$ observation in $S$. 
The corresponding path through the model can be retrieved by remembering what
series of transitions among states $x \in Q$ were chosen when solving the
recurrence relation.

The second problem is: \emph{Given an HMM $M$ and a sequence of observations 
$S$, compute $P(S|M)$, the probability of observing the sequence $S$ emitted by 
the model $M$.}

This problem is solved by the \emph{forward algorithm}, which relies on dynamic 
programming as well.
In essence, the forward algorithm sums the probabilities over all possible state
paths that can emit $S$.
The recurrence relation for the forward algorithm is nearly identical to that
for the Viterbi algorithm, except that it sums, rather than choosing the
maximum from, the probabilities at each step:

\begin{equation}
  \begin{array}{rcl}
  V_{1,k} &=& \mathrm{P}\big( s_1 \ | \ k \big) \cdot \pi_k \\
  V_{t,k} &=& \mathrm{P}\big( s_t \ | \ k \big) \cdot \sum\limits_{x \in Q} \left( a_{x,k} \cdot V_{t-1,x}\right)
  \end{array}
\end{equation}

The third problem is: \emph{Given a set of sequences of observations, $O$,
and a model $M$, determine the transition probabilities $\alpha$ and emission
probabilities $\beta$ that maximize the $P(O|M)$, the likelihood of observing
the set of sequences given the model.}

Typically, a solution to this problem is \emph{estimated} by the 
\emph{Baum-Welch algorithm}~\cite{Baum:1970tv}, which is
an expectation-maximization algorithm.
A more computationally efficient but less accurate alternative is the Viterbi
Training algorithm (not to be confused with the Viterbi algorithm), also known
as \emph{segmental k-means}~\cite{Rabiner:1989vx}.
A simulated annealing search approach to Baum-Welch can also be used to avoid
local optima~\cite{Baldi:1994vc}.

A further explanation of the above algorithms can be found in \cite{Rabiner:1989vx}.

Despite their origins in the field of speech recognition, hidden Markov models 
have been
used in a variety of areas within the realm of computational biology.
In the context of DNA sequence analysis, HMMs have been 
used~\cite{Dasgupta:2002wj} to detect ``CpG islands,'' regions of the genome where 
cytosine and guanine are predominant and adjacent in sequence.
CpG islands are useful for determining the start of transcription 
sequences--the markers that indicate the regions of the genome that code for
protein sequences.
Hidden Markov models were first used to search for DNA sequences in genome
databases by Churchill~\cite{Churchill:1989wj} in the late 1980s.
Later, Krogh et al.~\cite{Krogh:1994hv} used HMMs to model protein evolution.

\subsection{Profile Hidden Markov Models}

With respect to homology detection, \emph{profile} hidden Markov models have
been popular.
In particular, profile HMMs have been used to model families of protein
sequences, in order to predict whether newly-discovered sequences belong to
those families.
Profile hidden Markov models attempt to represent the evolutionary processes
underlying the differences among closely-related proteins.
In addition, HMM-derived clusterings of proteins have been published, such as
Pfam~\cite{Finn:2006ez}, PROSITE~\cite{Hulo:2006du}, and 
SUPERFAMILY~\cite{Wilson:2007cm}.

HMMER~\cite{Eddy:1998ut} and SAM~\cite{Hughey:1996ub} are two popular software 
tools for
homology detection in proteins (though both are also widely used in
nucleotide sequence analysis, as well).
Much of the work in this dissertation is based on HMMER; we chose it as it is
open-source and more actively maintained.

\newcommand\txprobj[3][]{a#1_{{#2}_{j-1}{#3}_j}}
\newcommand\txprobjj[3][]{a#1_{{#2}_{j-1}{#3}_j}}
\newcommand\alignwidth{\ensuremath C} 
\newcommand\pairedwith[1]{{\pi(#1)}}

HMMER models three types of events that may occur during the evolution of a
protein: \emph{insertion}, \emph{deletion}, and \emph{substitution} of an amino
acid at a particular position.
These three possible events become the three hidden states of the HMM.
Substitution events are modeled using a \emph{match} state, which also 
represents amino acids that are conserved, or have not changed, between 
proteins.
In essence, mutated amino acids can be represented as substitutions
using a substitution matrix, and since the most probable substitution in such a
matrix is the identity function, conserved amino acids can also be represented
using the same matrix.
Insertion and match states are both considered \emph{emission} states, as each
corresponds to the presence of an amino acid at a particular position in a 
protein.
Each emission state comprises a table of emission probabilities: the likelihood
that any particular amino acid will be present (emitted) at that position.
Intuitively, for each match state, the most common amino acid seen in the
training data will be the most probable amino acid in the emission table for
that column of the alignment.

HMMER uses the ``Plan7'' hidden Markov model architecture, which forbids direct
transitions between insertion states and deletion states~\cite{Eddy:1998ut}.
``Plan7'' is a pun on ``Plan9,'' the architecture by Krogh, et
al.~\cite{Krogh:1994hv} that allowed all 9 possible transitions among match,
insert, and delete states;
``Plan7'' gets its name because there are exactly 7 possible transitions into
the states of any column of the alignment used for training.
See Figure~\ref{plan7} for an illustration of the Plan7 architecture.

\begin{figure}[htb!]
\advance\leftskip-0.3in
\begin{center}
  \fbox{\includegraphics[width=4in] {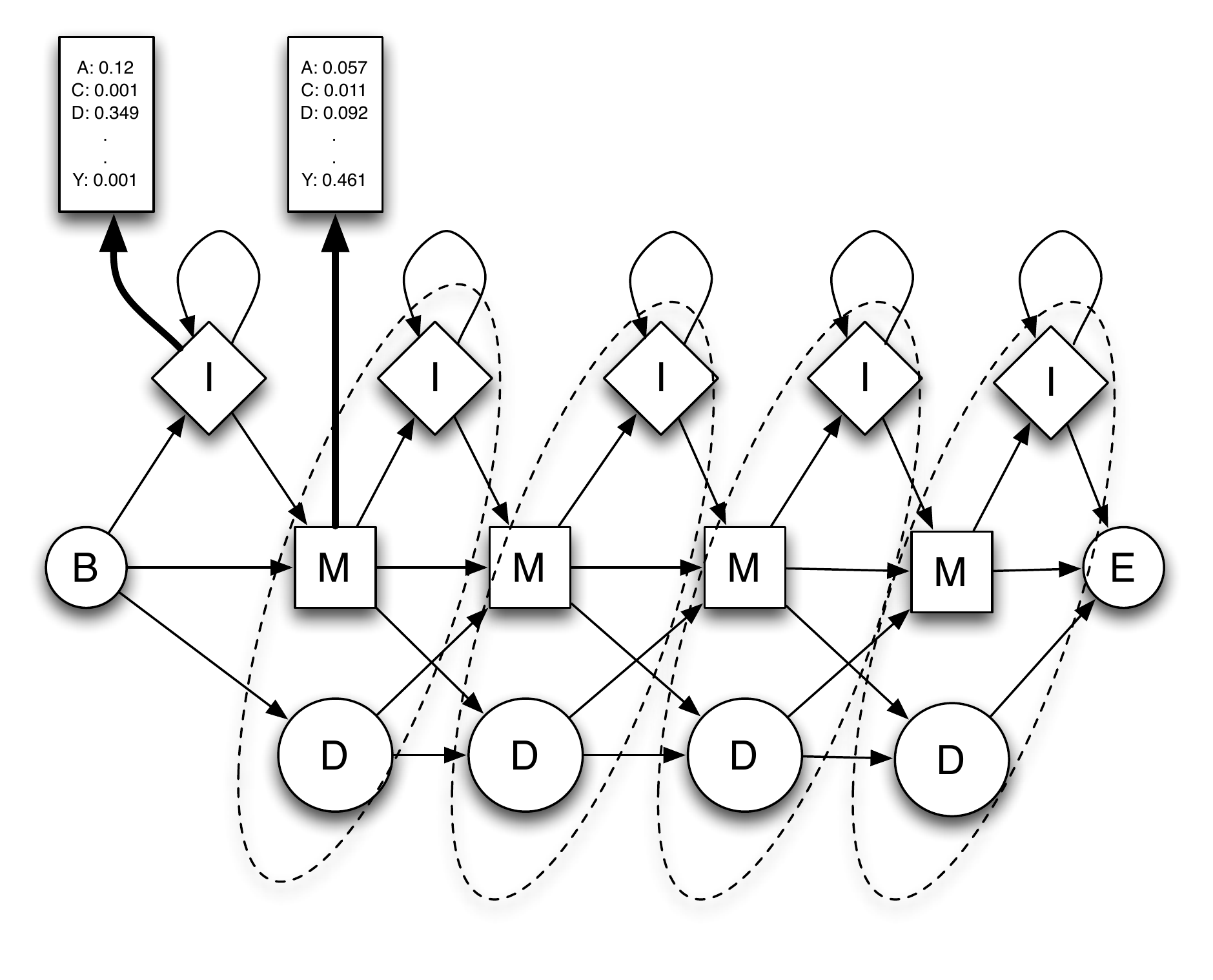}}
   \caption[The ``Plan7'' architecture for hidden Markov models, as implemented
   in HMMER.]{The ``Plan7'' architecture for hidden Markov models, as 
   implemented in HMMER.
   Dashed circles indicate \emph{nodes} of the model.
   A node groups a match, insertion, and deletion state, along with the
   emission probabilities for the match and insertion states.
   Note: this diagram simplifies the ``Plan7'' architecture; in reality, begin
   and end nodes are more complex, allowing for entire models to repeat.
   }
   \label{plan7}
 \end{center}
\end{figure}

HMMER trains a profile HMM (using a simulated annealing variant of the 
Baum-Welch algorithm)~\cite{McClure:1996uv} on a \emph{sequence 
profile}, which is an alignment of the protein sequences comprising some 
group--such as a SCOP superfamily or family--of putatively homologous proteins.
This alignment may be a sequence alignment or a structural alignment; in this
work we will focus on profiles derived from structural alignments.

An alignment used for training may of course contain \emph{gaps}.
A~gap in row~2, column~$j$ indicates that as proteins evolved, either 
protein~2 lost its amino acid in position~$j$, or 
other proteins gained an amino acid in position~$j$.
If~column~$j$ contains few gaps, 
it~is considered a \emph{consensus column},
and the few proteins with gaps may have lost amino acids via
\emph{deletions}.
Note that this model is directionless with respect to evolutionary change;
it does not distinguish between a residue being gained or lost over time.
If~column~$j$ contains \emph{mostly} gaps, 
it~is considered a \emph{non-consensus column},
and the few proteins without gaps may have gained amino acids via
\emph{insertions}. 

We refer to the amino acid sequence of a protein whose structure we do not know,
and wish to determine using homology detection, as a \emph{query sequence}.
Homology detection using a hidden Markov model involves \emph{aligning} a query
sequence to a hidden Markov model, or computing a \emph{path} through the model
that maximizes the likelihood of the model emitting the query sequence.
This alignment involves assigning successive amino acids in the query sequence
to successive nodes of the model. For a given node of the model, the match
and deletion states are mutually exclusive, as are the insertion and deletion
states. 
However, it is permissible for a path to assign amino acids to both the match 
and insert states of a node.
In addition, the match state consumes exactly one amino acid from the query
sequence, while the insert state may consume many.
The delete state consumes no amino acids from the query sequence.

Given a hidden Markov model, a protein whose query sequence has a higher 
probability is considered to be more likely to 
be homologous to the proteins in the alignment.
We~write a query sequence as $x_1, \ldots, x_{\scriptscriptstyle N}$,
where each $x_i$~is  an amino acid.
The number of amino acids, $N$, can differ from the number of columns
in the alignment,~\alignwidth. 

A~hidden Markov model carries emission probabilities on some states, and
transition probabilities on all edges between states.
Both the probabilities and the states are determined by the alignment:
\begin{itemize}
\item
For each column~$j$ of the alignment, the hidden Markov model has a
\emph{match state}~$M_j$.
The match state contains a table $e_{M_j}(x)$ which gives the
 probability that a homologous protein has amino acid~$x$ in
 column~$j$.
\item 
For each column~$j$ of the alignment, the hidden Markov model has an
\emph{insertion state}~$I_j$.
The insertion state contains a table $e_{I_j}(x)$ which represents the
probability that a homologous protein has gained amino acid~$x$ by
insertion at column~$j$.
\item
For each column~$j$ of the alignment, the hidden Markov model has a
\emph{deletion state}~$D_j$.
The deletion state represents the probability that a homologous protein
has lost an amino acid by deletion from column~$j$.
\end{itemize}
The probabilities $e_{M_j}(x)$ and $e_{I_j}(x)$ are \emph{emission probabilities}.
Each tuple of match, insertion, and deletion states is called a \emph{node}
of the hidden Markov model.

Each transition has its own probability:
\begin{itemize} 
\item
A~transition into a match state 
is more likely when column~$j$ is a consensus column.
Depending on the predecessor state, 
the probability of such a transition is 
$\txprobj M M$, $\txprobj I M$, or~$\txprobj D M$.
\item
A~transition into a deletion state 
is more likely when column~$j$ is a non-consensus column.
The probability of such a transition is 
$\txprobj M D$~or~$\txprobj D D$.
\item
A~transition into an insertion state 
is more likely when column~$j$ is a non-consensus column.
The probability of such a transition is 
$\txprobjj M I$~or~$\txprobjj I I$.
\end{itemize}

Due to the specific topology of the state-transition graph in the ``Plan7''
architecture, a reformulation of the Viterbi recurrence relations are warranted.
In particular, we need not consider all state transitions that would be possible
given a general topology, and instead, need consider only three possible
transitions at each node, which reduces the search space.
The variant of the Viterbi algorithm adapted for the ``Plan7'' architecture is
given by the recurrence relations:

\newcommand\vsum[2]{#2&{}+{}& #1}

\def\goo{18pt}
\def\gum{14pt}

\def\maxiquad{\hskip 1.2em\relax}
\begin{equation}
\begin{array}{@{}l@{}c@{}l}

V_{j}^{M}(i) &{}={}& \frac{e_{M_{j}}(x_{i})}{q_{x_{i}}} \times \max \left\{
  \begin{array}{l@{}c@{}l}
  V_{j-1}^{M}(i - 1) \times a_{M_{j-1}M_{j}}\\
  V_{j-1}^{I}(i - 1) \times a_{I_{j-1}M_{j}}\\
  V_{j-1}^{D}(i - 1) \times a_{D_{j-1}M_{j}}\\
  \end{array} \right.\\\\[\goo]
V_{j}^{I}(i) &=& \log\frac{e_{I_{j}}(x_{i})}{q_{x_{i}}} \times \max \left\{
  \begin{array}{l@{}c@{}l}
  V_{j}^{M}(i - 1) \times a_{M_{j}I_{j}}\\
  V_{j}^{I}(i - 1) \times a_{I_{j}I_{j}}\\
  \end{array} \right.\\\\[\gum]
V_{j}^{D}(i) &=& \max \left\{
  \begin{array}{l@{}c@{}l}
  V_{j-1}^{M}(i) \times a_{M_{j-1}D_{j}}\\
  V_{j-1}^{D}(i) \times a_{D_{j-1}D_{j}}\\
  \end{array} \right.\\

\end{array}
\end{equation}


\section{Other Homology Detection Methods}

\subsection{Threading Methods}

Threading is a methodology by which a query sequence is \emph{threaded} onto
a structural template, and the quality of the threading is evaluated by means
of an energy function or a statistical likelihood.
The idea of threading is based on the observation that the number
of unique protein folds found in nature is small with respect to the number of
distinct protein sequences, and that relatively few novel protein folds have
been found recently~\cite{Pearl:2003wb}.

THREADER~\cite{Jones:1992wx}, the original threading approach, aligns a protein 
sequence to a full tertiary structure model of a protein, and computes a score 
based upon a Boltzmann energy function and solvent potentials.
THREADER thus evaluates the propensity a sequence has for forming a particular
tertiary structure, but it cannot distinguish homologs (evolutionarily related
proteins that share structure and possibly function) from \emph{analogs} 
(proteins that happen to share similar structure but have no evolutionary 
relationship)~\cite{Orengo:1994fh, Jones:1997ve}.

GenTHREADER~\cite{Jones:1999im} improves upon THREADER, using an artificial 
neural network to compute a score
based upon multiple inputs: solvent and Boltzmann potentials like THREADER, but
also a sequence alignment score and length, and the lengths of the query
sequence and template.

Another popular and successful threading method is RAPTOR~\cite{Xu:2003p3417}, 
which relies on a template based on a 
\emph{contact map} to indicate which residues in a protein are in close
geometric proximity to one another, as well as the statistical propensities for
individual residues to be in such proximity.
RAPTOR then relies on linear programming to compute the optimal alignment
of a query sequence to this template, in order to minimize the statistical 
energy.
In Chapter~\ref{chapter:c3_smurflite}, we compare our results for remote
homology detection to those of RAPTOR.

Other threading methods include SPARKS X~\cite{Yang:2011id} and the
recently-developed RaptorX~\cite{Peng:2011wx}.

\subsection{Profile-Profile Hidden Markov Models}

Several recent efforts have improved upon profile hidden Markov models, by
aligning a profile HMM built from a training profile (much like HMMER) with
another profile HMM built from the query sequence.
HHPred~\cite{Soding:2005ff}, MUSTER~\cite{Wu:2008vh}, and HHblits~\cite{Remmert:2012cj} are three such
approaches.
Given a query sequence, HHPred relies on PSI-BLAST~\cite{Altschul:1997tl} to build a sequence
profile.
HHPred then builds a profile HMM on this profile, and uses a variant of the
Viterbi algorithm to align the \emph{query} HMM to candidate \emph{target} HMMs.
HHPred, along with other profile-profile hidden Markov model methods, relies on
the query profile to more faithfully represent the evolutionary variation in the
protein sequences that may be homologous to the query sequence.
In Chapter~\ref{chapter:c3_smurflite}, we compare our results for remote 
homology detection to those of HHPred.

\subsection{Markov random fields}

Some researchers have suggested generalizing HMMs to the more powerful
Markov random fields (MRFs).
Unlike HMMs, which model only local dependencies among neighboring residues,
MRFs can capture nonlocal interactions, such as the conservation of
hydrogen-bonded residues in paired $\beta$-strands.
SMURF (Structural Motifs Using
Random Fields)~\cite{Menke:2009, Menke:2010ti} used this $\beta$-strand 
information to recognize remote homologs
in the $\beta$-propeller folds better than HMM methods.
However, SMURF is limited by computational complexity,
because it uses multidimensional dynamic programming to
compute an optimal parse of a query sequence onto the MRF, and its
computational complexity is exponential in something called the interleave
number of a structure.
This interleave number is simply the number of
intervening $\beta$-strands (in sequence) between a pair of hydrogen-bonded,
paired $\beta$-strands.
$\beta$-propellers have a maximum interleave number of three,
and thus they are tractable for SMURF.
In contrast, some $\beta$-barrels and
sandwiches have an interleave number as high as 12, and thus, SMURF's
computational complexity becomes intractable on available computer systems.
Chapters~\ref{chapter:c3_smurflite}~and~\ref{chapter:c4_mrfy} explore two
alternative approaches to mitigate this
computational hindrance.

\section{Remote Homology Detection} 

While computing tertiary structure is computationally challenging, we can take
comfort in the fact that we do not always need tertiary structure to make useful
predictions as to function or evolutionary similarity.
In particular, supersecondary structure is often enough.
Since structure determines function, if we can classify a new protein of
unknown structure as sharing similar supersecondary structure to a group of
known proteins, we have evidence that the new protein shares a similar function
to those known proteins.

Protein sequence is much typically less conserved than
structure~\cite{Kolodny:2006jv, Wilson:2000ed}, so proteins 
of similar
structure and function, as seen at the SCOP superfamily level, may lack any
meaningful sequence similarity.
Simple sequence comparisons such as BLAST fail
to correctly identify these remote homologs.
However, if biologists sequence
the genome of an organism, they will wish to functionally annotate the proteins
coded for by its genes.
Anton or Folding@Home would require weeks or months of
computational time per gene to compute tertiary structure, and even a bacterium
has thousands of genes.
Supersecondary structure provides enough information to
make reasonable functional annotations, and we can compute it quickly enough to
scale to entire genomes.
Even threading approaches such as RAPTOR\cite{Xu:2003p3417} may require hours per
gene.
The methods we have developed are faster and more accurate than standard
threading approaches.

Threading methods such as RAPTOR~\cite{Xu:2003p3417} attempt to map a new protein sequence
onto templates built from individual solved proteins.
While a high-quality
threading hit may produce an accurate tertiary structure, this approach loses
the ability to take a larger evolutionary view of protein space.
Profile-based
methods--including SMURF--build knowledge about evolutionarily conserved
parts of the protein structure and sequence into their templates.
Rather than
matching a new protein sequence to a single best-fitting structure, we wish to
say that a new sequence belongs to a group of proteins that share evolutionary,
structural, and possibly functional similarities.

In order to predict that a new protein sequence shares structure and function
with a group of known proteins, we must be sure that these groups of proteins
are consistent.
In particular, since we use structure to infer function, we
wish to ensure that protein space is organized in a structurally consistent way.

\section{Outline of This Work}

In this dissertation, we present several approaches to remedying the complexity
of MRFs for remote homology detection, as well as an approach to improving the
quality of training data for remote homology detection.
Below is the outline of individual chapters in this dissertation.

We begin with a tour of
protein fold space (Chapter~\ref{chapter:c2_touring}), examining the structural
consistency of the SCOP protein structural hierarchy.
We also introduce a method for clustering protein structures such that 
manually-curated hierarchies such as SCOP~\cite{Murzin:1995uh} can be recreated with 
reasonable accuracy, based purely on automated structural alignments.
We also introduce a benchmark set, called MattBench, that we propose for use by
the developers of protein sequence or structural aligners.

In Chapter~\ref{chapter:c3_smurflite}, we discuss an approach to generalizing
Markov random fields to
the problem of remote homology detection in $\beta$-structural proteins.
We simplify the SMURF~\cite{Menke:2009, Menke:2010ti} 
Markov random field model by limiting the complexity of the
dependency graph, in order to bound the computational complexity of finding
an optimal parse of a query sequence to a model.
We combine these simplified Markov random fields with a model of ``simulated 
evolution'' to improve upon existing methods.

In Chapter~\ref{chapter:c4_mrfy}, we introduce an approach for remote homology
detection using the SMURF Markov random fields that does not require 
simplifying the dependency graph.
Instead, we introduce a stochastic search approach that quickly computes
approximate alignments to the Markov random field, and which should be 
generalizable to all protein folds.

Finally, in Chapter 5, we discuss the results and summarize the key findings of 
this dissertation followed by possible directions for future work.

\chapter{Touring Protein Space with Matt}

\label{chapter:c2_touring}


\section{Introduction}

Biologists have long relied on manual classification methods to organize the accepted gold-standard hierarchical classification systems for protein structural domains, SCOP~\cite{Murzin:1995uh, Andreeva:2004ic} and 
CATH.~\cite{Orengo:1997vy, Pearl:2003wb, Greene:2007iu}
Even now, when both SCOP and CATH have switched to hybrid manual/semi-automated
methods~\cite{Greene:2007iu}, these methods are still attempting to fit new protein
domain folds into an initial classification scheme that was derived manually. 
Expert biologists continue to modify the clustering structure based on
sequence, evolutionary, and functional information, not solely based on
geometric similarity of the placement of atoms in the protein backbone.

On the other hand, pairwise protein structural alignment programs superimpose 
protein domains to minimize a distance value based solely on geometric 
criteria~\cite{Gerstein:1998p1487}.
When computational biologists combine such a structural alignment with
hierarchical clustering, they obtain a fully automatic, unsupervised
partitioning of protein structural domains into hierarchical classification
systems~\cite{Tai:2008p1636}.
Such ``bottom up'' protein structure classifications, as they are
called in Valas et al.~\cite{Valas:2009p1477}, have been previously designed based
on VAST~\cite{Madej:1995wi,Gibrat:1996uy},
Dali~\cite{Holm:1996tv,Holm:1998um, Holm:2000vh}, and
others~\cite{Zemla:2007ds}, and have both practical and theoretical appeal. 
Practically,
researchers can assign new protein structures to clusters more quickly without
a human expert. Theoretically, a mathematical characterization of protein
similarity and dissimiliarity, if it proves biologically useful or meaningful,
is objective, uniformly applied, and gives a human-expert-independent map of
the known protein universe.

Unfortunately, multiple researchers have found that SCOP and CATH hierarchical 
classifications of protein structure both differ substantially from each
other~\cite{Hadley:1999p1580,Getz:2002we,Beck:2003p1743}, and also from the 
classification schemata that result from automatic bottom-up unsupervised 
clusterings of protein 
space~\cite{Gerstein:1998p1487, Hadley:1999p1580, Shindyalov:2000wg, 
Beck:2003p1743, Sam:2006te}, 
even when protein chains are broken up into the more modular units of 
``protein domains,'' as SCOP, CATH, and most automated schemes now 
do~\cite{Holm:1998um,Valas:2009p1477}.

Previous papers have characterized those protein domain clusters on
which SCOP and CATH
agree~\cite{Hadley:1999p1580,Getz:2002we,Beck:2003p1743}.  Previous
automatic methods seem to be able to match the
closest-homology {\em family\/} level of the SCOP hierarchy, but were
found to diverge considerably at the more distantly homologous {\em
superfamily\/} and at the quite remotely homologous {\em fold\/} levels of the 
SCOP
hierarchy~\cite{Gerstein:1998p1487, Hadley:1999p1580, Shindyalov:2000wg, 
Beck:2003p1743, Kolodny:2005p1181, Sam:2006te, Suhrer:2007hg},
with similar divergence from
CATH~\cite{Hadley:1999p1580,Harrison:2002wv,Beck:2003p1743}. This is unfortunate,
because, for example, the superfamily level of the SCOP hierarchy
clusters proteins that share similar topologies and are believed to have
evolved from a common ancestor~\cite{Murzin:1995uh}, allowing important
inferences to be made about
function~\cite{Sam:2006te,Valas:2009p1477}.
We focus on SCOP rather than CATH for the remainder of this chapter, though much
of what we say about SCOP could be applied to CATH.
Thus, the superfamily level of
the SCOP hierarchy has strong biological utility: if a fully automated, 
``bottom-up'',
distance-based clustering method cannot approximately replicate a particular
SCOP superfamily,
then such a method is not clearly meaningful or useful. 

This ties into a spirited debate among the computational proteins
community, about the central question of whether ``protein fold
space'' is {\em discrete\/} or {\em continuous}~\cite{Rost:2002wx}. A
continuous view comes from the theory that modern proteins evolved by
aggregating fragments of ancient
proteins~\cite{Rost:2002wx,Harrison:2002wv,Valas:2009p1477,Sadreyev:2009p2118}. A
discrete view comes from evolutionary process constrained by
thermodynamic stability of the structure~\cite{Sadreyev:2009p2118}. In
particular, if most mutations move the conformation of a stable folded
chain away from an ``island'' of thermodynamic structural stability,
then stabilizing selection will promote fold conservation, and
movements between folds will be uncommon~\cite{Choi:2006cv}. If geometric distance and
evolutionary relation approximately coincide, then an automatic method
that approximately matches SCOP at the superfamily level is
conceivable.

We present a bottom-up automatic hierarchical
classification scheme for protein structural domains based on the
multiple structure alignment program Matt~\cite{Menke:2008wu}. Matt,
which stands for ``multiple alignment with translations and twists'',
was specifically developed by our group to geometrically align more
distantly homologous protein domains. It accomplishes this by allowing
flexibility in the form of small, geometrically impossible bends and
breaks in a protein structure, in order to distort that structure into alignment
with another protein. Matt was shown to perform particularly
well compared to competing multiple and pairwise structure alignment
programs on proteins whose homology was similar to the SCOP
superfamily level~\cite{Menke:2008wu, Rocha:2009jm, Berbalk:2009tn}. 
Surprisingly, we find that our automatic classification
scheme based on a pairwise distance value derived from Matt, coupled
with a straightforward neighbor-joining algorithm to construct the
hierarchical clusters~\cite{Simonsen:2008p946} matches SCOP 
better than previous automatic methods, at the superfamily, and
even, to some extent, at the fold level. In comparison, the same
hierarchical clustering method using a pairwise distance value based
on DaliLite~\cite{Holm:2000vh}, a recent implementation of the Dali structural 
alignment algorithm, replicates previous findings and cannot mimic SCOP on the
superfamily level of the SCOP hierarchy. We thus conclude that perhaps
the threshold at which protein domain space is naturally discrete extends at 
least through the
superfamily level, and that perhaps the manually curated SCOP
hierarchy has {\em geometric\/} coherence at the superfamily level
(and in some parts of the fold hierarchy, see Section~\ref{touring-disc}) so 
these
clusters are intrinsic properties of the geometry of fold space, not
just human-generated categories.

A practical implication of our results may be that automatic methods
with a Matt-based distance value may ultimately help speed the
assignment of new protein structural domains to the appropriate place
in the SCOP hierarchy.  We note, however, that determining
where to place a new structure into an existing hierarchy is a much
simpler problem (analogous to ``supervised learning'') than creating
an entire cluster hierarchy from an automatic pairwise distances
from scratch (analogous to ``unsupervised learning''), and fairly
successful methods already exist to correctly place a new structure
into the existing SCOP
hierarchy~\cite{Getz:2002we,Cheek:2004vw,Chi:2006bh}. Thus the primary interest in this result may be that 
if a Matt distance value can ``recover'' SCOP superfamilies
to a great extent, this validates both automatic and hand-curated methods of classification, and the
entire concept of ``superfamily'' at the same time. Namely, at this level of
structural similarity, it appears we may not often have to choose between
evolutionary and geometric criteria for structural domain similarity.

A byproduct of our organization of protein space is that by looking at where 
agreement of our Matt clusters with SCOP is exact, we can construct a new set 
of gold-standard protein multiple structure alignments of distantly homologous 
proteins (and associated decoy sets) for which we can have confidence that the 
Matt structural alignment is meaningful. Thus, we introduce ``Mattbench,'' a 
set of structural alignments at two levels: superfamilies (consisting of 225 
alignments with between 3 and 15 proteins in each alignment set), and folds 
(consisting of 34 alignments with between 3 and 15 proteins in each alignment 
set). Mattbench is meant as an alternative to the 
SABmark~\cite{VanWalle:2005wp} benchmark set, which also attempts to mimic 
SCOP, but Mattbench's alignment sets only cover those subsets of SCOP 
superfamilies and folds where Matt finds geometric consistency. Thus while 
Mattbench is slightly less complete than SABmark in coverage, its alignments 
are likely to be more consistent, making it a better benchmark on which to test 
sequence alignment methods. Complete details on how Mattbench is constructed 
appear in Section~\ref{mattbench}; Mattbench itself can be downloaded from 
\url{http://www.bcb.tufts.edu/mattbench}.

Finally, we remark that this work, like most recent work that compares
different hierarchical classification systems, already presumes the
``structural domain'' as the basic structural unit (as do SCOP and
CATH), where many protein structures contain multiple structural
domains~\cite{Holm:1998um}. The problem of partitioning a protein into
its structural domains is far from
trivial~\cite{Veretnik:2004bn, Holland:2006wi} but there has been much
recent progress in computational methods that split a protein
structure automatically into domains and find the domain
boundaries~\cite{Holland:2006wi, Redfern:2007p1836}. In any case, that is
not the focus of our work, and we assume the protein has already been correctly 
split into domains as a preprocessing step.

\section{Methods}

\subsection{Representative Proteins}

From the 110,776 protein domains of known structure from ASTRAL version 1.75,
we construct a set of representative protein domains filtered to 80\% identity
(according to BLASTP~\cite{Altschul:1997tl}) and a minimum sequence length of
40 residues. This provides a reasonable first pass for identifying groups of
similar protein domains, and allows us to shrink the search space
significantly. The set of clusters is constructed by running a greedy,
agglomerative, minimum-linkage clustering algorithm based on this threshold of
80\% sequence identity. This produces 10,418 groups of proteins that share
significant sequence identity.

From each cluster, we identify a representative. First, we discard engineered
or mutant proteins, and any proteins whose X-ray crystallography resolution is
$> 5.0 \text{\AA}$, from any cluster that has alternative representatives that
meet our criteria. Next, treating each cluster as a (potentially, but not
necessarily, complete) graph whose nodes are the constituent proteins and whose
edge weights are the sequence identity values from the BLASTP alignments with
at least 80\% identity, we consider the weighted degree (sum of edge weights)
of each protein, and we favor the proteins with greatest weighted degree. We
break ties first by the date the structure was determined (preferring more
recent structures), then by the quality of the solved structure. The remaining
ties typically come from sequences with $\geq$ 99\% identity, and we break them
arbitrarily. The resulting set has 10,418 representative protein domains.

\subsection{Distance Values}

For these 10,418 representatives, we performed an all-pairs structural 
alignment using both DaliLite~\cite{Holm:2000vh}, the structural aligner used 
in the FSSP classification scheme~\cite{Holm:1998um} and 
Matt~\cite{Menke:2008wu}. 
In each case, a distance (or dissimilarity) measure is derived for each pair. 
For DaliLite, the Z-score proved to be a good measure, so we used it without 
further modification.

For Matt, we used a new distance value that is a modification of the $p$-value
score computed in Menke, et al.~\cite{Menke:2008wu}. 
Let $c$ be the length of the aligned core shared between the two proteins (in 
residues), $r$ be the RMSD (root mean square deviation) of the alignment, 
$l_{1}$ and $l_{2}$ be the lengths of the two protein domains being aligned (in 
residues), and $k_{1}$, $k_{2}$, and $k_{3}$ be the constants from the Matt 
$p$-value. We compute the distance between two Matt-aligned proteins as 
follows: 

 $\displaystyle d=\frac{1}{ k_{1} \times ( r - k_{2} \times \frac{c^{2}} 
 {\frac{l_{1} + l_{2}}{2}} + k_{3})}$

This value differs from the formula that Matt uses to compute a $p$-value only 
in that it squares the core-length term to place more weight on longer aligned 
cores ($c^{2}$ instead of $c$). We found this improved performance.

\subsection{Distance Threshold}

Based on each of the  Dali Z-score and Matt distances, we next learned the distance cutoffs that most closely mimicked the family, superfamily, and fold levels of the SCOP hierarchy as follows:
\begin{algorithm}
Initialize a training set $T$ and a set of already-chosen pairs $A$\;
    \For{$i = 1 \to 10000$}{
      Choose proteins $p$, $q$ such that $p \neq q$ and $p$ and $q$ are in the same SCOP grouping, and the pair $p,q \not \in A$\;
      Choose proteins $r$, $s$ such that $r \neq s$ and $r$ and $s$ are in different SCOP groupings, and the pair $r,s \not \in A$\;
      $A \leftarrow \left\{p, q\right\}$\; 
      $A \leftarrow \left\{r, s\right\}$\;
      $T \leftarrow dist(p,q)$ with label {\em true}\;
      $T \leftarrow dist(r,s)$ with label {\em false}\;
      Compute true positive rate $R_{tp}$,  true negative rate $R_{tn}$, positive rate $R_{p}$, and negative rate $R_{n}$ for $T$ based on the class labels $true$ and $false$\;
      Determine the value of $d_{p,q}$ that maximizes $\frac{R_{tp}+R_{tn}}{R_{p}+R_{n}}$\;
    }
\end{algorithm}


In other words, we set $d_{p,q}$ to be the value corresponding to the point on
the Receiver Operating Characteristic (ROC) curve that intersects the tangent
iso-performance line~\cite{Vuk:2006wv}, maximizing the sum $R_{tp} + R_{tn}$.
The area under the ROC curve measure (AUC) is a summary statistic that captures
how well the pairwise distance score can discriminate between structures that
share or do not share SCOP cluster membership.

We note that setting the pairwise distance cutoffs (determining the value of
$d_{p,q}$ in step 4) is the only ``supervision'' our algorithm uses in
constructing its clustering (see discussion below). 
We emphasize that once the three single scalar pairwise distance cutoff 
(corresponding to SCOP `family`, `superfamily`,
and `fold` levels of dissimilarity) are set, \emph{no further information} from 
SCOP is utilized to produce the clustering.

\subsection{Clustering and Tree-cutting}\label{touring-clustering}

Based on the distance functions, we computed values for all pairwise alignments
based on the Matt or DaliLite output, and represented this as a distance
matrix. We ran the ClearCut program~\cite{Simonsen:2008p946} in strict
neighbor-joining mode (-N option) to produce a dendrogram based on these Matt
or DaliLite distance values. We then recursively descended this tree to produce
family, superfamily, and fold-level groupings as follows. For a given subtree,
if all leaves (protein domains) in that subtree are within a threshold $t$ of
one another (where $t$ is the family, superfamily, or fold threshold), then
those leaves are all merged into a new grouping of that level. Otherwise, we
recursively descend into the two subtrees of that subtree's root until we reach
a subtree all of whose leaves fall within a given threshold (family,
superfamily, or fold; based on Matt distance or DaliLite Z-score as
appropriate) of one another. Thus, we are performing a total-linkage
clustering, but using the topology of the dendrogram to determine which protein
domains get left out of a given cluster.

We remark that Sam et al.~\cite{Sam:2006te} did an extensive study of
clustering and tree-cutting methods, and looked at their effect on performance
for several distance values. They tested 3 ``SCOP-dependent'' and 7
``SCOP-independent'' tree-cutting strategies. However, their
``SCOP-independent'' strategies all required as input the target number of SCOP
clusters to produce at each level. In contrast, our method discovers the number
of clusters as an organic function of the protein domain space, based only on a
globally learned dissimilarity cutoff; it is thus of independent interest that
we nearly replicate the number of SCOP clusters at each level (see Table~\ref{touring-clusters}).

\subsection{Jaccard Similarity Metric}

The Jaccard index, or Jaccard similarity coefficient, of two sets $A$ and $B$
is defined as $J(A,B) = {\frac{|A \cap B|}{|A \cup B|}}$. Based on the Jaccard
index of a cluster (e.g. family or superfamily or fold) produced by our
algorithm (a ``Matt family'' or ``DaliLite family'') and a SCOP grouping of the
same level, and looking at the identity of protein domains in the two
groupings, we can compare how alike they are. We can thus easily find the most
similar SCOP family to each Matt family, $S \rightarrow M$ and vice versa, $M
\rightarrow S$. This directional mapping is neither one-to-one nor onto, but
each cluster on the `source' side will be mapped to some most-similar cluster
on the `sink' side. The resulting directed graph allows us to explore the
distribution of Jaccard indices as well as the distribution of degrees of each
cluster. A perfect matching would correspond to every Jaccard index being 1.0,
and every cluster having degree 1. Clearly, we do not expect to achieve a
perfect matching, but this metric allows us to compare the quality of
clustering, relative to SCOP, of our algorithm using the Matt distance and the
DaliLite Z-score distance.

Each direction of the metric is produced as follows, using as an example the
comparison of Matt families to SCOP families. Consider the set of Matt families
and SCOP families as a bipartite graph, with the Matt families on one side of
the bipartition and the SCOP families on the other. Initially, the graph has no
edges. For each Matt family, find the most similar (by Jaccard index) SCOP
family. A weighted, directed edge is drawn from each Matt family to its most
similar SCOP family; the edge weight is equal to the Jaccard index, which
ranges from 0 to 1. This is performed until each Matt family has been matched
to a SCOP family. This process is repeated in the other direction, matching
each SCOP family to its most similar Matt family, and the same thing is done
for Matt and DaliLite at the superfamily and fold levels of the SCOP hierarchy.

Recall that in this analysis, as is standard~\cite{Hadley:1999p1580}, we are
considering only the protein domains that were identified as cluster
representatives within each group of protein domains that share 80\% sequence
identity.

\subsection{Benchmark Set} \label{mattbench}

Developers of protein sequence aligners--and structural aligners--typically test
their alignment quality on gold-standard benchmark sets such as 
HOMSTRAD~\cite{Mizuguchi:1998fp} and SABmark~\cite{VanWalle:2005wp}.
With the hierarchy of Matt-derived folds, superfamilies, and families 
constructed, we produced a benchmark set of protein alignments at two levels: 
superfamilies (consisting of 225 alignments), and folds (also referred to as 
the ``twilight zone'' of protein homology, consisting of 34 alignments). 
The ``twilight zone''~\cite{Rost:1999taa} is the region of low sequence identity
(between 20\% and 35\%) at which homology recognition based upon sequence 
alignment becomes difficult.
At the superfamily level, we generated the benchmark set as follows:

\begin{small}
\begin{enumerate}
\item Choose Matt superfamilies that contain at least three representative proteins.
\item For each Matt superfamily:
\begin{enumerate}
\item Identify the most similar SCOP superfamily (by Jaccard index) and take the intersection of it and the Matt superfamily. Call this set $S$.
\item run BLAST on all pairs of proteins in $S$, storing the maximum e-value as $E$.
\item For any pair of proteins $p,q \in S$ that share greater than 50\% sequence identity, remove the shorter one (breaking ties arbitrarily by alphabetic order of protein name). Call this set $S'$. Proceed if and only if $S'$ still has at least three proteins.
\item Run a Matt multiple alignment on $S'$, and store this alignment as the Mattbench alignment for $S'$
\end{enumerate}
\item For each Mattbench superfamily $S$, produce a decoy set $D$ as follows:
\begin{enumerate}
\item Consider every Matt representative protein $p \not \in S$. For each $p$:
\begin{enumerate}
\item discard $p$ if it is in the most similar (by Jaccard index) SCOP superfamily to $p$'s Matt superfamily
\item run BLAST on $p$ against every protein $s \in S$, storing the e-value $e_{s, p}$ and sequence identity $i_{s,p}$
\item run Matt on $p$ against every protein $s \in S$, storing the Matt distance $m_{s,p}$
\item discard $p$ if $\exists s$ such that $i_{s,p} \geq 50\%$
\item discard $p$ unless $\exists s$ such that $e_{s,p} < E$ (this is the $E$ stored as the maximum e-value above)
\item discard $p$ unless $\forall s, m_{s,p} > T_{superfamily}$ (the superfamily threshold used in Matt clustering)
\item if $p$ has not been discarded, add it to the benchmark decoy set $D$.
\end{enumerate}
\end{enumerate}
\end{enumerate} 
\end{small}

The ``twilight zone'' benchmark set is generated in an identical manner, except
that the Matt and SCOP fold levels are used, and the sequence identity cutoff
is 20\% rather than 50\%. The BLAST E-value criterion is the same used by
SABmark~\cite{VanWalle:2005wp} and ensures that each decoy is a useful decoy
rather than an obvious negative match. 
The Matt distance criterion is present because, if the decoy protein is within 
the threshold of some protein in that superfamily, the decoy is only {\em not} 
in that superfamily because of the overall topology of the cluster--that is,
because while the decoy may be similar to some protein in that cluster, it
is not similar enough to all of the proteins of that cluster to warrant 
inclusion. 
The purpose of the decoy set is to act as a set of likely false positives, that
a sequence aligner will find challenging to distinguish from the true positives.
Both benchmarks can be found at
\url{http://www.bcb.tufts.edu/mattbench}.

\section{Results}

\subsection{Pairwise Distance Comparisons}

\begin{small}
\begin{table}[htb!]

\caption{ROC Area for pairwise performance vs. SCOP 
\label{touring-pairwise-roc}}
\begin{tabular}{llllll} 
\hline\noalign{\smallskip}
 & Matt          & DaliLite  \\
\noalign{\smallskip}
\hline
\noalign{\smallskip}
Families       & 0.922         & 0.958 \\
Superfamilies  & 0.842         & 0.615 \\
Folds          & 0.840         & 0.871 \\
\hline
\end{tabular}\\{Note: While DaliLite slightly outperforms Matt at family and fold levels, Matt significantly outperforms DaliLite at the superfamily level.}
\end{table}
\end{small}

We first asked if a pairwise Matt or DaliLite distance cutoff could correctly
distinguish among pairs of proteins that were in the same SCOP cluster from
those that were not. Table~\ref{touring-pairwise-roc} shows the ROC area at the 
SCOP family, superfamily, and
fold level, for the Matt and DaliLite distance scores. Note that at the family
and fold levels, these values are very close (DaliLite outperforms Matt by a
small margin), but at the superfamily level, Matt significantly outperforms
DaliLite, achieving 0.842 ROC area vs. DaliLite's 0.615. Matt was developed to
better align structures at the superfamily level of homology, but the size of
the gap in ROC area is still surprising. We further remark that at the fold
level, DaliLite's seemingly competitive performance is somewhat illusory, since
it is shattering many SCOP folds, each into many tiny pieces (see below).

\subsection{Clustering Performance}

\begin{small}
  \begin{center}
\begin{table}[htb!]
\caption{Number of clusters at each level for each method
\label{touring-clusters}}
\begin{tabular}{llllll} 
\hline\noalign{\smallskip}
& SCOP        & Matt         & DaliLite  \\
\noalign{\smallskip}
\hline
\noalign{\smallskip}
   Families       & 3471        & 3498         & 3081      \\
   Superfamilies  & 1656        & 1716         & 2455      \\
   Folds          & 981         & 891          & 2277      \\
\hline
\end{tabular}\\{Note: Matt more closely matches the number of families, 
superfamilies, and folds in SCOP than DaliLite does. DaliLite clustering 
results in too few families, but too many superfamilies and folds with respect 
to SCOP.}
\end{table}
\end{center}
\end{small}

While the pairwise performance of Matt compared to DaliLite at the superfamily
level is impressive, pairwise similarity does not necessarily translate into
better clustering performance. Thus, we next explore Matt's clustering 
performance. First we give the simplest possible comparison: raw numbers of
clusters produced by Matt and DaliLite compared to SCOP at the three levels.
Recall that unlike the clustering algorithm explored 
by~Tai,~et~al.~\cite{Tai:2008p1636},
the number of clusters produced by our dendrogram and tree-cutting method is a
direct consequence of the pairwise distance threshold, and is not artificially
set to match SCOP (see Section~\ref{touring-clustering}).
Table~\ref{touring-clusters} shows that 
the Matt clustering
produces approximately the same number of clusters as SCOP at all three levels.
While DaliLite also produces approximately the same number of clusters at the
family level, at the superfamily and fold levels it produces many more clusters
than SCOP. We next explore exactly how both methods split and merge SCOP 
clusters in more detail.

\begin{small}
  \begin{center}
\begin{table*}[htb!]
\caption{Descriptive statistics for the family, superfamily, and fold 
levels \label{touring-stats}}
\begin{tabular}{llllllll} 
\hline\noalign{\smallskip}
Family & Max Deg. & $\mu$ Deg. & $\sigma$ Deg. & Min Sim. & $\mu$ Sim. & $\sigma$ Sim. \\
\noalign{\smallskip}
\hline
\noalign{\smallskip}
Matt $\rightarrow$SCOP & 30 & 3.63 & 5.470 & 0.005 & 0.611 & 0.373\\
DaliLite $\rightarrow$ SCOP & 45 & 3.902 & 6.919 & 0.001 & 0.598 & 0.380\\
SCOP $\rightarrow$ Matt & 15 & 1.873 & 2.160 & 0.127 &  0.712 & 0.336\\
SCOP $\rightarrow$ DaliLite & 12 & 1.983 & 1.823 & 0.001 & 0.655 & 0.347\\
\hline\noalign{\smallskip}
Superfamily & Max Deg. & $\mu$ Deg. & $\sigma$ Deg. & Min Sim. & $\mu$ Sim. & $\sigma$ Sim. \\
\noalign{\smallskip}
\hline
\noalign{\smallskip}
Matt $\rightarrow$ SCOP & 28 & 3.633 & 5.094 & 0.003 & 0.587 & 0.389\\
DaliLite $\rightarrow$ SCOP & 153 & 16.61 & 36.54 & 0.001 & 0.428 & 0.406\\
SCOP $\rightarrow$ Matt & 15 & 1.704 & 1.913 & 0.020  & 0.714 & 0.326\\
SCOP $\rightarrow$ DaliLite & 10 & 1.470 & 1.229 & 0.001  &0.713 & 0.324\\
\hline\noalign{\smallskip}
Fold & Max Deg. & $\mu$ Deg. & $\sigma$ Deg. & Min Sim. & $\mu$ Sim. & $\sigma$ Sim. \\
\noalign{\smallskip}
\hline
\noalign{\smallskip}
Matt $\rightarrow$ SCOP & 18 & 3.719 & 4.258 & 0.004  & 0.467 & 0.354\\
DaliLite $\rightarrow$ SCOP & 149 & 26.57 & 40.87 & 0.001  & 0.321 & 0.389\\
SCOP $\rightarrow$ Matt & 6 & 1.958 & 1.122 & 0.022 & 0.512 & 0.326\\
SCOP $\rightarrow$ DaliLite & 3 & 1.117 & 0.353 & 0.001  & 0.758 & 0.299\\
\hline
\end{tabular}\\{Note: $\mu$ Degree is the average number of clusters from the 
first scheme that map to a single cluster in the second, and $\sigma$ Degree 
gives the standard deviation. Similarly, we give min, $\mu$, and $\sigma$ of 
the Jaccard similarity.}
\end{table*}
\end{center}
\end{small}

The Jaccard index serves as a good indicator of how well Matt and DaliLite
match SCOP. As the raw numbers of clusters in Table~\ref{touring-clusters} 
suggest, DaliLite often
shatters SCOP superfamilies into multiple clusters. DaliLite also shatters SCOP
folds into many more shards on average than Matt. How can this be given the
very similar pairwise classification performance at the fold level? We defer
this question until Section~\ref{touring-disc}. 
We note that even at the family
level, Matt performs slightly better than DaliLite at both the average degree
and average Jaccard similarity metrics. The average number of Matt or DaliLite
families that match to a single SCOP family is between 3.5 and 4; however,
notice that a large majority of Matt or DaliLite families map to a single SCOP
family and the average is pulled up by a few outliers (see histograms in 
Figures~\ref{family-split}~and~\ref{family-split-scop}). 
Average degree values at the superfamily and fold levels stay nearly
constant for Matt, whereas DaliLite's average degree values rise to 16.61 for
the superfamily level and 26.57 at the fold level. In the other direction,
considering how many Matt or DaliLite clusters span multiple SCOP clusters, at
the family level the average degree for Matt and DaliLite are nearly identical
(between 1.8 and 2). At the superfamily and fold levels, we would expect
DaliLite to outperform Matt by virtue of the fact that it creates many smaller
clusters (see Table~\ref{touring-clusters}), and DaliLite does, but by a fairly 
small margin (1.4 to
1.7 at the superfamily level and 1.1 to 2 at the fold level). The distributions
are displayed in more detail in the histograms in 
Figures~\ref{family-split},\ref{family-split-scop},
\ref{superfamily-split},~\ref{superfamily-split-scop},
\ref{fold-split},~and \ref{fold-split-scop}.

\begin{figure}[htb!]
\begin{center}
  \fbox{\includegraphics[width=5in]{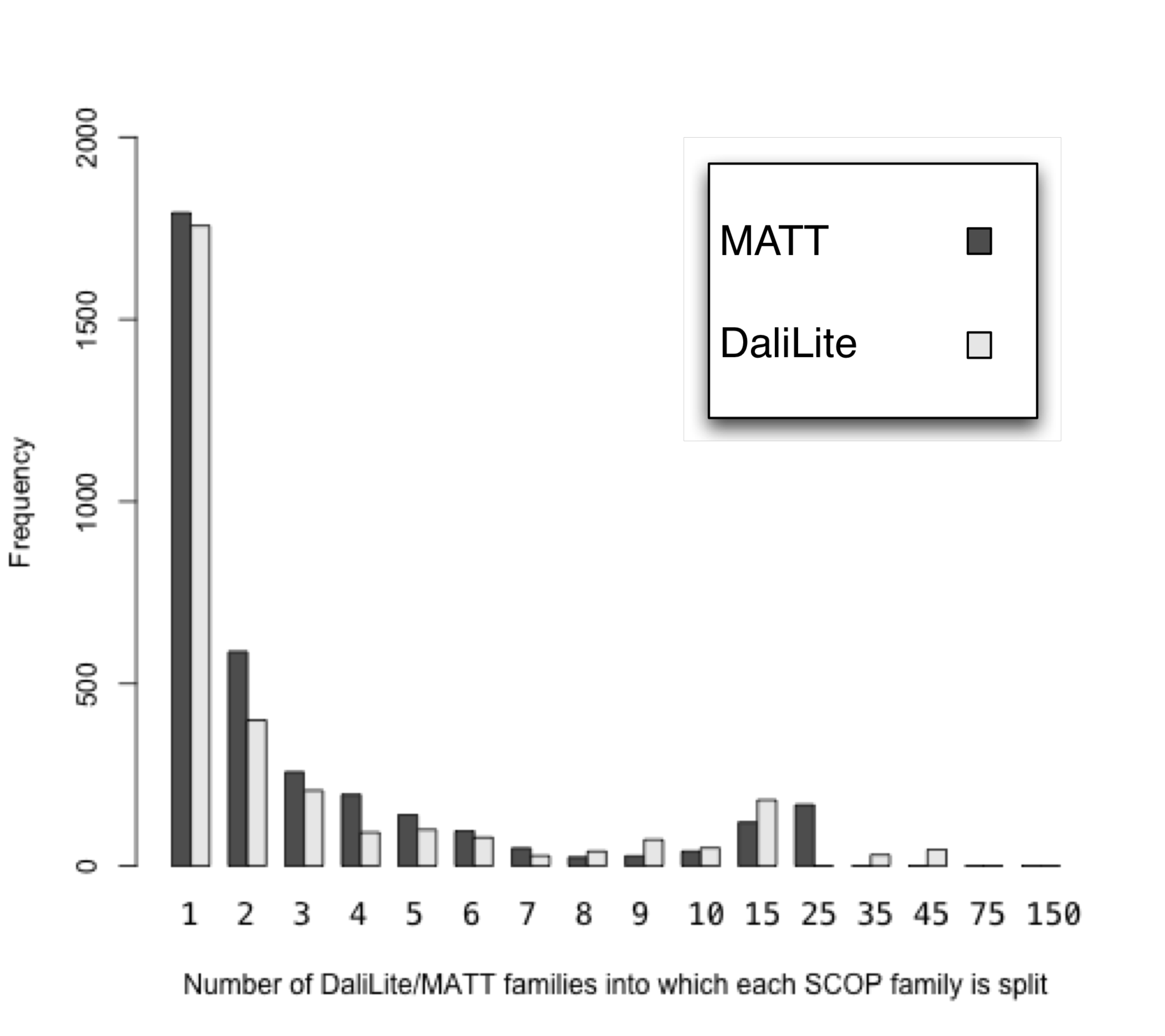}}
   \caption{Number of Matt vs. DaliLite families into which each SCOP family is 
shattered.}
   \label{family-split}
 \end{center}
\end{figure}

\begin{figure}[htb!]
\begin{center}
  \fbox{\includegraphics[width=5in]
  {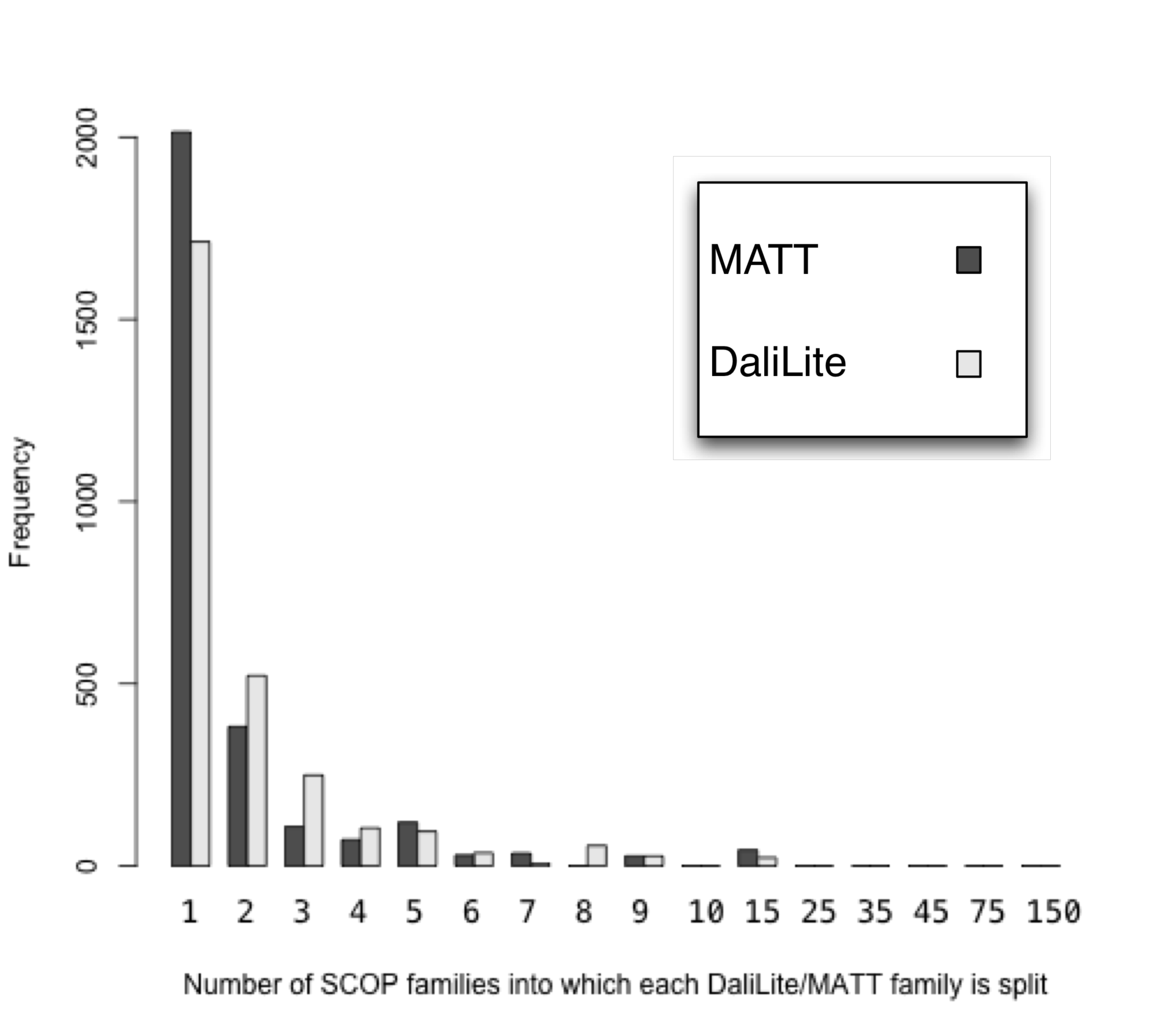}}
   \caption{Number of SCOP families into which each Matt or DaliLite family is 
shattered.}
   \label{family-split-scop}
 \end{center}
\end{figure}

\begin{figure}[htb!]
\begin{center}
  \fbox{\includegraphics[width=5in]
  {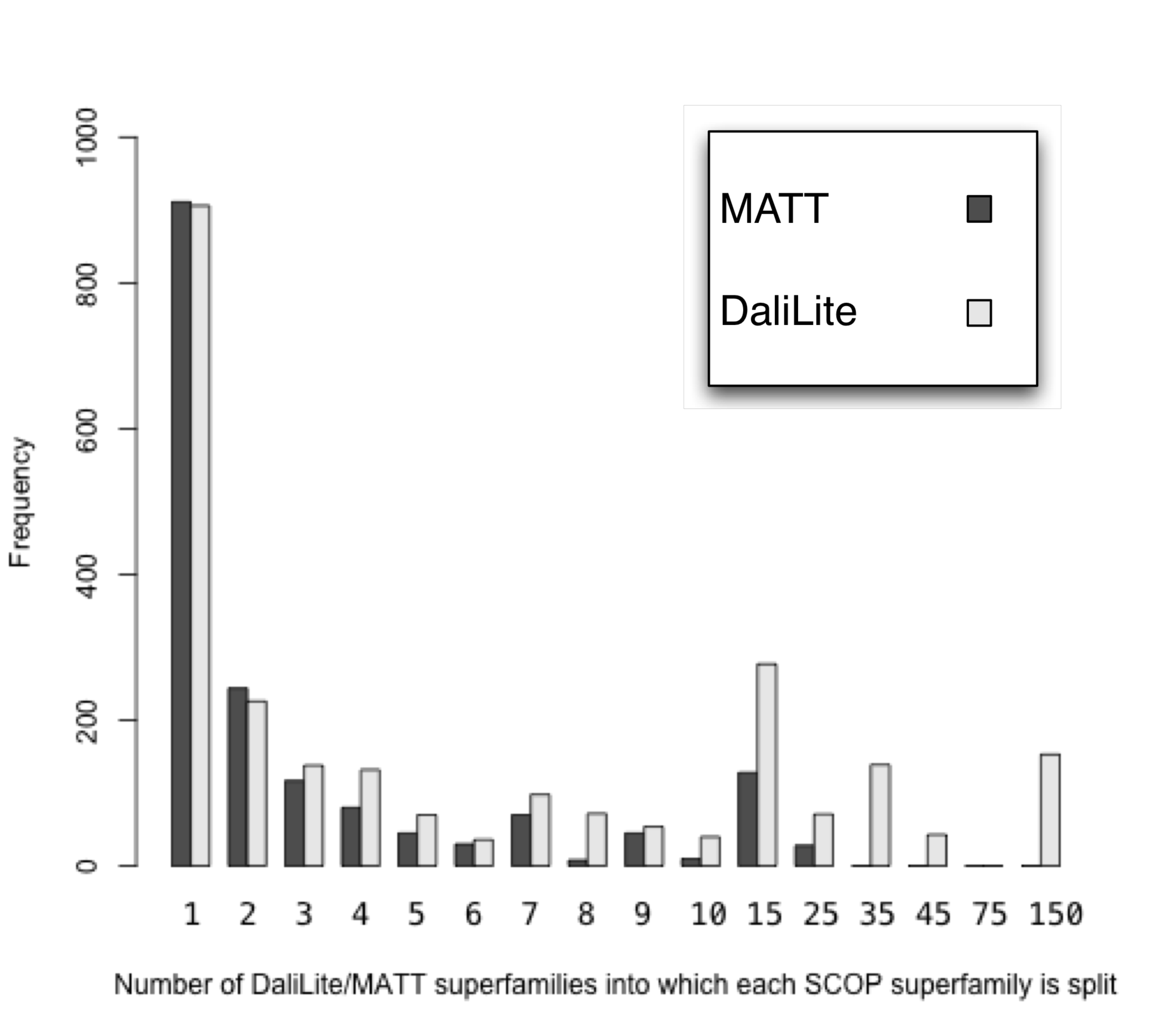}}
   \caption[Number of Matt vs. DaliLite superfamilies into which each SCOP 
   superfamily is shattered.]{Number of Matt vs. DaliLite superfamilies into 
   which each SCOP superfamily is shattered.
   Note the tail of the distribution, in which DaliLite breaks SCOP 
   superfamilies into many small pieces.}
   \label{superfamily-split}
 \end{center}
\end{figure}

\begin{figure}[htb!]
\begin{center}
  \fbox{\includegraphics[width=5in]
  {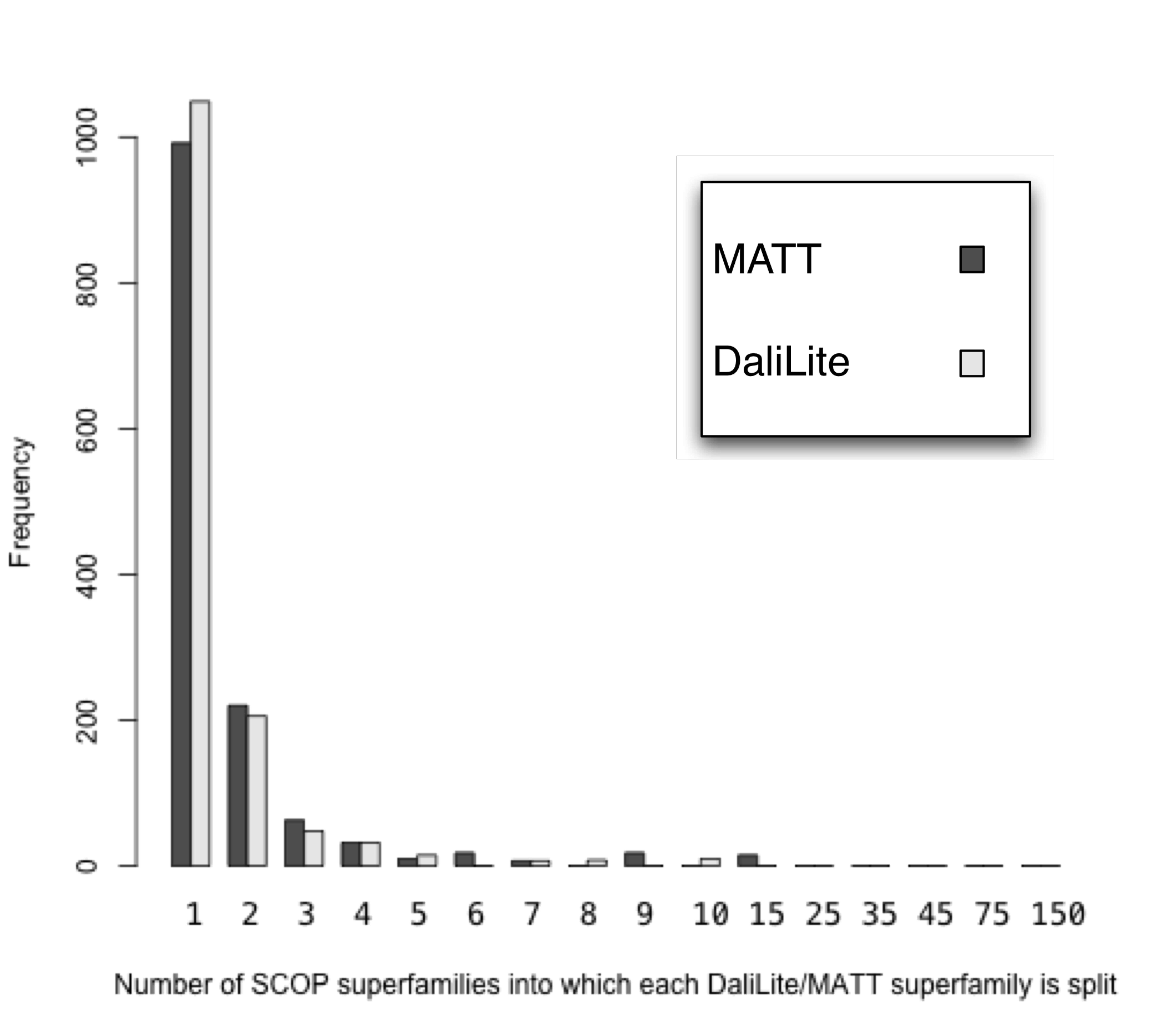}}
   \caption{Number of SCOP superfamilies into which each Matt or DaliLite 
   superfamily is shattered.}
   \label{superfamily-split-scop}
 \end{center}
\end{figure}

\begin{figure}[htb!]
\begin{center}
  \fbox{\includegraphics[width=5in]
  {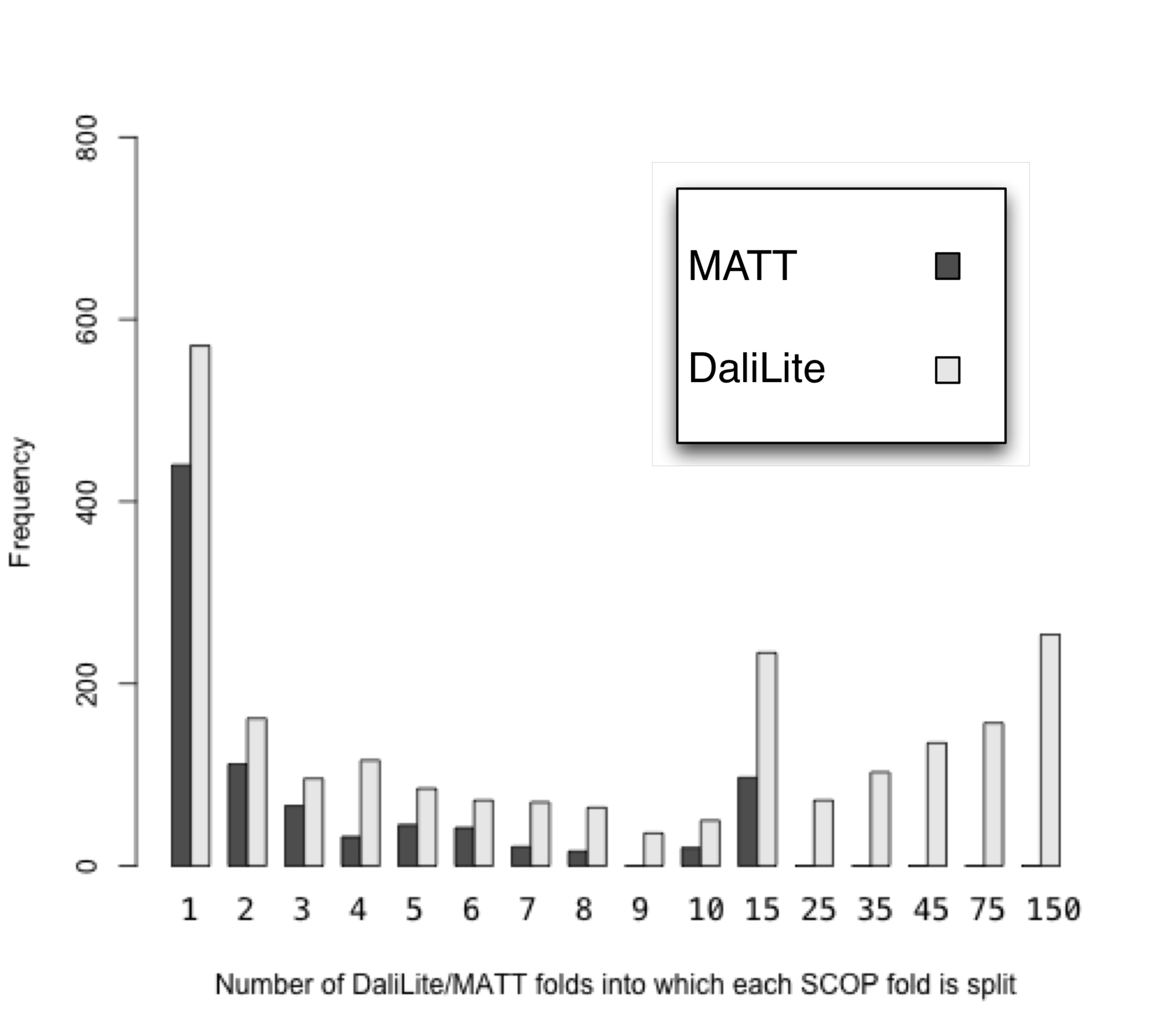}}
   \caption{Number of Matt vs. DaliLite folds into which each SCOP 
   fold is shattered.
   Note the tail of the distribution, in which DaliLite breaks SCOP 
   folds into many small pieces.}
   \label{fold-split}
 \end{center}
\end{figure}

\begin{figure}[htb!]
\begin{center}
  \fbox{\includegraphics[width=5in]
  {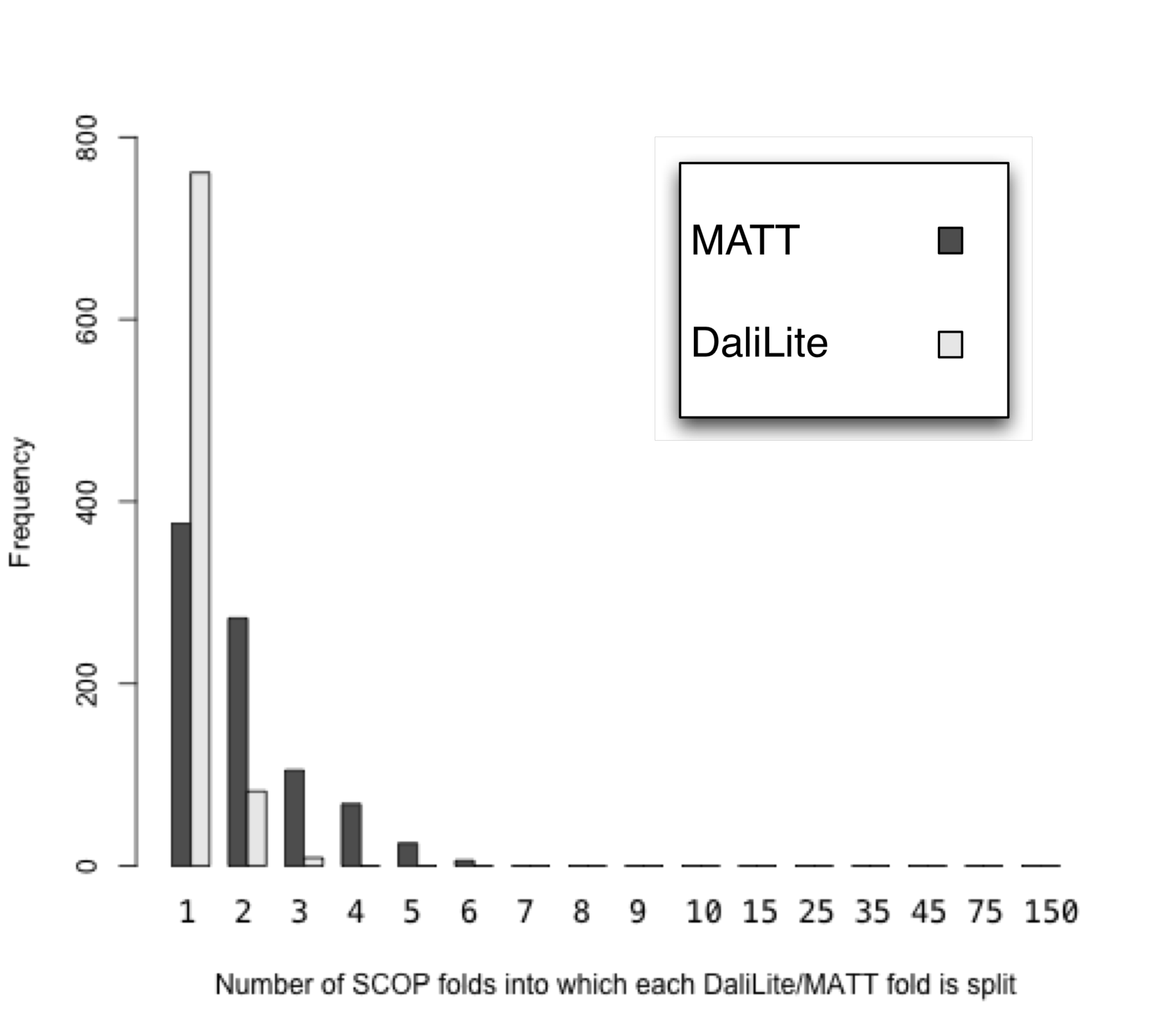}}
   \caption{Number of SCOP folds into which each Matt or DaliLite 
   fold is shattered.}
   \label{fold-split-scop}
 \end{center}
\end{figure}

\subsection{Specific Example}

We thought it would be illuminating to provide a pictorial example of a single
SCOP superfamily that Matt splits into two superfamilies. Consider the SCOP
superfamily ``DHS-like NAD/FAD-binding domain'' (SCOP ID 52467). 
There are 24 proteins
from this superfamily in our representative set. Matt places 17 of them in one
superfamily, but the remaining 7 in a different superfamily. 
Figure~\ref{touring-example}a gives an
example protein from the Matt superfamily of size 17, while 
Figure~\ref{touring-example}b gives an
example protein from the Matt superfamily of size 7. Both Matt superfamilies
contain the same single flat $\beta$-sheet of 6 or 7 strands, surrounded by
$\alpha$-helices. In addition, the proteins in the Matt superfamily of size 7
have a second short 3-4 strand $\beta$-sheet. The second short $\beta$-sheet is
physically on one end of the first $\beta$-sheet in 3-dimensional space, but
sometimes occurs between the second-to-last and last $\beta$-strands in the
first $\beta$-sheet in terms of linear (sequence) ordering, or else at the very
end. The second $\beta$-strand is also partially surrounded by $\alpha$-helices.

Because of the common central motif, it is very possible that these proteins
are evolutionarily related and thus belong in the same SCOP superfamily.
However, geometrically, the additional short $\beta$-sheet is significant
enough for Matt to place them in different superfamilies. Matt does, however,
place them in the same fold.

\begin{figure*}[htb!]
\centering
\mbox{
\subfigure[{\bf Example from Matt superfamily 252} Pyruvate oxidase from {\em Lactobacillus plantarum}, PDB ID 2ez9:a]{\includegraphics[width=2.5in]{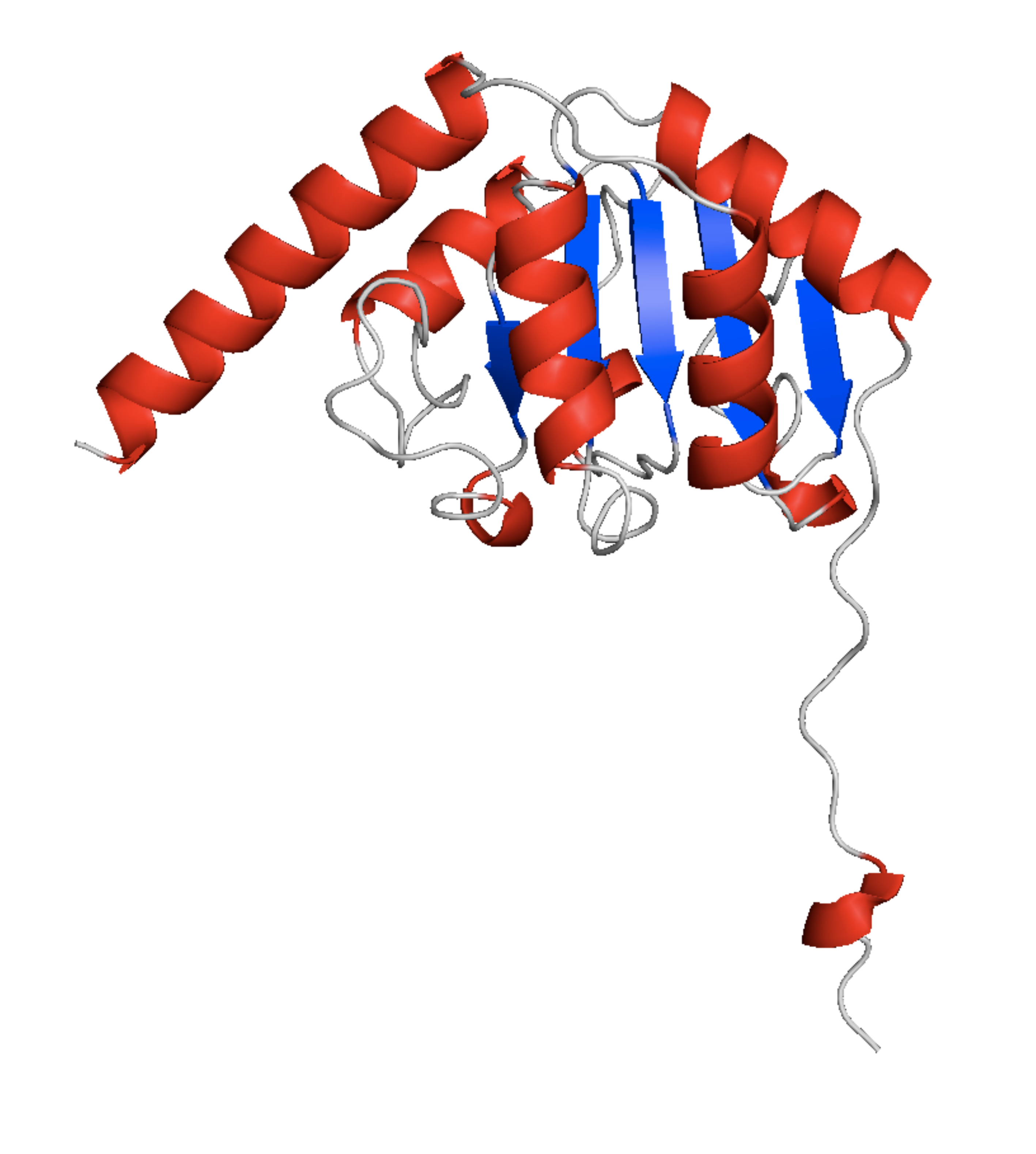} }
\quad
\subfigure[{\bf Example from Matt superfamily 212} AF1676, Sir2 homolog (Sir2-AF1) from Archaeon {\em Archaeoglobus fulgidus} , PDB ID 1m2k:a]{\includegraphics[width=2.5in]{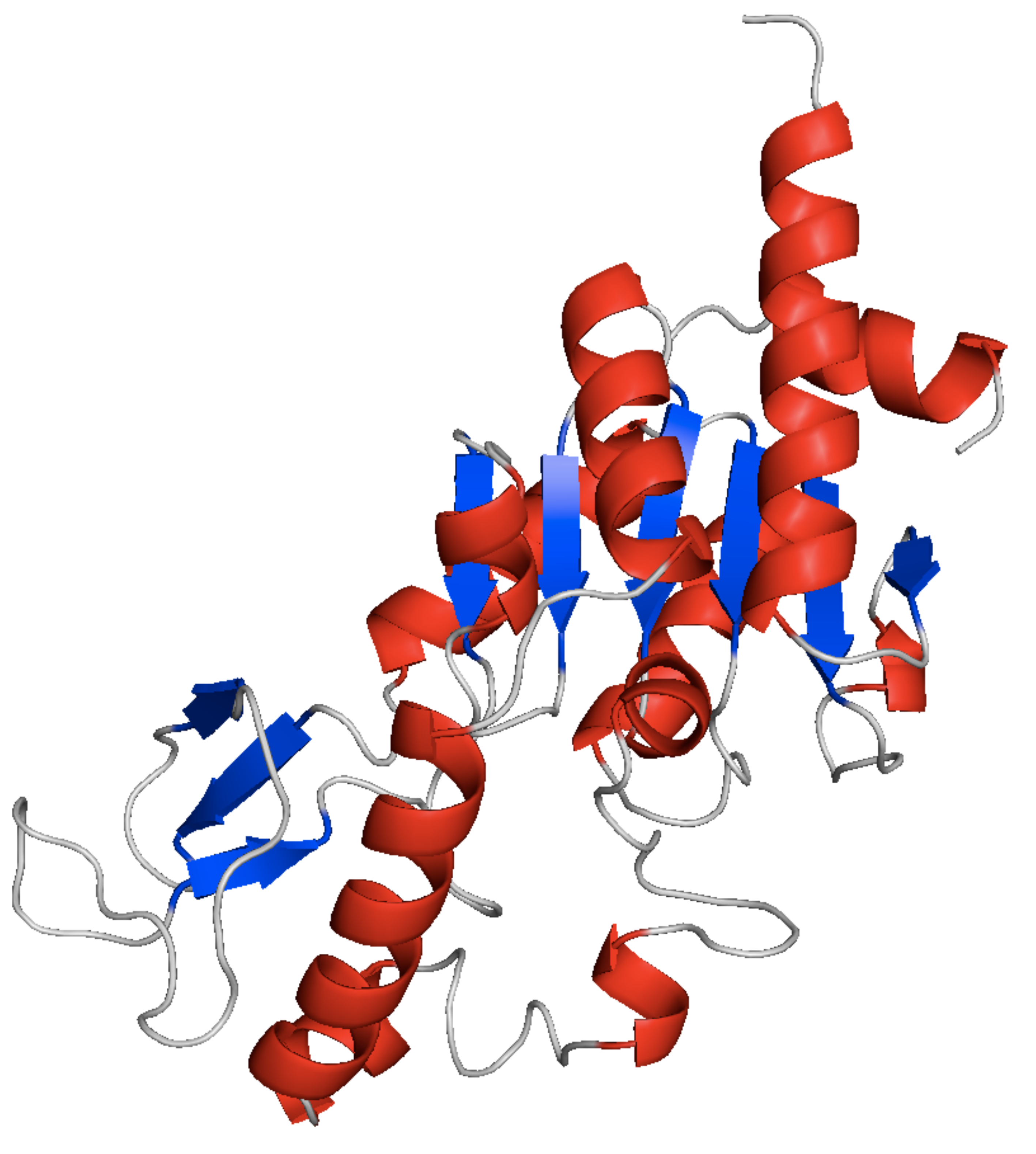}}}
\caption{Example of a SCOP superfamily split by Matt}\label{touring-example}
\end{figure*}

\section{Discussion}\label{touring-disc}

We have shown that using more modern structure alignment programs, we can
approximately match SCOP at the superfamily level. Of course, any mapping
between one set of clusters based on geometric equivalence, and another set of
clusters based on geometric as well as evolutionary equivalence, will be
imperfect--yet the Matt clusters at the superfamily level seem sufficiently
interesting that differences between Matt and SCOP could be illuminating.

As noted earlier, DaliLite tends to shatter SCOP folds into many more shards
than Matt. How can this be, given the very similar pairwise classification
performance at this level? One possibility is that the Matt-based distance
value is more stable in regions far beyond the specific thresholds we learned,
and that this leads to the topology of the resulting dendrogram (before
cutting) more faithfully representing the relationships between more- and 
less-closely related folds. 
In other words, DaliLite's Z-scores may result in more
`spoilers'--individual proteins with large distances to many other proteins in 
the same cluster--that break up clusters (due to our total-linkage requirement) 
than Matt's distance value.
While we have only compared Matt to DaliLite, comparisons to other aligners 
such as TM-Align~\cite{Zhang:2005do} would undoubtedly be interesting.
We focused on the comparison to DaliLite due to it being the aligner underlying
the FSSP database~\cite{Holm:1998um}.

What do Matt's clustering results mean for protein fold space at the ``fold''
level of structural homology? Here, while the Matt clustering clearly seems
more informative than that produced by DaliLite, performance is still uneven.
There seem to be some SCOP folds where the Matt split appears meaningful, and
others where it is more arbitrary. For example, a notoriously difficult SCOP
fold for multiple automatic methods is the enormous $\beta/\alpha$ TIM barrel
fold. SCOP places 33 separate superfamilies into this one fold, but both of our
clustering approaches seem to split it into multiple folds. For example,
DaliLite splits the TIM barrel SCOP fold into 106 separate folds. Matt splits
the TIM barrel SCOP fold into `only' 17 separate folds, which is better than
106, but inspection of the boundaries between these Matt fold classes shows
more continuity of shape, and the cuts appear to be somewhat arbitrary.

Thus, while touring protein space with Matt seems to lend support to a more
discrete view of protein space through the superfamily level, further study of
individual clusters may be warranted to determine the breakpoint distance at
which continuity takes over. Perhaps the degree of similarity of different
individual SCOP folds can be characterized, similarly to what Suhrer, et
al.~\cite{Suhrer:2007hg} did at the family level.

We have made the Mattbench benchmark set available at
\url{http://www.bcb.tufts.edu/mattbench}. 
We hope that developers of protein sequence
alignment tools will consider testing their performance on Mattbench, as well as
SABmark~\cite{VanWalle:2005wp} and
HOMSTRAD~\cite{Mizuguchi:1998fp}.

\chapter{Simplified Markov Random Fields and Simulated 
Evolution Improve Remote Homology Detection for Beta-structural Proteins}

\label{chapter:c3_smurflite}

\section{Introduction}

Many researchers use hidden Markov models (HMMs) to annotate proteins according to homology,
with popular systems such as Pfam~(\cite{Finn:2008iw}) and Superfamily~(\cite{Wilson:2007cm}) based on HMM 
methods integrated into UniProt. However, HMMs are limited in their power to recognize remote homologs
 because of their inability to model statistical dependencies between amino-acid residues that are close in space but
 far apart in sequence~(\cite{Lifson:1980vd, Zhu:1999wr, Olmea:1999ug, Cowen:2002p588,Steward:2002wz}).

For this reason, many have suggested (\cite{White:1994ty, Lathrop:1996gm, Thomas:2008uw, Gopalakrishnan:2009wx, Menke:2010ti, Peng:2011vk})
that more powerful Markov random fields (MRFs) be used. MRFs employ an auxiliary
 {\em dependency graph\/} which allows them to model more complex statistical
 dependencies, including statistical dependencies that occur between
 amino-acid residues that are hydrogen-bonded in $\beta$-sheets.
 
However, as the dependency graph becomes more complex, major design 
difficulties emerge. 
First, the MRF becomes more difficult to train.
Second, finding the optimal-scoring parse of the target to the model quickly 
becomes computationally intractable.

We have built a fully automated system, SMURFLite, that combines
the power of Markov random fields with Kumar and Cowen's 
Simulated Evolution (\cite{Kumar:2010wv}) (which offloads information
about pairwise dependencies in $\beta$-sheets into new, artificial
training data), in order to build the first MRF models that are computationally 
tractable for {\em all \/} $\beta$-structural
proteins, even those with limited training data. 
The SMURFLite system builds in part on the
SMURF MRF (\cite{Menke:2010ti}), which uses multidimensional dynamic
programming to simultaneously capture both standard HMM models and the
pairwise interactions between amino acid residues bonded together in
$\beta$-sheets. 

Unlike the full SMURF MRF, where the computational
requirements of the random field become prohibitive on folds with
deeply interleaved $\beta$-strand pairs, such as barrels, SMURFLite is
tractable on all $\beta$-structural proteins (see Figure~\ref{beta-interleave}).
SMURFLite enables researchers to trade modeling power for computational cost
by tuning an {\em interleave threshold\/}. 
The interleave threshold represents the maximum number of unrelated 
$\beta$-strands that can occur in linear sequence between the $\beta$-strands
hydrogen-bonded in a $\beta$-sheet while
still being retained as pairwise dependencies in the MRF. As the interleave
threshold increases, computation time increases, but so does the power of the 
MRF (see Figure~\ref{barwin_barrel}).

We first test SMURFLite on all propeller and barrel folds in the mainly-$\beta$
class of the SCOP hierarchy in stringent cross-validation
experiments. We show a mean 26\% (median 16\%) improvement in Area Under Curve 
(AUC) for
$\beta$-structural motif recognition as compared to
HMMER (\cite{Eddy:1998ut}) (a popular HMM method) and a mean 33\% (median
19\%) improvement as compared to RAPTOR (\cite{Xu:2003p3417}) (a
well-known threading method), and even a mean 18\% (median 10\%) improvement in AUC over
HHPred (\cite{Soding:2005ff, Soding:2005fa}) (a profile-profile HMM method), despite HHpred's use of extensive
additional training data. We demonstrate SMURFLite's ability to scale to whole
genomes by running a SMURFLite library of 207 $\beta$-structural SCOP
superfamilies against the entire genome of \textit{Thermotoga
maritima}, and make over a hundred new fold predictions (available at \url{http://smurf.cs.tufts.edu/smurflite}). The majority of these predictions are for genes that display
very little sequence similarity with any proteins of known structure,
demonstrating the power of SMURFlite to recognize remote homologs.

We offer an online server (\url{http://smurf.cs.tufts.edu/smurflite}) for
predicting remote homologs from our library of  207 mainly-$\beta$ 
superfamilies using SMURFLite. The online server
sets the interleave threshold (the parameter that determines the
complexity of the MRF) to 2; we have also
shown that increasing the interleave number for SMURFLite can
dramatically improve accuracy, but at a great computational cost. 
While the primary intent of using simulated
evolution in conjunction with simplified MRFs is to compensate for the
removal of highly-interleaved $\beta$-strand pairs required for
computational feasibility, we surprisingly find that simulated
evolution can still improve full-fledged SMURF in cases of sparse
training data. For instance, the 5-bladed $\beta$-propellers have only
three superfamilies in SCOP, two of which contain only one family.  We
find that for the 5-bladed $\beta$-propeller fold, combining SMURF and
simulated evolution improves AUC from 0.73 for full SMURF alone to
0.89.

\section{Methods}

\subsection{Summary of SMURF Markov random field framework} \label{mrf}

SMURF and SMURFLite rely on training data in the form of a multiple structure 
alignment with $\beta$-strand annotation.
This alignment is created using the Matt program (\cite{Menke:2008wu}).
$\beta$-strand annotation is done on a structure-by-structure basis, where the 
$\beta$-strand residue pairing is determined using the same algorithm 
implemented by the Rasmol (\cite{Sayle:1995ux}) visualization program.
Essentially, $\beta$-strands are detected by analyzing the $\psi$, $\phi$, and
$\omega$ angles, as well as the distance between a hydrogen from the amine 
group and an oxygen from the carboxyl group of the amino acid at which that
hydrogen is pointing, if any.
If this hydrogen and oxygen point at each other and are within $3 \text{\AA}$,
they are considered to be hydrogen-bonded.
If successive hydrogen bonds are to amino acids near (3-5 residues) in sequence,
an $\alpha$-helix is inferred; if those bonds are to distant amino acids in
sequence, a $\beta$-strand is inferred.
A postprocessing step annotates those $\beta$-strand residues that appear in 
more than half the structures in the alignment as $\beta$-conserved. 
As gaps in $\beta$-strands would complicate training, this post-processing step 
makes $\beta$-conserved
template strands contiguous in the alignment exactly as in \cite{Menke:2010ti}. 
Specifically, any gaps in a column, that otherwise comprises at least half
$\beta$-structural amino acids, are removed from the alignment.
Recall that up to half the sequences are allowed to \emph{not} 
participate in $\beta$-strands at any given position of the alignment, the
non-$\beta$-strand amino acids in those positions are still treated as if they
participate in $\beta$-strands.

The result at
this stage is a sequence alignment (resulting from the Matt structural
alignment) with conserved $\beta$-strand pairs annotated according to the 
residue positions and conformation (buried or exposed to solvent).

The pairwise probability portion of the MRF is based on the $\beta$~probability 
tables that were computed by collecting a set of amphipathic 
$\beta$-sheets
from the Protein Data Bank (PDB) (\cite{Berman:2000hl}) and tabulating the 
frequencies of
pairs of hydrogen-bonded residues in two tables, one for buried residues and one
for residues exposed to solvent (\cite{Bradley:2001tj}, \cite{Menke:2010ti}). 
The $\beta$-structural proteins chosen were filtered to 25\% sequence identity
to prevent over-representation of highly-sampled sequences.
Amphipathic $\beta$-sheets are those $\beta$-sheets that are ``confused'' as to 
their hydrophobicity, and thus have residues whose sidechains may alternate as 
to the direction in which they pack.
For each residue position, the most
likely conformation (buried or exposed) is chosen based on whether that residue
pairing is most probable from the buried or exposed $\beta$-pairing tables.

Given a trained MRF, SMURF and SMURFLite align a query sequence to the MRF. The
query phase of SMURF and SMURFLite computes the alignment of the sequence to the
states of the MRF that maximizes the combined score: 

\begin{equation}log \left(\text{HMM score}\right) + log \left(\text{pairwise score}\right) \nonumber \end{equation}

 In this combined score, the HMM score is the conditional probability
of observing the sequence given the HMM portion of the model, and the pairwise
score is the conditional probability of observing the paired $\beta$-strand
components of the sequence given the $\beta$-pair portion of the model. 
Let the sequence have residues $r_{1}..r_{n}$, and the MRF have match states 
$m_{1}..m_{l}$, deletion states $d_{1}..d_{l}$, and insertion states 
$i_{1}..i_{l}$. 
Suppose that $r_{1}..r_{k}$ and match states $m_{1}..m_{s}$ have been assigned. 
Then, the probability of assigning $r_{k}$ to the next match state 
$m_{j}=m_{s+1}$ is:
\begin{eqnarray}Pr\left[m_{j}|r_{k},u_{j-1}\right] = HMM\left[ m_{j},r_{k} 
  \right] \cdot \nonumber \\ transition\left[u_{j-1},m_{j}\right] \cdot 
  \nonumber \\ \beta  strand\left[r_{j},r_{k},m_{j},m_{k}\right] \nonumber 
\end{eqnarray}
where $u_{j-1}$ can be either $d_{j-1}$, $i_{j-1}$, or $m_{j-1}$ depending 
on whether the current state is a deletion, insertion, or match state. 
When the current state is a match state, the SMURFLite template replaces the 
$transition\left[u_{j-1},m_{j}\right]$ term with a value of $1$. 
The $\beta strand$ component is set to be identically $1$ unless the particular 
match state $m_{j}$ participates in a $\beta$-strand that is matched with a 
state $m_{k}$ earlier in the template. 
This component is the primary difference between our MRF and an ordinary 
HMM~(\cite{Menke:2010ti}). 

SMURFLite computes the maximum score of a sequence using multidimensional 
dynamic programming on the MRF. This dynamic programming resembles the classic 
Viterbi algorithm (\cite{Viterbi:1967hq}) used on HMMER's ``Plan7'' 
(\cite{Eddy:1998ut}) HMMs, except that some states are $\beta$-strand states, 
which are required to be match states, and which are paired with other 
$\beta$-strand nodes in the model. 
Because the pairwise component of the score can only be calculated for a given 
MRF node once it is determined what residue occupies the paired MRF node 
earlier in the sequence, each time the dynamic programming reaches a state in 
the MRF that corresponds to the first residue of the first $\beta$-strand in a 
set of paired $\beta$-strands, we need to keep track of multiple cases, depending on what residue in sequence is mapped to that state. 
SMURFLite uses a multidimensional array to memoize these possible subproblem 
solutions. 
A maximum gap size is set to the longest gap seen in the training data plus 20, 
for computational efficiency. 
When paired $\beta$-strands follow each other in sequence with no interleaving 
$\beta$-strands between them, the number of dimensions in the table for the 
dynamic programming is directly proportional to the maximum gap length. 
Let us call the last MRF state for the first of every pair of $\beta$-strands a ``split'' state and the first MRF state for the second of that pair a ``join'' state. 
Then, at every split state, the number of dimensions of the dynamic program 
will be multiplied by the maximum gap length, because the dynamic program must 
keep track of scores for each possible sequence position (up to the maximum gap 
length) that could be mapped to that state. 
At the corresponding join state, the number of dimensions will be reduced by 
the maximum gap length, because the scoring function can calculate all the 
pairwise probabilities of placing that residue into the join state, and then 
simply take the maximum of all ways to have placed its paired residue into the 
split state. 
However, when other $\beta$-strands are interleaved, the dynamic program must 
open additional multidimensional tables before clearing the previous ones from 
memory.
An example of this interleaving is shown in Figure~\ref{beta-interleave}.
Thus, the number of elements in the multidimensional table is never more than 
the sequence length times the maximum gap length raised to the power of the 
interleave number.

\subsection{Datasets} \label{datasets}

From SCOP (\cite{Murzin:1995uh}) version 1.75, we chose the folds ``5-bladed 
Beta-Propellers'', ``6-bladed Beta-Propellers'', ``7-bladed Beta-Propellers'', 
and ``8-bladed Beta-Propellers''.  We also chose superfamilies from all of the 
mostly-$\beta$ folds containing the word ``barrel'' in their description, 
whether open or closed, restricted to those superfamilies comprising at least 
four families (in order to facilitate leave-family-out cross-validation). These 
superfamilies were: ``Nucleic acid-binding proteins'' (50249), ``Translation 
proteins'' (50447), ``Trypsin-like serine proteases'' (50494), ``Barwin-like 
endoglucanases'' (50685), ``Cyclophilin-like'' (50891), ``Sm-like 
ribonucleoproteins'' (50182), ``PDZ domain-like'' (50156), ``Prokaryotic 
SH3-related domain'' (82057), ``Tudor/PWWP/MBT'' (63748), ``Electron Transport 
accessory proteins'' (50090), ``Translation proteins SH3-like domain'' (50104), 
``Lipocalins'' (50814) and ``FMN-binding split barrel'' (50475). 
Of these, we removed the superfamilies ``Lipocalins'' and ``Trypsin-like serine 
proteases,'' which were not structurally consistent enough to permit a multiple 
structure alignment for training HMMER or the SMURF variants, and which were 
broken into distinct superfamilies by~\cite{Daniels:2011dc}, with the result 
that 11 superfamilies containing barrels were selected. 
In addition, for the whole-genome search on \textit{Thermotoga maritima}, out 
of 354 total superfamilies within the SCOP class ``All beta proteins'', 288 
(81\%) of which contain at least two protein chains, 207 superfamilies (71\%) 
were structurally consistent enough to be aligned using the Matt 
(\cite{Menke:2008wu}) structural alignment program.  We built SMURFLite 
templates for these 207 superfamilies, and obtained from UniProt the protein 
sequences for \textit{Thermotoga maritima}, comprising 1852 genes.

\subsection{Training and testing process} \label{training}

For the $\beta$-propeller folds, strict leave-superfamily-out cross-validation 
was performed. 
The propeller folds are structurally highly consistent (\cite{Menke:2010ti}), 
and thus high-quality multiple structure alignments 
were possible using Matt (\cite{Menke:2008wu}) without descending to the superfamily level. 
For each propeller fold, its constituent superfamilies were identified.
For each superfamily, one pass of cross-validation was performed. 
Given a superfamily to be left out, a training set was established from the protein 
chains in the remaining superfamilies, with duplicate sequences removed. 
An HMM (in the case of HMMER and HHPred) or MRF (in the case of SMURF and 
SMURFLite) was trained on the training set (HMMER parameter settings are 
discussed below). 
Protein chains from the left-out superfamily were used as positive test 
examples. 
Negative test examples were protein chains from all other folds in SCOP classes 
1, 2, 3 and 4 (including propeller folds with differing blade counts), 
indicated as representatives from the non-redundant Protein Data Bank 
repository (nr-PDB) (\cite{Berman:2000hl}) database with non-redundancy set to 
a BLAST E-value of $10^{-7}$. 

The $\beta$-propellers are atypical of most $\beta$-structural SCOP folds, in 
that they structurally align well at the fold level of the SCOP hierarchy. 
For the $\beta$-barrel superfamilies, strict leave-family-out cross-validation 
was performed. 
The barrel superfamilies are distinguished by strand number and shear as well 
as other structural features (\cite{Murzin:1995uh}), and so like most 
$\beta$-structural motifs they do not align well structurally at the fold 
level. 
For this reason, the superfamily level was chosen for training.
This cross-validation was similar to that chosen for the $\beta$-propellers,
except that it was done at the superfamily level, and thus each pass of the
cross-validation involved leaving out a family and training on a structural 
alignment of representatives from the remaining families in that superfamily.

Each test example was aligned to the trained HMM (from HMMER and HHPred) and 
MRF, and was also threaded, using RAPTOR, against each individual chain in the 
training set (RAPTOR parameters are discussed below). The score reported for 
HMMER and HHPred was the output HMM score, and the score reported for SMURF and 
SMURFLite was the combined HMM and pairwise score from the MRF. For RAPTOR, the 
score reported for a test example was the highest score from all the scores 
resulting from threading that test example onto each chain in the training set. 
For each training set, the scores for each method were collected and a ROC 
curve (a plot of true positive rate versus false positive rate) computed. We 
report the area under the curve (AUC statistic) from this 
ROC curve~(\cite{Sonego:2008uy}).

\subsection{$p$-values}
SMURFLite computes the $p$-value for an alignment in a similar manner to HMMER, 
using an 
extreme value distribution (EVD) (\cite{Eddy:1998ut}). 
An EVD is fitted to the distribution of raw scores over a random sampling of 
5000 protein chains from across the SCOP hierarchy. 
The $p$-value is then simply computed as $1-cdf\left( x \right)$ for any raw 
SMURFLite score $x$, where $cdf$ is the cumulative distribution function for 
the EVD.

\begin{figure*}[h!tpb]
\centerline{\includegraphics[width=5.5in]{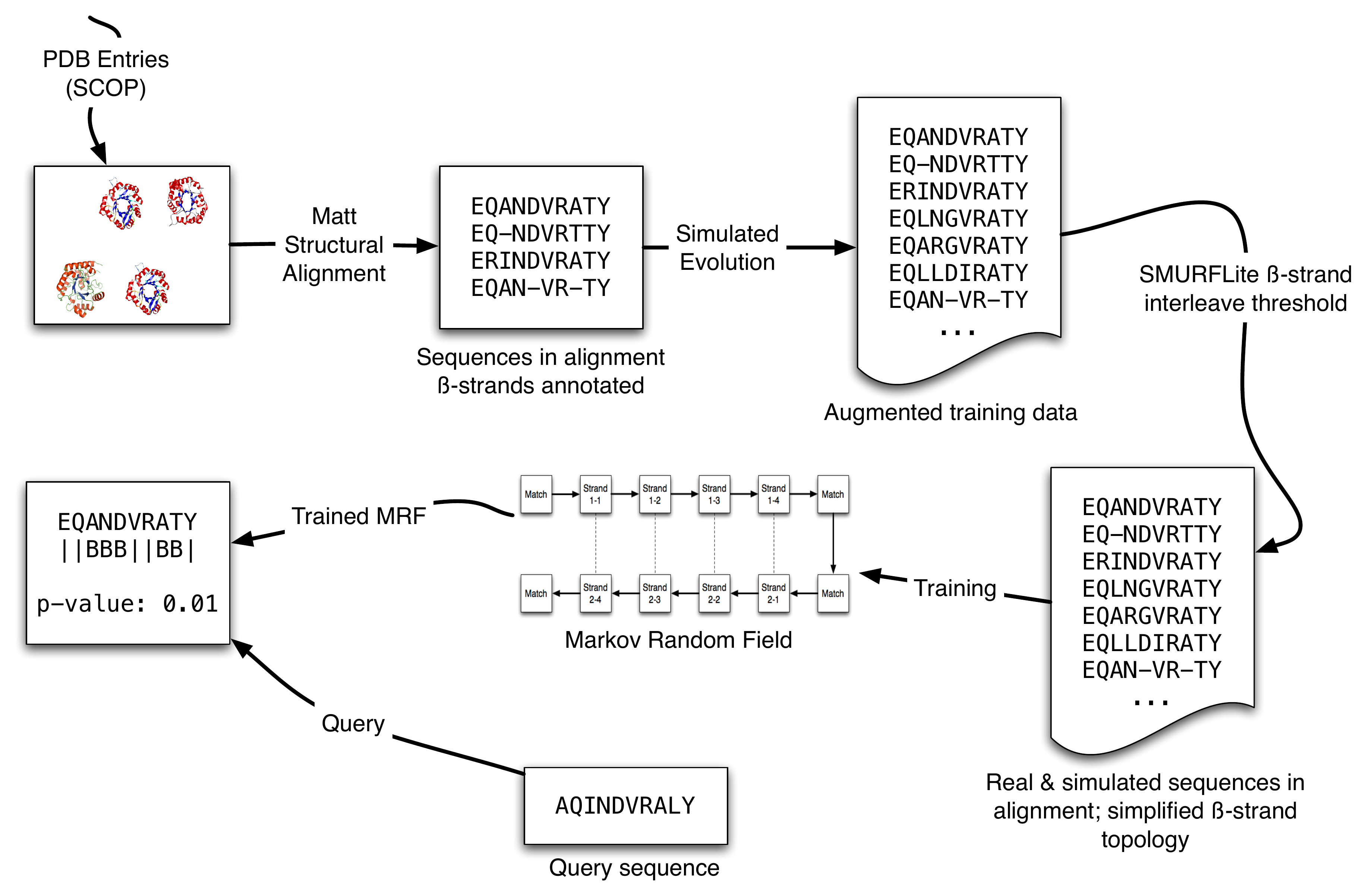}}
\caption{The SMURFLite pipeline, including simulated evolution and simplification of the $\beta$-strand topology}\label{pipeline}
\end{figure*}

\subsection{SMURFLite augmented training data}

Kumar and Cowen~\cite{Kumar:2009tp, Kumar:2010wv} showed that ``simulated
evolution,'' augmenting limited training data with additional sequences
produced by mutating the original sequences, improved the performance of HMMER
at recognizing the same-superfamily level of homology. Kumar and
Cowen~\cite{Kumar:2010wv} used two types of simulated evolution: point-wise and
pairwise. 
Here we add only pairwise mutations based on $\beta$-strand pairings, as
we expect long-range interactions between $\beta$-strands to be highly conserved
across similar structures. 
We postulated that the elimination of the $\beta$-strand pairs SMURFLite
must disregard because of computational complexity might be
compensated for by augmenting the training data with artificial
sequences based on likely mutations in those paired $\beta$-strands. This
training-data augmentation comes at insignificant runtime cost and is
done before $\beta$-strand pairs are removed from the template (but after
their interleave number has been identified, where we define
interleave number next below).
The mutation frequencies are given by the Betawrap and SMURF (\cite{Bradley:2001tj,
Menke:2010ti}) pairwise probability tables. 
Using the same algorithm as \cite{Kumar:2010wv}, we generate 150 new artificial 
training sequences from each original training sequence. 
For each artificial sequence, we mutate at a 50\% mutation rate per length of 
the $\beta$-strands. 
The parameters 150 and 50\% were recommended by \cite{Kumar:2010wv}; we also 
evaluated a 10\% mutation rate (a secondary peak according to their work) and 
performance was slightly worse (data available from the authors).

\begin{figure}[!tpb]
\centerline{\includegraphics[width=5.5in]{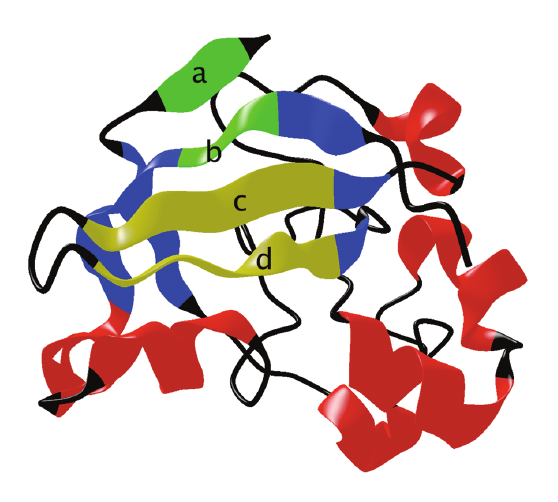}}
\caption[A closed $\beta$-barrel (PDB ID 1bw3, a Barwin domain) from the 
superfamily ``Barwin-like endoglucanases''  to illustrate interleaving of 
strand pairs.]{A closed $\beta$-barrel (PDB ID 1bw3, a Barwin domain) from the 
superfamily ``Barwin-like endoglucanases''  to illustrate interleaving of 
strand pairs. $\beta$-strands a and b, which close the barrel, have interleave 
4, while strands c and d, which are adjacent in sequence, have interleave 1. 
Strands b and c have interleave 2.
This is because, if we begin at the N-terminal end, the order of the 
$\beta$-strands is a, c, d, b.
}\label{barwin_barrel}
\end{figure}

\begin{figure}[ht!]
  \centering
  \subfigure[]{
  \includegraphics[height=2.0in]{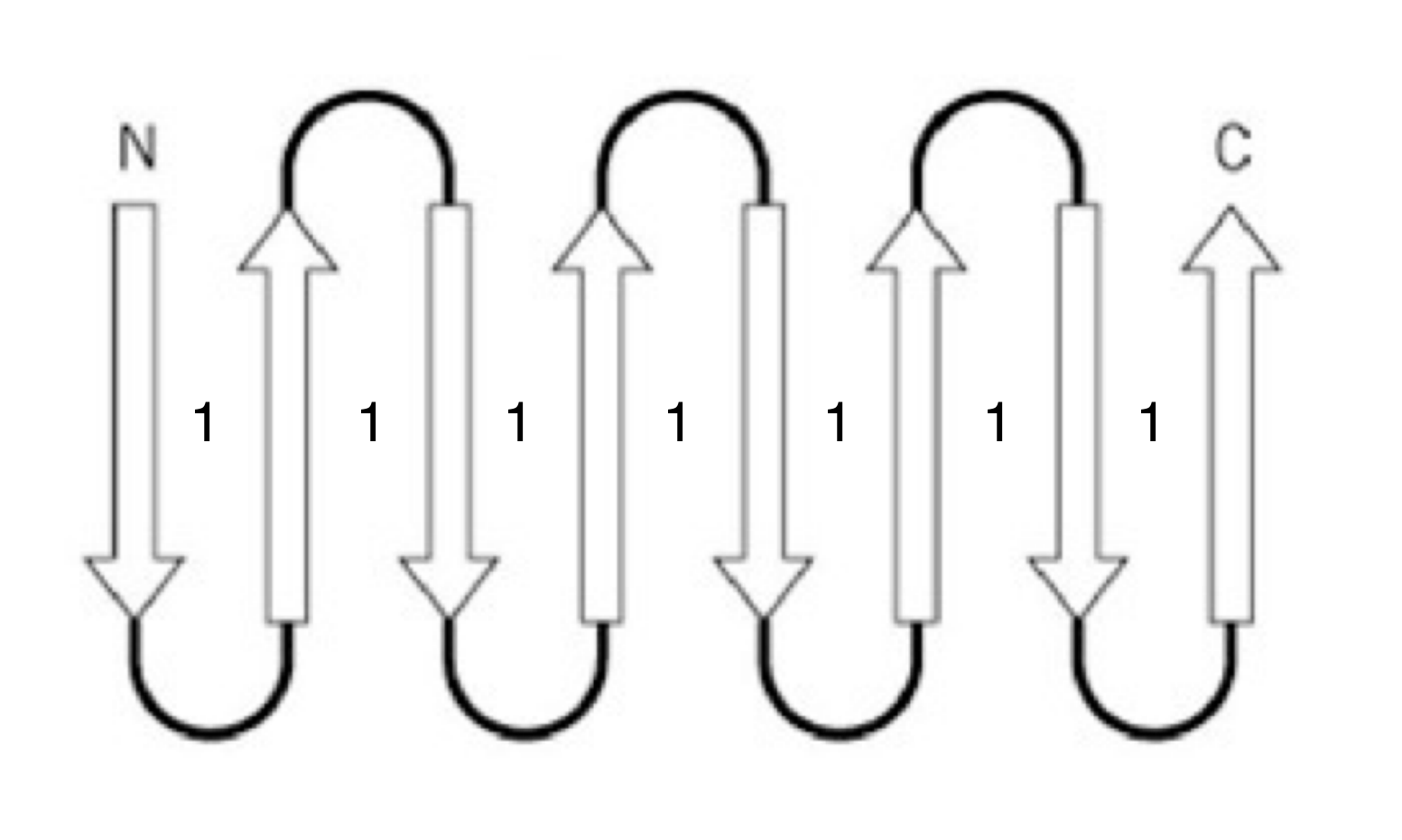}
  \label{beta-interleave-1}
  }
  \subfigure[]{
  \includegraphics[height=2.0in]{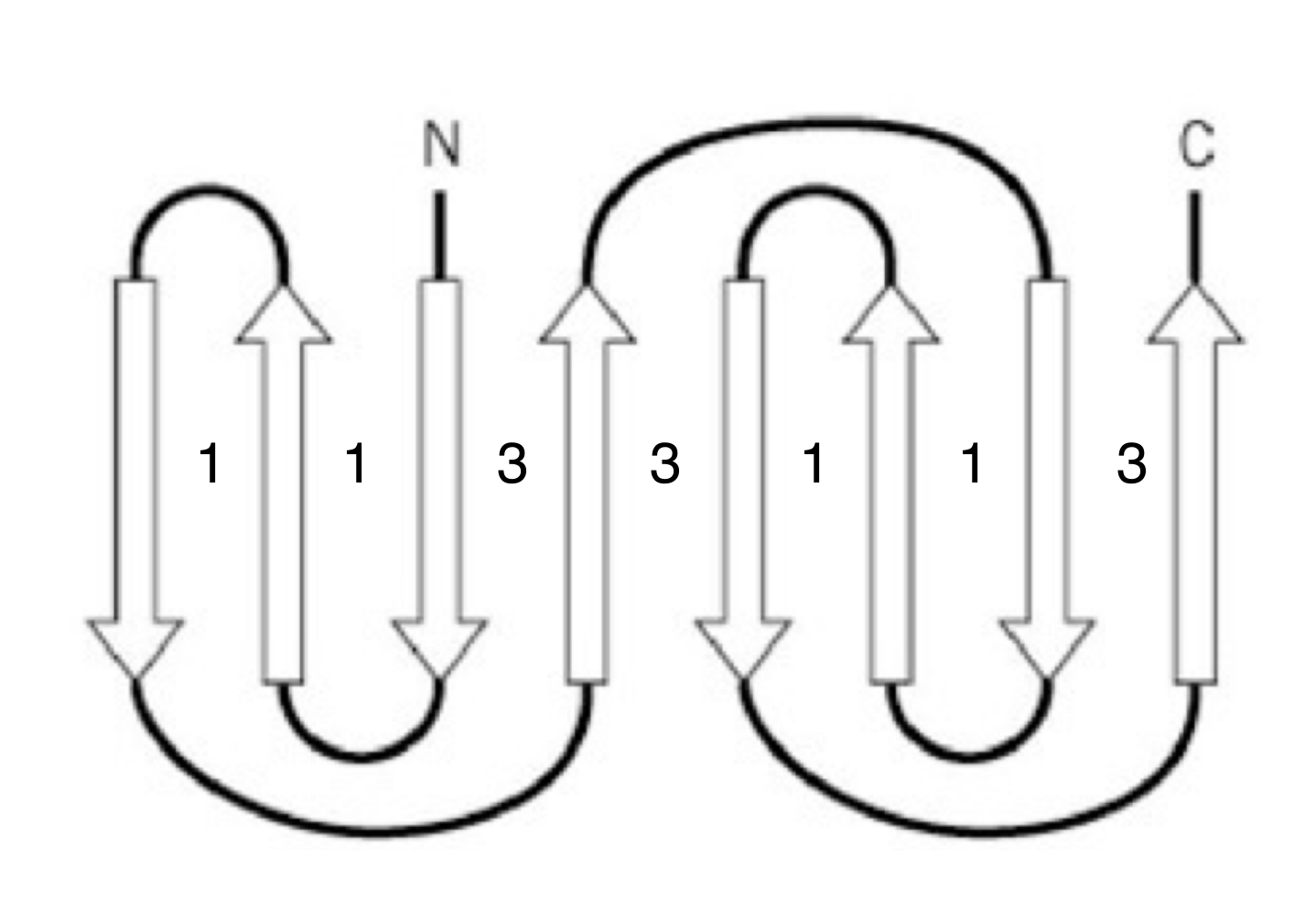}
  \label{beta-interleave-2}
  }
  \subfigure[]{
  \includegraphics[height=2.0in]{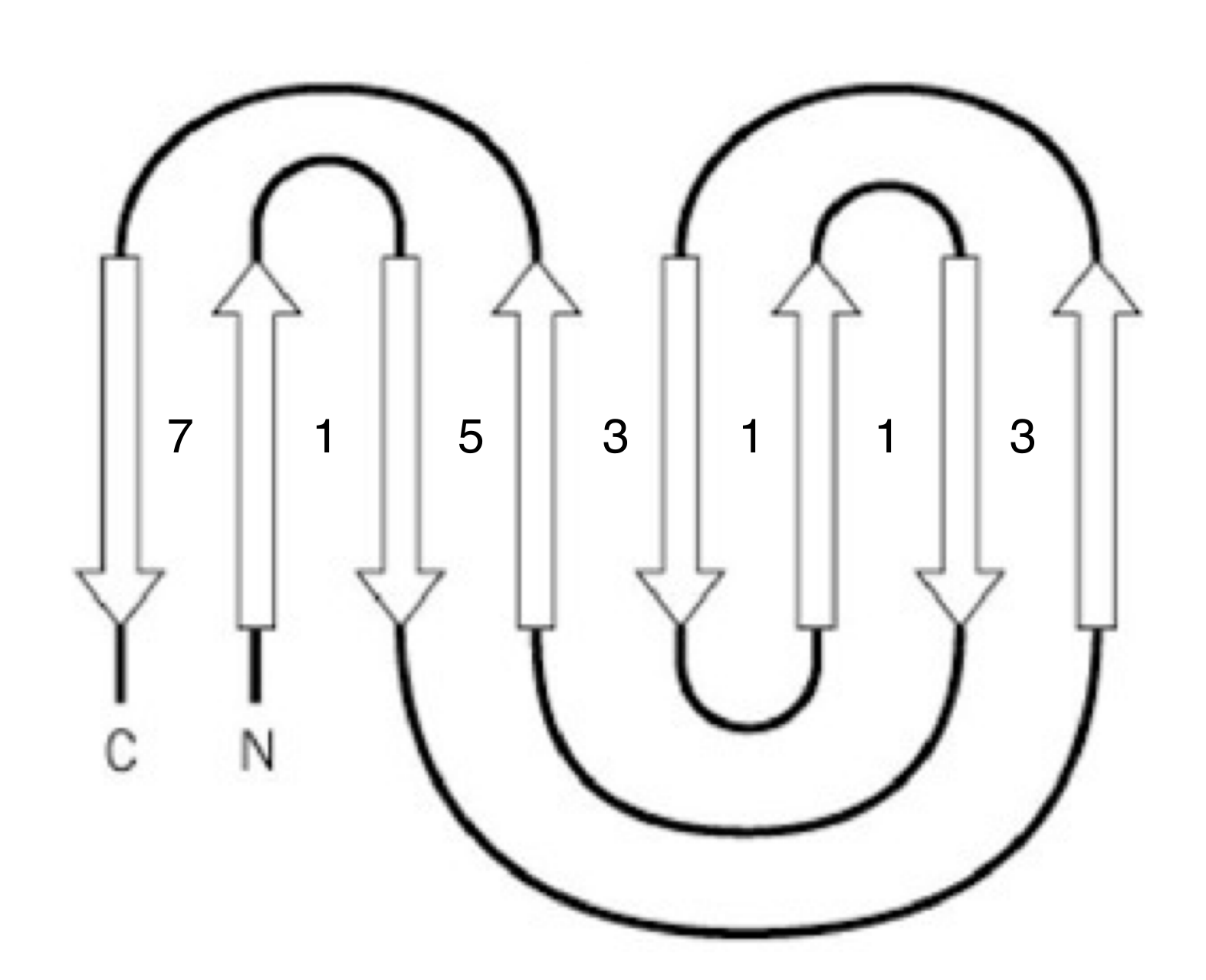}
  \label{beta-interleave-3}
  }

  \caption[Interleave number explained]{\textbf{(a)} An ``up-and-down''
   $\beta$-sheet.
  Unless the C-terminal and N-terminal ends are hydrogen-bonded together,
  the interleave is 1 for each pair of strands.\\
\textbf{(b)} A ``greek key'' $\beta$-sheet.
  The numbers between each pair of $\beta$-strands indicate the interleave.
  The maximum interleave in this instance is 3.\\
\textbf{(c)} A ``jelly roll'' $\beta$-sheet.
The numbers between each pair of $\beta$-strands indicate the interleave.
The maximum interleave in this instance is 7, between the C-terminal and
N-terminal strands.
}\label{beta-interleave}
\end{figure}

\subsection{SMURFLite simplified random field}

SMURFLite trains a MRF on a template built from a multiple structure
alignment. $\beta$-strands in the aligned set of structures are found by
the program SmurfPreparse which is part of the SMURF
package (\cite{Menke:2009, Menke:2010ti}). 
The program not only
outputs the positions of the consensus $\beta$-strands in the alignment,
it also declares a position buried or exposed based on which of the
two tables is the best fit to the amino acids that appear in that
position in the training data. 
SMURFLite then assigns an interleave
value to each $\beta$-strand pair, as follows:  Any pairwise interaction 
between $\beta$-strands whose interleave value equals or exceeds the SMURFLite
threshold is removed from the training data. 

Consider three $\beta$-strands: A, B, and C. Suppose strand A interacts with 
strand B and the
(A,B) pair has an interleave value of 4, while strand B also interacts
with strand C and that (B,C) pair has an interleave value of just
1. With a SMURFLite threshold of 2, the (A,B) pair would be discarded,
but the (B,C) pair would remain in the training data. Thus, interleave
numbers are properties of {\em pairs} of $\beta$-strands; a $\beta$-strand
may be involved in multiple pairings, each of which may have a
distinct interleave value. Discarding $\beta$-strand pairs whose
interleave value equals or exceeds the threshold guarantees that the
MRF will have no $\beta$-strand pairs greater than or equal to that
threshold, and thus bounds the computational complexity, which is
exponential in the maximum interleave value found in a training
template. 
Figure~\ref{simplified} illustrates which $\beta$-strand pairs would be removed
for two different topologies.

\begin{figure}[ht!]
  \centering
  \subfigure[]{
  \includegraphics[height=2.0in]{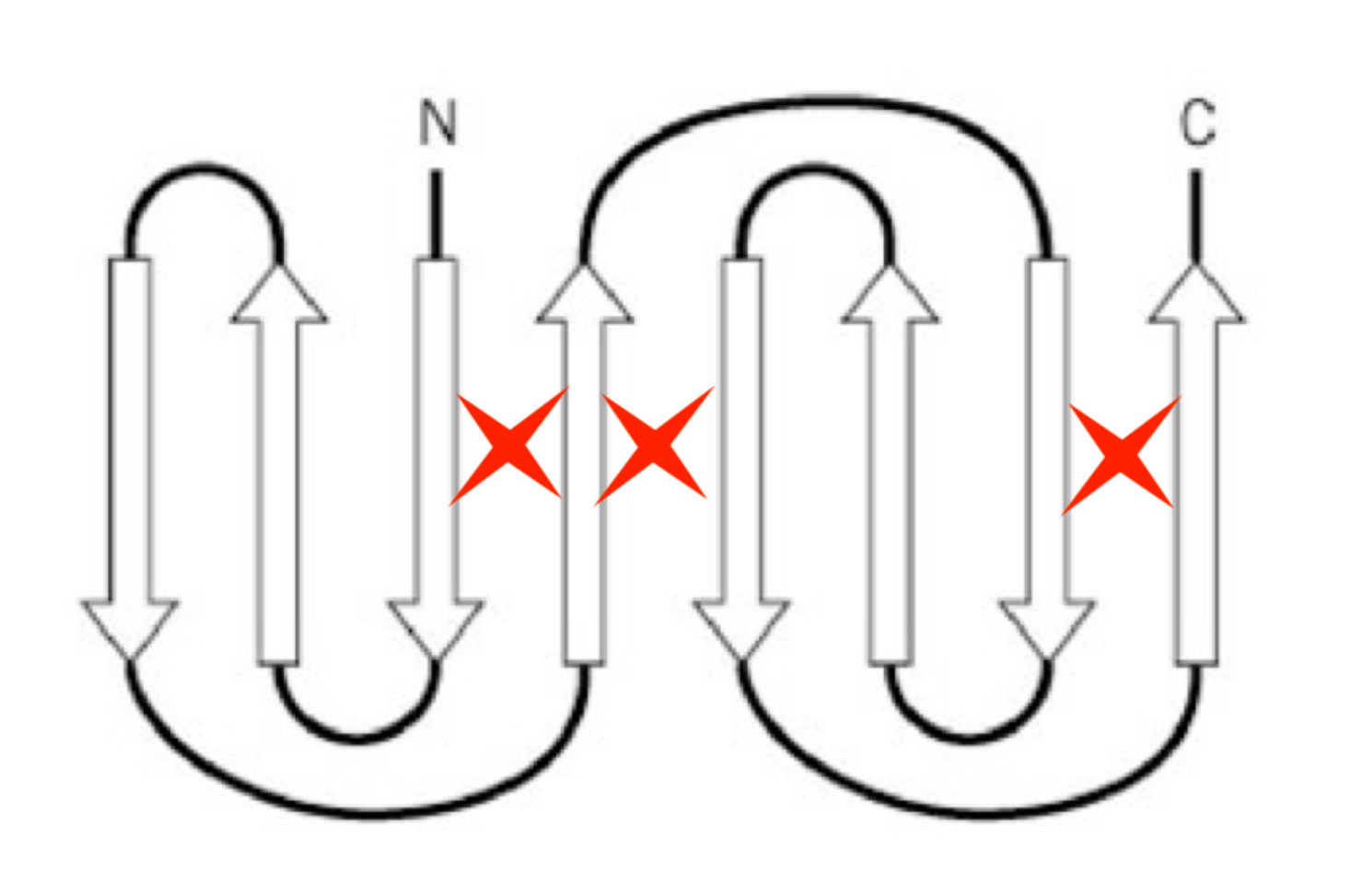}
  \label{simplified-2}
  }
  \subfigure[]{
  \includegraphics[height=2.0in]{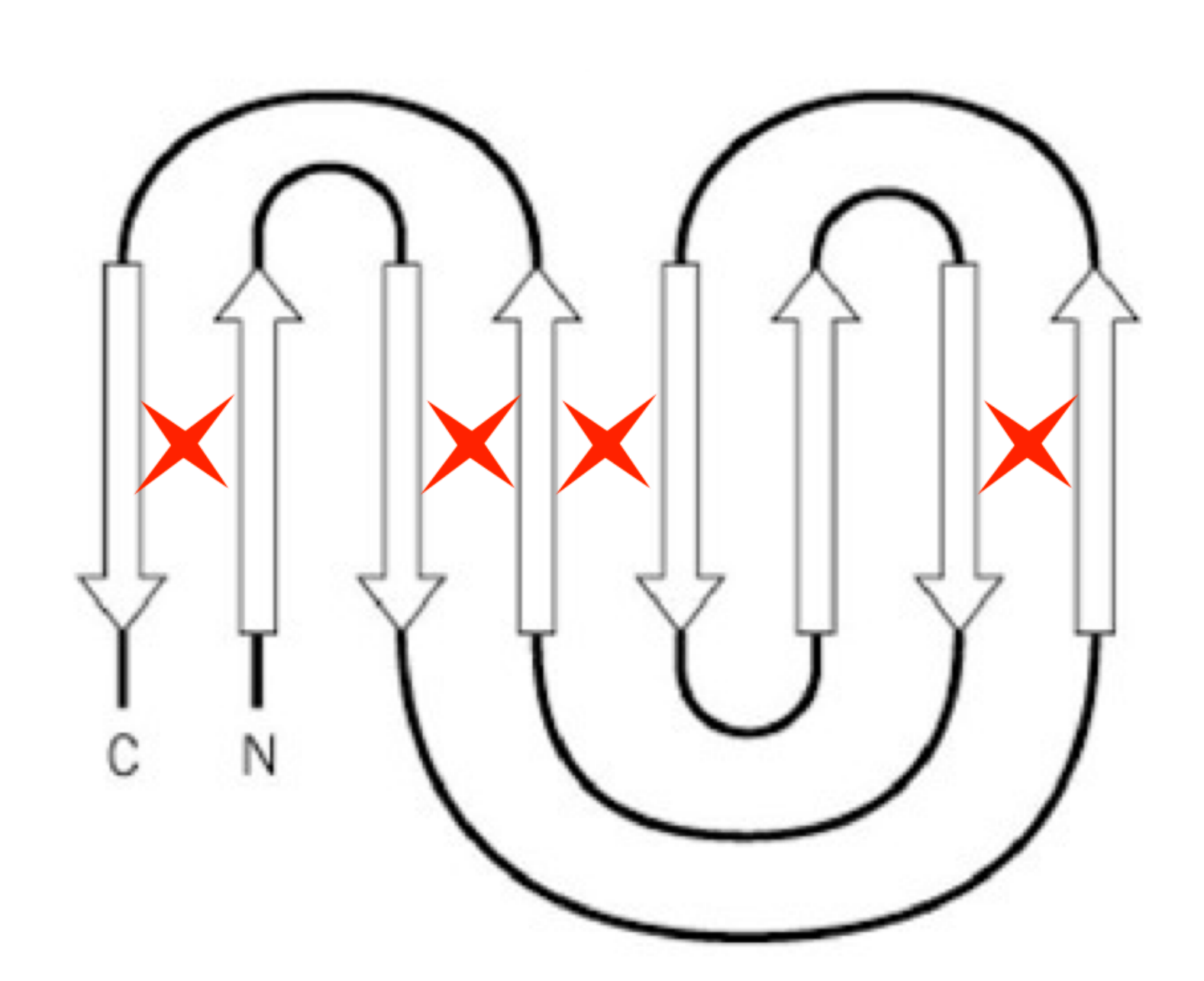}
  \label{simplified-3}
  }
  \caption[SMURFLite simplified Markov random fields]{\textbf{(a)} A ``greek 
  key'' $\beta$-sheet, indicating which
  $\beta$-strand pairs would be removed by SMURFLite with an interleave 
  threshold of 2.
\textbf{(b)} A ``jelly roll'' $\beta$-sheet, indicating which $\beta$-strand
pairs would be removed by SMURFLite with an interleave threshold of 2.
}\label{simplified}
\end{figure}

Note that SMURFLite with an interleave threshold of 0, which
will discard all $\beta$-strand pair information, is simply an HMM.

\subsection{HMMER implementation} \label{hmmer_imp}

SMURFLite was tested against HMMER version 3.0a2 with the ``--seqZ 1'' and
``--seqE 10000'' options applied to hmmsearch, and the ``--symfrac 0.2'' and
``--ere 0.7'' options applied to hmmbuild. The --seqZ 1 option ensures that
E-values are comparable regardless of the size of the sequence database, while
the --seqE 10000 option forces HMMER to return results for all query sequences.
The --symfrac 0.2 option requires that only 20\% of sequences need to be in
agreement to cause a match state in a given column (the default is 50\%). Given
the remote homology at which we were performing experiments, 50\% was an
unreasonably high threshold that led to few match states being found. This
option was also used by~\cite{Kumar:2009tp}. The --ere option sets the minimum
relative entropy per position target to 0.7 bits (the default is 0.59). Note
that HMMER versions 3.0a2 and 3.0 both use SAM sequence entropy
(\cite{Karplus:1998ub}) by default. This entropy weighting scheme has been
shown to be superior for remote homology detection tasks (\cite{Kumar:2009tp,
Johnson:2006tt}).

HMMER 3.0a2 was used despite having been superseded by version 3.0, because it
uniformly performs better on this task. This is because version 3.0 contains
computational optimizations that cause it to reject a sequence (with no score
provided) quickly if it does not appear to align well. These optimizations,
however, cause nearly all query sequences outside the family level of homology
to fail and return no score, with the result that HMMER version 3.0 never
surpasses an AUC of 0.5.

\subsection{RAPTOR implementation} \label{raptor_imp}

SMURFLite was tested against RAPTOR, which was run with the options ``-a nc''
indicating that the default threading algorithm described in the RAPTOR paper
(\cite{Xu:2003p3417}) was used. In addition, RAPTOR used the weighting
parameters ``weightMutation = 1.4009760151,'' ``weightSingleton = 1,''
``weightLoopGap = 16.841836238,'' ``weightPair = 0,'' ``weightGapPenalty = 1,''
``weightSStruct = 3.0137849223.'' RAPTOR uses both sequence and structural
features, and these options represent the recommended balance of these features
(\cite{Xu:2003p3417}).

\subsection{HHPred implementation} \label{hhpred_imp}
SMURFLite was tested against HHPred version 1.5.1. HHPred HMMs for each SCOP 
family were downloaded from the HHPred web site, and queried using hhsearch. 
The score of the best-scoring family HMM within each superfamily was used in 
computing ROC curves.

\subsection{Whole-genome search} \label{whole genome}

All 1852 protein sequences from \textit{Thermotoga maritima} were queried 
against $\beta$-structural templates constructed from the 
nr-PDB~(\cite{Berman:2000hl}) with non-redundancy determined by an E-value of 
$10^{-7}$, organized according to those 207 $\beta$-structural superfamilies 
from 
SCOP that were able to be aligned using the Matt structural alignment program, 
using SMURFLite with an interleave threshold of 2 and simulated evolution 
mutation rate of 50\% on the residues that participate in $\beta$-strands. 
We computed $p$-values and alignments for all $1852 \times 207$ possible hits.

\section{Results}

\subsection{SMURFLite Validation}

SMURFLite's ability to recognize $\beta$-propellers and barrels was compared to 
HMMER (\cite{Eddy:1998ut}),
RAPTOR (\cite{Xu:2003p3417}), and HHPred (\cite{Soding:2005fa}) in a stringent 
cross-validation experiment, as explained in Section~\ref{datasets}.

SMURFLite was tested on these 5 propeller folds and 11 barrel
superfamilies, with {\em interleave\/} thresholds of 1, 2, and 3, and
with and without simulated evolution on the
$\beta$-strands (\cite{Kumar:2010wv}). Here the interleave threshold is a
parameter of SMURFLite that trades off the computational complexity
with the ability of the MRF to capture complicated long-range
dependencies.


The balance between accuracy and computational efficiency is
determined by the interleave threshold at which SMURFLite is run. In
particular, we found that SMURFLite set to an interleave threshold of
3 or less was always fast. Thus, our first question is how SMURFLite
with and without simulated evolution performs on our test set when
the interleave threshold is set to 3 or less. We found that SMURFLite
became slow at an interleave threshold of 4, and essentially
intractable at an interleave threshold of 5 or above. While SMURFLite
with an interleave threshold of 1 or 2 requires roughly 1 second of
wall-clock time on a 12-core 2.4GHz AMD Opteron server, an interleave
threshold of 4 raises this run-time requirement to 7-10
minutes. Restricting the interleave threshold to 3 or less has
different impacts on the different folds in our test set. In
particular, the $\beta$-strands in the propeller folds never have an
interleave greater than 3, which means that full SMURF, as we know, is
tractable on these folds. However, we were still interested in how
simplifying the random field to an interleave of 2 or 1 would impact
performance, and also whether simulated evolution would help. In
contrast, the barrel superfamilies in our test set contain a maximum
$\beta$-strand interleave of between 4 and 8. Interestingly, none of
these barrels contained any $\beta$-strands with an interleave of 3 in
the consensus Matt (\cite{Menke:2008wu}) alignment, so our restriction of 
running SMURFLite
with an interleave threshold of 3 or less is equivalent, on the
barrels, to running SMURFLite with an interleave threshold of 2.
In other words, running the interleave-threshold filter at a threshold of
3 produced identical training data to running it at a threshold of 2.

\begin{figure*}[ht!]
\centering
\subfigure[Full structure of a 7-bladed $\beta$-propeller]{
\includegraphics[width=3.2in]{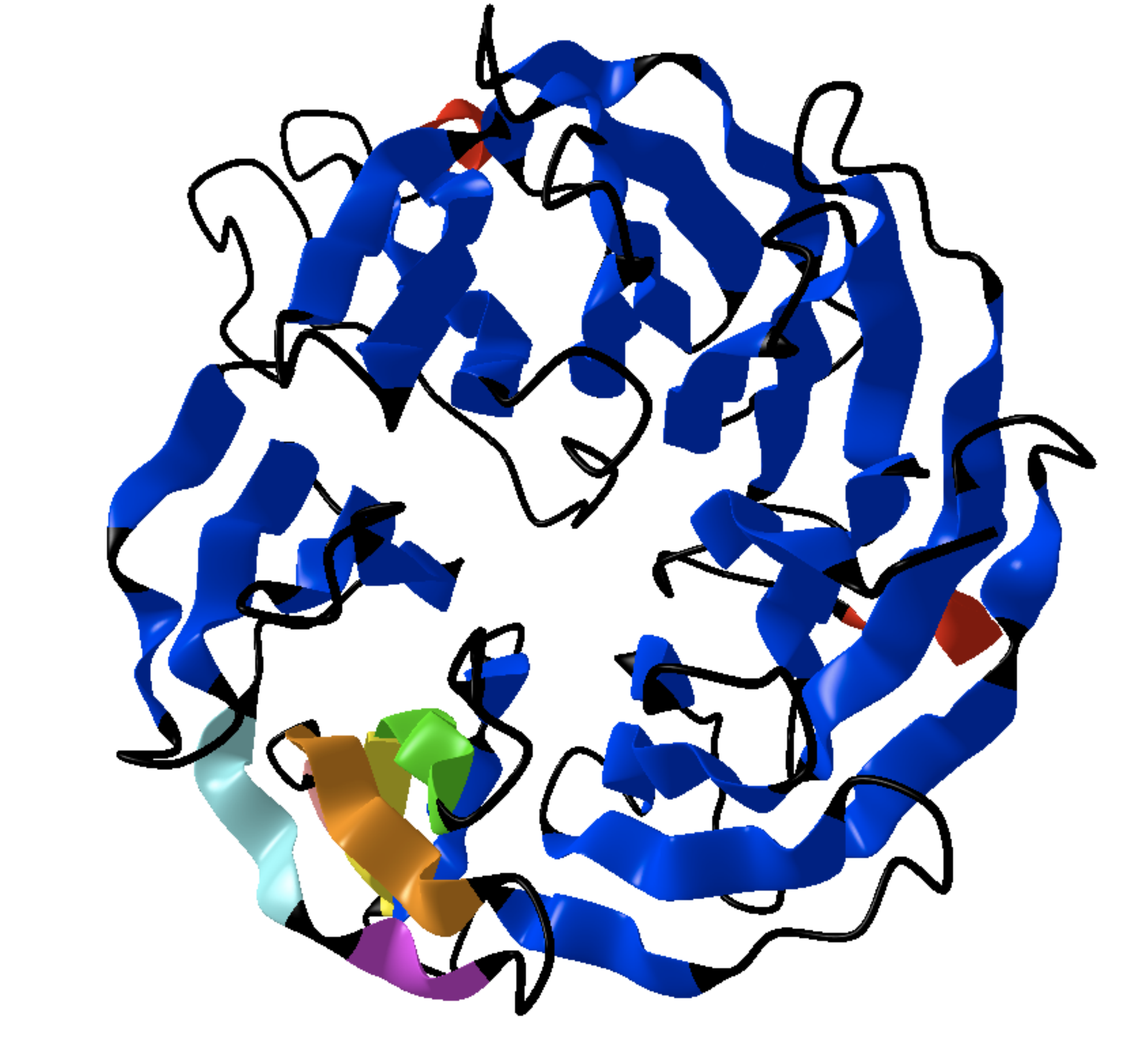}
\label{fig:subfig1}
}
\subfigure[The most complicated propeller blades have an interleave of 2. 
Detail of one blade from the structure above, with individual $\beta$-strands 
labeled a through g in sequential order. 
The interleave values are as follows: (a,c): 2; 
(b,c): 1; (d,e): 1; (f,c): 3; (f,g): 1.]{
\includegraphics[width=3.2in]{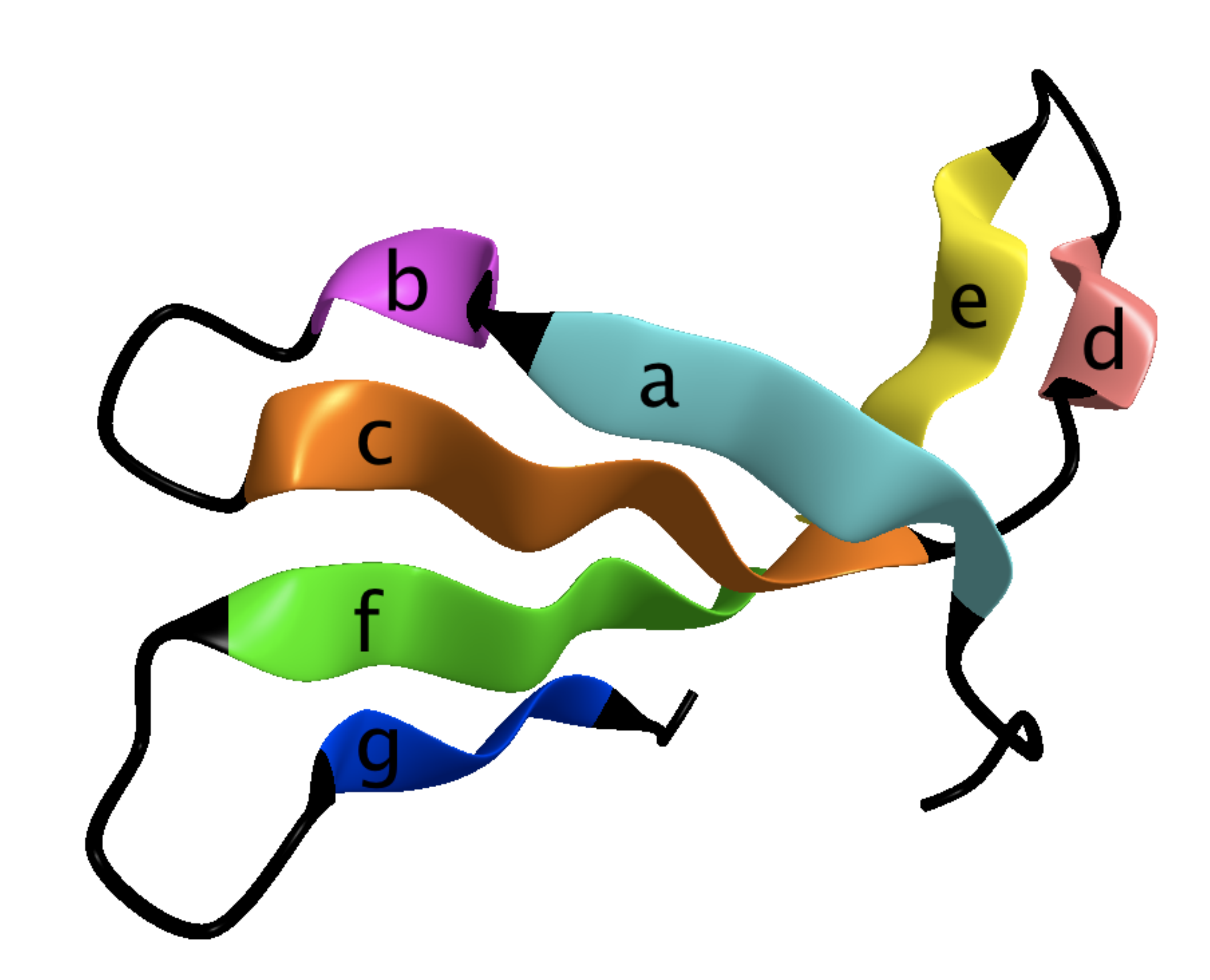}
\label{fig:subfig2}
}
\caption[$\beta$-propeller detail]
{A 7-bladed $\beta$-propeller, ``Quinohemoprotein 
amine dehydrogenase'' B chain from {\em Paracoccus denitrificans.} }
\label{propellers}
\end{figure*}

SMURFLite with interleave threshold 2 and simulated evolution performs
well on all propeller folds, with AUCs between 0.89 and 0.99. It
always performs better than HMMER, and better than RAPTOR and HHPred
except on the 7-bladed propellers (of which there are 39 non-redundant solved 
structures in 19 SCOP families), where HHPred achieves an AUC of
0.99 and RAPTOR achieves an AUC of 0.95 versus an AUC of 0.93 for
SMURFLite with interleave threshold 2 and no simulated evolution (see 
Table~\ref{propellertable}). Interestingly, on
the 5-bladed propellers (of which there are only 14 non-redundant solved 
structures in 7 SCOP families), adding simulated evolution seems to greatly
improve performance; even SMURFLite with an interleave threshold of 2
with simulated evolution outperforms full-fledged SMURF. While these results 
focus on the accuracy of the MRF score for the remote homolog decision problem, 
as opposed to the question of alignment quality, we note that SMURFLite with an 
interleave threshold of 1 or 2 produces highly similar alignments to full 
SMURF, particularly with respect to placing the ``blades'' of the 6-, 7-, and 
8-bladed propellers.

For all 11 $\beta$-barrel superfamilies, there is a maximum interleave
number that ranges from 4 (as in the ``Sm-like ribonucleoproteins'')
to 8 (as in the ``Cyclophilin-like'' superfamily). We find that for 6
of the 11 $\beta$-barrel superfamilies, SMURFLite with an interleave threshold 
of 2 and simulated evolution outperforms HMMER, RAPTOR, and HHPred. 
For two of the remaining superfamilies, HMMER performs best, for two of the
remaining superfamilies, RAPTOR performs best, and for one superfamily,
HHPred performs best (see Table~\ref{barreltable}).

As discussed above, SMURFLite begins to test the limits of
computational tractability when interleave numbers of 4 are
allowed. Since many barrel structures had $\beta$-strand pairs with
interleaves of 4, we wished to test if incorporating these more
long-range pairwise dependencies into our MRF would improve
performance. Some barrel superfamilies on which we tested have only
strand pairs of interleave 1 or 2, excepting a pair of $\beta$-strands
that close the barrel and thus have an interleave equivalent to the
number of strands in the barrel. Certainly, including that last strand
is beyond the computational power of SMURFLite. Other barrels, whether
open or closed, have more complex strand topology and interleaves of 3
or 4 are common even in the middle of the barrels. We chose to run
SMURFLite with an interleave of 4 on one of the barrel superfamilies
of moderately complex topology, the ``Barwin-like endoglucanase''
superfamily, of which an example appears in
Figure~\ref{barwin_barrel}. The ``Barwin-like endoglucanase''
superfamily contains ``Barwin,'' a protein that may be involved in a
common defense mechanism in plants (\cite{Svensson:1992}).

On the ``Barwin-like endoglucanase'' superfamily, we find an enormous
improvement in performance from capturing that last strand pair, with
AUC improving from 0.63 for SMURFLite with an interleave threshold of
2 and simulated evolution, to 0.94 for SMURFLite with an interleave
threshold of 4 and simulated evolution (see
Figure~\ref{barwinplot}). Note that both HMMER and RAPTOR fail
entirely on this superfamily, achieving an AUC of less than 0.5.

\rowcolors{2}{gray!25}{white}

\newcolumntype{H}[1] {%
>{\raggedright}%
p{#1}}

\begin{small}

\begin{center}
\begin{table*}[!t]
\caption{AUC on $\beta$-Propeller folds\label{propellertable}}
{\begin{tabular*}{\textwidth}{@{\extracolsep{\fill}}H{1.5cm}p{1.1cm}p{1.3cm}p{1.0cm}p{0.7cm}p{0.7cm}p{0.7cm}p{0.7cm}p{0.7cm}p{0.7cm}}\hline
 & HMMER & RAPTOR & HHPred & SL1 & SL1E & SL2 & SL2E & SL3 & SL3E\\
 \hline
5-bladed & - & - & - & 0.75 & {\bf 0.89} & 0.73 & {\bf 0.89} & 0.73 & {\bf 0.89}\\
6-bladed & 0.82 & 0.82 & 0.88 & 0.92 & 0.93 & {\bf 0.96} & 0.95 & {\bf 0.96} & {\bf 0.96}\\
7-bladed & 0.89 & 0.95 & {\bf 0.99} & 0.92 & 0.91 & 0.93 & 0.91 & 0.93 & 0.91\\
8-bladed & - & 0.64 & {\bf 0.99} & {\bf 0.99} & {\bf 0.99} & {\bf 0.99} & {\bf 0.99} & {\bf 0.99} & {\bf 0.99}\\
\hline
\end{tabular*}}\\{Note: for SMURFLite, the number (1,2,3) indicates the 
interleave threshold, and SEv is simulated evolution. A dash ('-') in a 
result entry indicates the method failed on these structures, i.e. an AUC of 
less than 0.6.
For issues of space, we abbreviate the SMURFLite entries. For example, SL1
indicates SMURFLite with an interleave threshold of 1, while SL3E indicates
SMURFLite with an interleave threshold of 3, augmented by simulated evolution.}
\end{table*}
\end{center}
\end{small}

\rowcolors{2}{gray!25}{white}

\begin{small}

\begin{center}
\begin{table*}[!t]
\caption{AUC on $\beta$-Barrel superfamilies\label{barreltable}}
\footnotesize{
\begin{tabular*} {\textwidth}{@{\extracolsep{\fill}}H{3.2cm}p{1.1cm}p{1.2cm}p{1.0cm}p{1.1cm}p{1.1cm}p{1.1cm}p{1.1cm}}\hline
 & HMMER & RAPTOR & HHPred & SMURF\-Lite~1 & SMURF\-Lite 1, SimEv & SMURF\-Lite~2 & SMURF\-Lite 2, SimEv\\
 \hline 
{\bf SMURFLite performs best}  & & & & & & \\
\hline
Translation proteins & - & - & 0.66 &{\bf 0.93} & 0.92 & {\bf 0.93} & {\bf 0.93}\\
Barwin-like endoglucanases & - & - & 0.75 & - & {\bf 0.77} & - & 0.63\\
Cyclophilin-like & 0.67 & 0.61 & 0.7 & 0.82 & {\bf 0.85} & 0.82 & 0.83\\
Sm-like ribonucleoproteins & 0.73 & 0.71 & 0.77 & 0.76 & 0.71 & 0.76 & {\bf 0.85}\\
Prokaryotic SH3-related domain & 0.81 & -  & - & {\bf 0.83} & 0.82 & {\bf 0.83} & {\bf 0.83}\\
Tudor/PWWP/MBT & 0.78 & 0.74 & 0.67 & 0.83 & 0.77 & {\bf 0.83} & 0.79\\
Nucleic acid-binding proteins & 0.75 & - & 0.67 & 0.76 & 0.89 & 0.76 & 0.92\\
 \hline
{\bf HHPred performs best} & & & & & & & \\
\hline
Translation proteins SH3-like & 0.83 & 0.81 & {\bf 0.86} & 0.62 & - & 0.62 & -\\
 \hline
{\bf RAPTOR performs best}  & & & & & & &\\
\hline
PDZ domain-like & 0.96 & {\bf 1.0} & 0.99 & 0.97 & 0.97 & 0.97 & 0.97\\
FMN-binding split barrel & 0.62 & {\bf 0.82} & 0.61 & - & - & - & -\\
 \hline
{\bf HMMER performs best}  & & & & & & &\\
\hline
Electron Transport accessory proteins & {\bf 0.84} & - & 0.77 & 0.63 & - & 0.63 & 0.66\\
\hline
\end{tabular*}}\\{Note: for SMURFLite, the number (1,2) indicates the interleave threshold, and SimEv is simulated evolution. A dash ('-') in a result entry indicates the method failed on these structures, i.e. an AUC of less than 0.6}
\end{table*}
\end{center}

\end{small}

\rowcolors{2}{white}{white}

\begin{figure}[h!tp]
\centerline{\includegraphics[width=9cm]{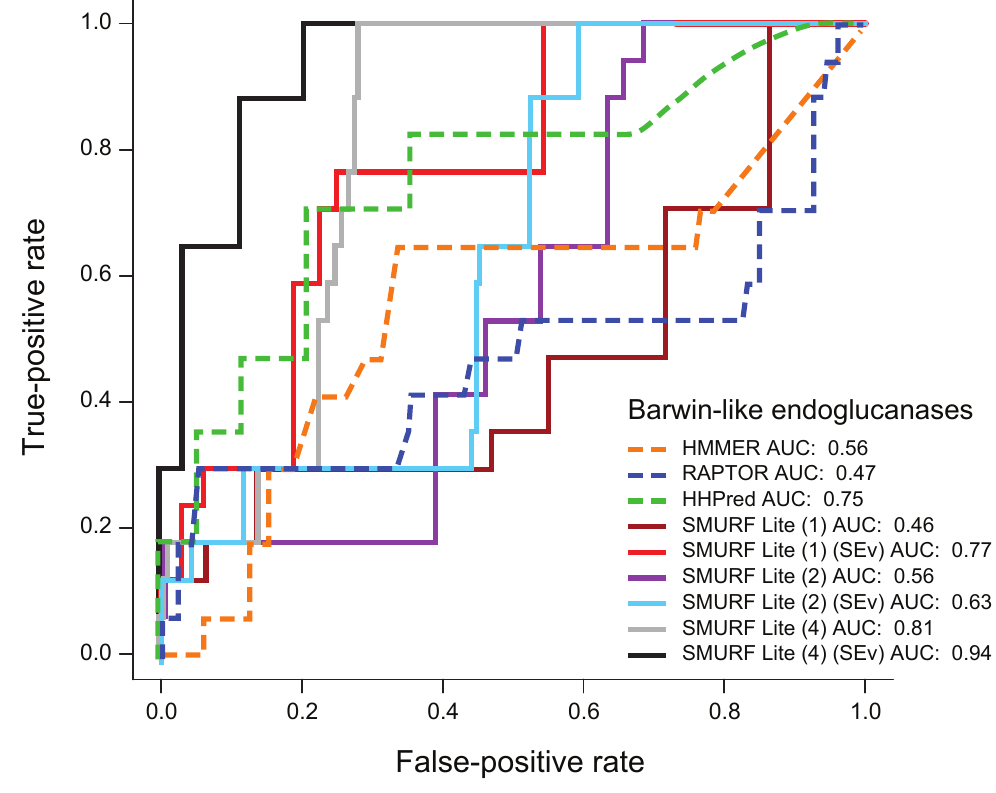}}
\caption[Performance of SMURFLite compared to other methods on the
``Barwin-like endoglucanases'' $\beta$-barrel superfamily according to the AUC
(Area Under Curve) measure.]{Performance of SMURFLite compared to other methods 
on the
``Barwin-like endoglucanases'' $\beta$-barrel superfamily according to the AUC
(Area Under Curve) measure. For SMURFLite, the number (1,2,4) indicates the
interleave threshold (indicating which strand pairs in the barrel participate
in the MRF; note that interleave 3 is omitted since it is identical to
interleave 2 for this fold), and SimEv indicates that simulated evolution was
also performed on the $\beta$-strands in the training data. As the interleave
threshold increases and the MRF becomes more powerful, performance tends to
improve. Including simulated evolution also improves
performance.}\label{barwinplot}
\end{figure}

\subsection{SMURFLite on Whole Genomes}

We considered all 1852 genes from the bacterium \textit{Thermotoga maritima}, a thermophilic organism that bears some similarity to Archaea and whose cell is wrapped in an outer membrane, or ``toga''~(\cite{Huber:1986ws}).
Out of 354 total superfamilies within the SCOP class ``All beta proteins'', 
288 (81\%) of which contain at least two protein chains, 207 superfamilies (71\%) were structurally consistent enough to be aligned using the Matt~(\cite{Menke:2008wu}) structural alignment program.  We built SMURFLite templates for these 207 superfamilies, and obtained from UniProt the protein sequences for each of 1852 genes.
We predict 139 of the 1852 genes from \textit{Thermotoga maritima} to
belong to one of the 207 $\beta$-structural SCOP superfamilies we
consider, with a $p$-value of less than 0.005. Of the 139 genes about which we
make predictions, 28 already have solved structures in the PDB,
however, since there is a substantial time lag before new PDB
structures are assigned to SCOP, only one of those structures was
already given a SCOP assignment (and thus only one of these 28
structures potentially informed SMURFLite training). Thus, determining
the correct SCOP assignments of the remaining 27 (an easy computational
problem given full structural information) allows us to estimate the
accuracy of SMURFLite predictions on these structures. Using the 
Matt~(\cite{Menke:2008wu}) structural
alignment program and the methodology of~(\cite{Daniels:2011dc}), we
computed SCOP superfamilies for all 27, and in 100\% of the cases,
SMURFLite's predictions matched the structural alignments and hence
SCOP superfamily assignments. We now survey the remaining 111
structures on which SMURFLite makes predictions, for which structural
information is not yet available. 8 of these 111 structures also had
their SCOP superfamilies predicted in the study of \cite{Zhang:2009ku}
and in all 8 cases, our predictions are in agreement with the other
study. We note that for most of these 111 structures, not only is there no
solved structure, but there is also no close homology to
proteins of solved structure. In particular, none have
BLAST hits in UniProt with solved structure and greater than 80\%
sequence identity, 18 have BLAST hits in UniProt with solved structure
and between 30\% and 80\% sequence identity, and 4 have BLAST hits in
UniProt with solved structure and less than 20\% sequence identity. As
an example, the gene Q9X087 shares only 20\% sequence identity with
its closest structurally-solved BLAST hit (Rhoptry protein from
\textit{Plasmodium yoelii yoelii}, which forms an $\alpha$-helical structure) but we predict it to belong in the
``beta-Galactosidase/glucuronidase domain'' SCOP superfamily with a
$p$-value of 0.0006.

All models predicted can be found at \url{http://smurf.cs.tufts.edu/smurflite/}

\section{Discussion}

We have presented SMURFLite, a method that combines long-range pairwise
$\beta$-strand interactions via a simplified Markov random field with simulated
evolution, a method that augments training data to capture pairwise
$\beta$-strand interactions as well. SMURFLite in most cases performs
considerably better than HMMER and RAPTOR;
however, we examine those structures for which this is not so. We postulate
that RAPTOR performs best in the case when there is significant structural
conservation across families, whereas HMMER excels when there is a small but
highly conserved sequence signature in members of a superfamily. In all four 
$\beta$-barrel superfamilies on which RAPTOR achieves an AUC of less than 0.5, 
we see considerable structural variation in the protein backbones within each 
superfamily, according to the metric discussed in
Chapter~\ref{chapter:c2_touring}, as compared to the other barrel 
superfamilies. In contrast, the barrels on which RAPTOR performed best 
exhibited little structural variation. 

The cases in which SMURFLite performs poorly exhibit an interesting property: 
the structural alignment of the protein chains used in the training set 
preserves few, or sometimes none, of the $\beta$-strands as ``consensus'' 
$\beta$-strands. 
When a significant number of $\beta$-strands are missing in this manner from 
the training data, SMURFLite exhibits poor specificity, scoring some 
non-homologous sequences comparably to homologous ones. 
The ``Translation Proteins SH3-Like Domain,'' a superfamily in which HMMER 
significantly outperforms SMURFLite, is one in which the consensus alignment 
obtained from Matt retains zero $\beta$-strands, even though each individual
structure has four strands. Thus, SMURFLite behaves like HMMER, except without
HMMER's heuristic for quickly failing bad alignments, leading SMURFLite to
report more false positives. 

The very premise of SMURFLite rests on the
conservation of $\beta$-strands, and this finding emphasizes the importance of
evolutionarily faithful structural alignments.
In future work, we will also consider alternative structural aligners, such as 
TMalign~(\cite{Zhang:2005do}), in cases where they produce alignments that 
better conserve secondary structure.

We also compared SMURFLite to HHPred, though in a sense this is not an 
entirely fair comparison, because HHPred uses \textit{all} of protein 
sequence space to build profiles for training; thus it can leverage a much 
larger training set than HMMER, RAPTOR, or SMURF or SMURFLite. Thus it is 
somewhat surprising that SMURFLite outperforms HHPred in median AUC on the 
propellers and barrels. We expect HHPred to excel in particular on 
superfamilies and folds with a high HHPred NEFF~(\cite{Soding:2005ff}), where 
NEFF is the ``number of effective families'' available for training the HHPred 
HMM.
NEFF is a measure of the information-theoretic entropy among a set of sequences;
the greater the sequence diversity of such a set, the greater the NEFF.

In contrast, simulated evolution seems to help SMURFLite most on those
structural motifs where the HHPred NEFF is lowest; i.e. it can generate diverse
training data when diverse training data is lacking. A profile version of
SMURFLite would be close in spirit to HHPred, and based on the previous
discussions we would expect profiles might improve performance; this will be a
subject for future investigation. We observed that simulated evolution either
improves or does not affect AUC for $\beta$-barrel superfamilies and
$\beta$-propeller folds with a HHPred NEFF of 20 or lower. The only cases in
which we observed simulated evolution decreasing AUC were those cases where the
NEFF was greater than 20.

While the intent of using simulated evolution in conjunction with simplified
MRFs is to compensate for the removal of highly-interleaved $\beta$-strand
pairs required for computational feasibility, we find that simulated
evolution can still improve full-fledged SMURF in cases of sparse training
data. For instance, the 5-bladed $\beta$-propellers have only three 
superfamilies in SCOP, two of which contain only one family.  We find that for 
the 5-bladed $\beta$-propeller fold, combining SMURF and simulated evolution 
improves AUC from 0.73 for full SMURF alone to 0.89. 

It is worth noting that simulated evolution on a simple \emph{pointwise}
basis, as implemented by Kumar and Cowen~\cite{Kumar:2009tp}, could likely be
incorporated into the hidden Markov itself as a set of Dirichlet mixture priors.
However, it is not clear how the \emph{pairwise} model could be incorporated.
In addition, we determine $\beta$-strand paired residues on the \emph{full}
Markov random field, before removing any pairing information.
Thus, in this case, simulated evolution may be mitigating the loss of this
$\beta$-strand pairing information.

We have demonstrated that SMURFLite is a powerful MRF methodology for 
$\beta$-structural motif recognition that is computationally tractable enough 
to scale to whole genomes, requiring approximately three hours to scan the 
\textit{Thermotoga maritima} genome on a small compute cluster. 
We have also shown that increasing the interleave number for SMURFLite can have 
dramatic effects on performance, but at a great computational cost. Methods 
that allow us to retain all $\beta$-strand pairs, such as loopy belief 
propagation\cite{Pearl:1988wz} or stochastic search, merit investigation. 
As our dependency graph 
is not a tree, loopy belief propagation may present difficulties with 
convergence and inexact inference. 
Nonetheless, looking at heuristic methods (\cite{Smyth:1997ty, Weiss:1999uu}) 
that approximately compute the SMURF score more efficiently may add even more 
power to our approach in practice.

\chapter{Protein Remote Homology Detection Using Markov Random Fields and
Stochastic Search}

\label{chapter:c4_mrfy}

\section{Introduction}

In Chapter~\ref{chapter:c3_smurflite}, we explored a method for simplifying the
computational complexity of the Markov random field, as well as an approach,
called ``simulated evolution,'' for mitigating the loss in accuracy resulting
from this simplification.
We showed that this approach, called SMURFLite, outperformed several existing
methods at remote homology detection in $\beta$-barrels and propellers.

In this work, we demonstrate another approach for computing the SMURF energy 
function for remote homology detection.
Building upon the $\beta$-structural Markov random field templates from SMURF 
and
SMURFLite, we demonstrate a method for remote homology detection that does not
discard any $\beta$-structural information, and yet remains computationally 
tractable on any protein structure.

We have developed MRFy, an algorithm that relies on stochastic search to find
a near-optimal parse of a query sequence onto the SMURF Markov random field.
We also provide an implementation of MRFy, written in the Haskell functional
programming language; this implementation is discussed in our ``experience 
report'' on computational biology software in 
Haskell~\cite{Daniels:2012cm}, which is not part of this dissertation.

We test MRFy on the same set of barrel folds in the mainly-$\beta$
class of the SCOP hierarchy as was used to test SMURFLite in 
Chapter~\ref{chapter:c3_smurflite}, in stringent cross-validation
experiments. 
We show a mean 0.4\% (median 1.7\%) improvement in Area Under 
Curve (AUC) for
$\beta$-structural motif recognition as compared to the SMURFLite results in
Chapter~\ref{chapter:c3_smurflite} and~\cite{Daniels:2012dg}.
By these same benchmarks, we show a mean 5.5\% (median 16\%) improvement over
HMMER (\cite{Eddy:1998ut}) (a popular HMM method), a mean 29\% (median
16\%) improvement as compared to RAPTOR (\cite{Xu:2003p3417}) (a
well-known threading method), and a mean 13\% (median 14\%) improvement in AUC 
over HHPred (\cite{Soding:2005ff}) (a profile-profile HMM method).

\section{Methods}

\subsection{Markov random field model}

MRFy builds on the SMURF and SMURFLite Markov random field model, as 
discussed in Chapter~\ref{chapter:c3_smurflite}, which uses multidimensional 
dynamic programming to simultaneously capture both standard HMM models and the
pairwise interactions between amino acid residues bonded together in
$\beta$-sheets. 

In particular, the ``Plan7'' hidden Markov model is modified to represent
hydrogen-bonded $\beta$-strands with additional, non-local edges.
Because the $\beta$-strands in a SMURF or MRFy template represent 
\emph{consensus} 
$\beta$-strands, those present in at least some fraction (in our experiments, 
at least
half) of the sequences participating in the training alignment, we prohibit
insertions and deletions in those strands.
Thus, we collapse those nodes of the ``Plan7'' model to be just match states;
the transitions to insertion and deletion states are removed.
Figure~\ref{mrfy_model} illustrates this architecture.

\begin{figure}[htb!]
\begin{center}
  \fbox{\includegraphics[width=5in]{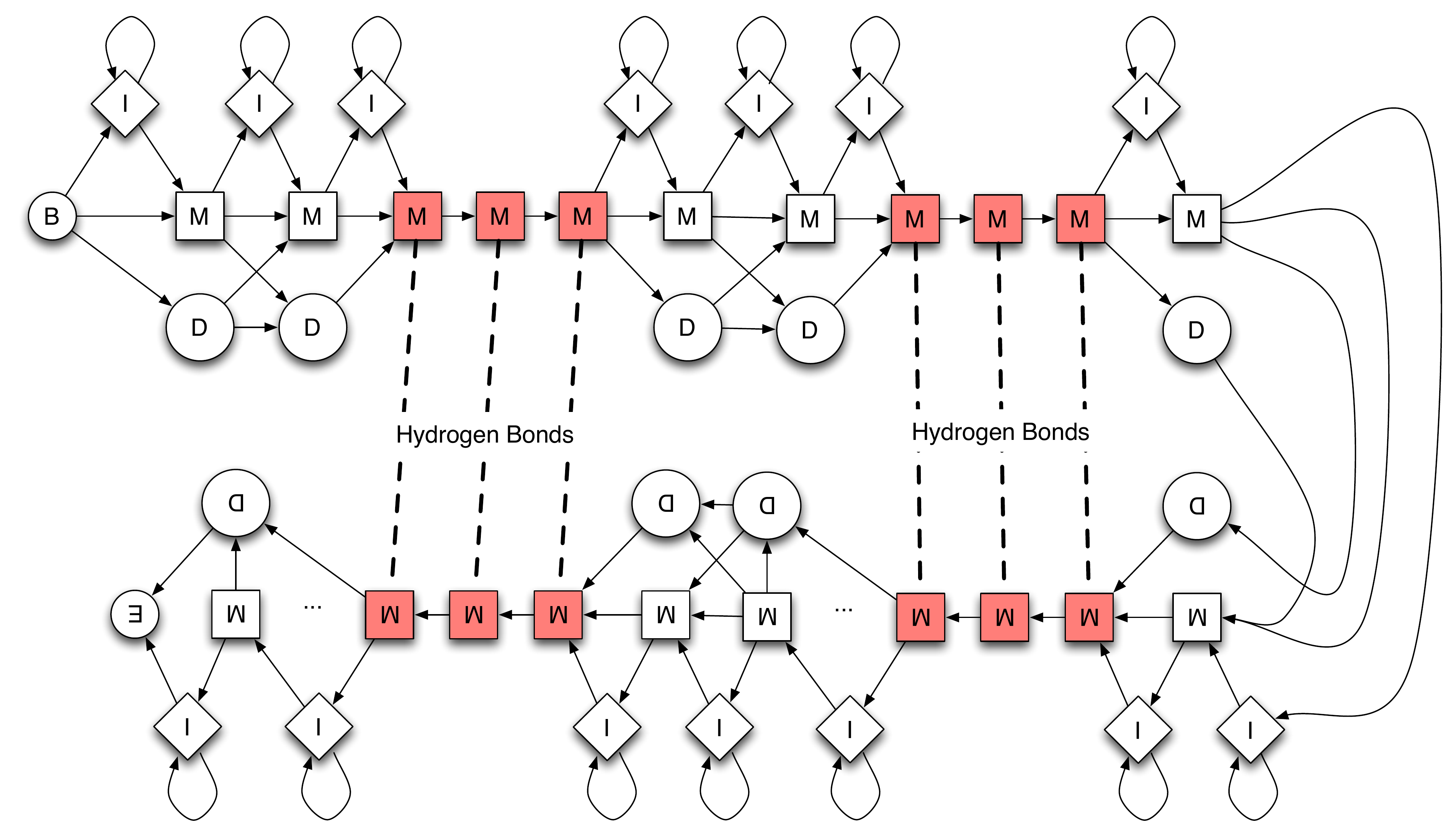}}
   \caption{A Markov random field with two $\beta$-strand pairs}
   \label{mrfy_model}
 \end{center}
\end{figure}

Recall from~\ref{chapter:introduction} that the standard form of the Viterbi
recurrence relations for computing the most likely path of a sequence through
a hidden Markov model is:

\def\maxiquad{\hskip 1.2em\relax}
\begin{equation}
\begin{array}{@{}l@{}c@{}l}

V_{j}^{M}(i) &{}={}& \frac{e_{M_{j}}(x_{i})}{q_{x_{i}}} \times \max \left\{
  \begin{array}{l@{}c@{}l}
  V_{j-1}^{M}(i - 1) \times a_{M_{j-1}M_{j}}\\
  V_{j-1}^{I}(i - 1) \times a_{I_{j-1}M_{j}}\\
  V_{j-1}^{D}(i - 1) \times a_{D_{j-1}M_{j}}\\
  \end{array} \right.\\\\
V_{j}^{I}(i) &=& \frac{e_{I_{j}}(x_{i})}{q_{x_{i}}} \times \max \left\{
  \begin{array}{l@{}c@{}l}
  V_{j}^{M}(i - 1) \times a_{M_{j}I_{j}}\\
  V_{j}^{I}(i - 1) \times a_{I_{j}I_{j}}\\
  \end{array} \right.\\\\
V_{j}^{D}(i) &=& \max \left\{
  \begin{array}{l@{}c@{}l}
  V_{j-1}^{M}(i) \times a_{M_{j-1}D_{j}}\\
  V_{j-1}^{D}(i) \times a_{D_{j-1}D_{j}}\\
  \end{array} \right.\\

\end{array}
\end{equation}

In the SMURF or MRFy Markov random field model, we add non-local interactions 
to these
probabilities, resulting in conditional probabilities.
When column~$j$ of an alignment is part of a $\beta$-strand and is paired
with another column  $\pairedwith j$,
the probability of finding amino acid~$x_i$ in column~$j$ 
depends on whatever amino acid~$x'$  is in column~${\pairedwith i}$.
If~$x'$ is in position~$i'$ in the query sequence, Viterbi's
equations are altered; for example,
$V_{j}^{\prime M}(i)$ depends not only on
$V_{j-1}^{\prime M}(i-1)$ but also on
$V_{\pairedwith j}^{\prime M}(i')$.
The distance between $j$~and~$\pairedwith j$ can be as small as a few
columns or as large as a few hundreds of columns.
Because $V_j^{\prime M}(i)$~depends not only on nearby values but also on
$V_{\pairedwith j}^{\prime M}(i')$,
we must modify the Viterbi recurrence relations.

Note that hydrogen-bonded $\beta$-strand residues may only occupy match states 
in the Markov random field, so only the corresponding terms of the recurrence
relation need be modified.
The revised Viterbi recurrence relation for the Markov random field is:

\begin{equation}
\begin{array}{@{}l@{}c@{}l}

V_{j}^{M}(i) &{}={}& \frac{e_{M_{j}}(x_{i})}{q_{x_{i}}} \times \max \left\{
  \begin{array}{l@{}c@{}l}
  V_{j-1}^{M}(i - 1) \times a_{M_{j-1}M_{j}} \times P(x_{i}|x_{\pi j})\\
  V_{j-1}^{I}(i - 1) \times a_{I_{j-1}M_{j}} \times P(x_{i}|x_{\pi j})\\
  V_{j-1}^{D}(i - 1) \times a_{D_{j-1}M_{j}} \times P(x_{i}|x_{\pi j})\\
  \end{array} \right.\\[\goo]

\end{array}
\end{equation}

where $x_{\pi j}$ represents the amino acid in column $\pi j$, which is 
hydrogen-bonded to the amino acid $x_{i}$ in column $j$.

For reasons of convenience, as well as avoiding floating-point underflow due to
exceedingly small numbers, we typically work in negative log space.
Since a probability can range from 0 to 1, the log of a probability must be a
negative number, and thus the negative log of that probability is a (small)
positive number.
Each probability is transformed into its negative log, resulting in the final
form:

\begin{equation}\label{viterbi-final}
\begin{array}{l@{}c@{}l@{}}
V_{j}^{\prime M}(i) &=& e^{\prime}_{M_{j}}(x_{i}) + \min \left\{
  \begin{array}{l@{}c@{}l}
  \vsum{V_{j-1}^{\prime M}(i-1)} {a^{\prime}_{M_{j-1}M_{j}}} + P'(x_{i}|x_{\pi j})\\
  \vsum{V_{j-1}^{\prime I}(i-1)} {a^{\prime}_{I_{j-1}M_{j}}} + P'(x_{i}|x_{\pi j})\\
  \vsum{V_{j-1}^{\prime D}(i-1)} {a^{\prime}_{D_{j-1}M_{j}}} + P'(x_{i}|x_{\pi j})\\
  \end{array} \right.\\\\
V_{j}^{\prime I}(i) &=& e^{\prime}_{I_{j}}(x_{i}) + \min \left\{
  \begin{array}{l@{}c@{}l}
  \vsum{V_{j}^{\prime M}(i - 1)} {a^{\prime}_{M_{j}I_{j}}}\\
  \vsum{V_{j}^{\prime I}(i - 1)} {a^{\prime}_{I_{j}I_{j}}}\\
  \end{array} \right.\\\\
V_{j}^{\prime D}(i) &=& \min \left\{
  \begin{array}{l@{}c@{}l}
  \vsum{V_{j-1}^{\prime M}(i)} {a^{\prime}_{M_{j-1}D_{j}}}\\
  \vsum{V_{j-1}^{\prime D}(i)} {a^{\prime}_{D_{j-1}D_{j}}}\\
  \end{array} \right.\\
\end{array}
\end{equation}

given the transformations:

\begin{equation}
  \begin{array}{l@{}c@{}l@{}}
  a'_{s \hat{s}} = - \log a_{s \hat{s}} \\
  e^{\prime}_{s}(x) = - \log\frac{e_{s}(x)}{q_{x}}\\
  V_j^{\prime M}(i) = - \log V_j^{M}(i)\\
  P^{\prime}(x_{i}|x_{\pi j}) = - \log P(x_{i}|x_{\pi j})
  \end{array}\\
\end{equation}

This is exactly the recurrence relation that SMURF~\cite{Menke:2010ti} and 
SMURFLite~\cite{Daniels:2012dg} solve using multidimensional dynamic 
programming.
As demonstrated in Chapter~\ref{chapter:c3_smurflite}, as the \emph{interleave} 
of the $\beta$-strands increases, the computational complexity grows 
exponentially.

As an alternative to solving these more complex recurrence relations, we might
consider a divide-and-conquer approach.
Each $\beta$-strand can be thought of as breaking the larger model into two 
smaller
models; collectively, all the $\beta$-strands divide the Markov random field 
into many small, \emph{independent} hidden Markov models.
Thus, for any particular path through the Markov random field, corresponding to
a particular placement of query sequence residues onto the nodes of the model,
we could compute the augmented Viterbi score by summing the Viterbi scores of
each smaller hidden Markov model, along with the contribution to the Viterbi
score from the $\beta$-strands.

Since only match states are allowed for $\beta$-strand residues, the 
contribution of each such residue is only:

\begin{equation}
  \begin{array}{@{}l@{}c@{}l}
  V_{j}^{\prime M}(i) = e^{\prime}_{M_{j}}(x_{i}) + 
  \vsum{V_{j-1}^{\prime M}(i - 1)} {a^{\prime}_{M_{j-1}M_{j}}} + P(x_{i}|x_{\pi j})
  \end{array}
\end{equation}

The asymptotic complexity of the Viterbi algorithm is $O(mn)$, where $m$ is the
length of the model and $n$ is the length of the query sequence.
Furthermore, the asymptotic complexity of the beta-strand contribution to the
Viterbi score for a particular placement of residues is just $O(b)$, where $b$
is the combined length of the $\beta$-strands.

Thus, a new algorithm for computing the optimal path through a Markov random 
field for a given query sequence presents itself.
Since we require that every $\beta$-strand position be occupied by a residue 
(as we force those positions into match states), we could simply consider every
possible assignment of a residue to a $\beta$-strand, computing the score for 
each one, and choose the best-scoring placement.

Metaphorically, we can picture the residues of the query sequence as beads,
and the Markov random field as the string of a necklace.
The $\beta$-strands can be thought of as particular substrings of the string 
that
must be covered by beads, while non-$\beta$ regions may be exposed (resulting in
delete states in the model).
To continue the metaphor, we may force extra beads onto non-$\beta$ regions of 
the string, resulting in insert states in the model.
Given that the beads already have a specified order, we must consider all the
ways to slide the beads up and down the string such that all of the 
$\beta$-regions
are covered.
Since the regions between $\beta$-strands can have their contribution to the
score computed according to the Viterbi recurrence relations, we need only
consider all the unique ways to assign residues to the $\beta$-strand nodes.

\subsection{Proof that the model is exponential in complexity}\label{mrfy-proof}

Here, we prove that there are an exponential number of possible $\beta$-strand
placements that must be considered.

\begin{definition}
Let a Markov random field model $(N,B)$ be defined as a sequence $N$ of nodes 
$n_{i}, i \in (1..m)$, and a sequence $B$ of $\beta$-strands 
$b_{i}, i \in (1..k)$.
Each $\beta$-strand has length $l_{i}$, and contains a subsequence of the nodes
$N$.
This subsequence is determined by the
specifics of the model, which can be referred to as $b_{ij}, i \in (1..m), j \in
(1..l_{i})$.
Let a query sequence be defined as a sequence $R$ of residues
$r_{i}, i \in (1..n)$.
\end{definition}

\begin{definition}
  Let $L = \displaystyle \sum \limits_{i, i <= k} l_{i}$.
\end{definition}

\begin{lemma}
Given a model $(N,B)$ and a query sequence $R$, 
$\displaystyle L$ residues are placed in 
$\beta$-strands.
\end{lemma}

\begin{proof}
Because each $\beta$-strand $b_{i}$ must be populated by exactly $l_{i}$ 
residues,
$\forall j, j > 1$, $b_{ij}$ is uniquely determined by the sequence $R$.
For each $\beta$-strand position $b_{ij}$, one residue is placed.
Thus, $\displaystyle \sum \limits_{i, i<=k} l_{i}$ residues are placed in 
$\beta$-strands.
\end{proof}

\begin{theorem}
For a Markov random field $(N,B)$ with $k$ $\beta$-strands $b_{i}$, each of 
length $l_{i}$, and thus containing positions for residues $b_{ij}$ and a query 
sequence $r{i}$ of length $n$, there are $O(n^{k})$ ways
to assign residues to the $\beta$-strands.
\end{theorem}

\begin{proof}
From the $n$ residues in the query sequence $R$, we need to place $L$ residues
across all $B$ $\beta$-strands.
We represent this as choosing  
an index $i \in (1..n)$ for the first position $b_{i1}$ of each $\beta$-strand.
Since each $\beta$-strand $b_{i}$ consumes $l_{i}$ residues, this choice
for the first $\beta$-strand, $b_{11}$, leaves $n - L - l_{i}$ possible 
placements for $b_{21}$.
In practice, $\beta$-strands range from two to twelve residues, so to simplify 
counting, we assume each $l_{i}$ is simply a
maximum length $l_{max}$.
This only decreases the number of possible assignments, yielding a lower bound
on the number of placements.
Then choosing an index to place on $b_{i1}$, in general, leaves $n - L - 
(i \times l_{max})$ choices for $b_{(i+1)1}$.
Thus, there are:
\begin{equation}
  \prod_{i \in (1..k)}{n - L - (i \times l_{max})} = 
  (n - 2L - k\times(l_{max}))^{\lceil \frac{k}{2} \rceil}
\end{equation}
possible placements of $R$ onto $(N,B)$.
Asymptotically, as $n$ grows, this is dominated by $n^k$, leading to an
asymptotic complexity of $O(n^k)$.
\qed
\end{proof}

A typical Markov random field might have 10 or 20 $\beta$-strands, and a typical
protein query sequence might have between 300 and 600 residues.
Thus, if we wish to consider all possible paths through a Markov random field
for a protein sequence, we must consider as many as 
$600^{10} \approx 6 \times 10^{27}$ possible paths through the model.
Clearly, this computation can be broken into many parallel parts, but this still
poses an intractable problem in many cases.

\subsection{Stochastic search}

Since an exhaustive search for an optimal alignment of a protein sequence to a
Markov random field is exponential in complexity, we turn to stochastic search
to mitigate this complexity.

Stochastic search encompasses a family of approaches for finding optimal or
near-optimal solutions to optimization problems.
Stochastic search approaches are promising when a search space is large, so that
exhaustive search is prohibitive, and when an optimization problem does not lend
itself to analytic solutions.
The generic form of stochastic search is that a solution is guessed at and 
evaluated, and then subsequent guesses are made as refinements to this initial
guess, until some termination condition is met.
The function used for evaluation is called the \emph{objective function}.

Framed as an optimization problem, MRFy, like SMURF, seeks to minimize the 
augmented Viterbi
score (see Equation~\ref{viterbi-final}), which equates to maximizing 
probability (recall that this score is the negative log of a probability).
SMURF finds this minimum exactly, using multi-dimensional dynamic programming,
which is exponential in the interleave number of beta strands (see 
Chapter~\ref{chapter:c3_smurflite}).
MRFy, in contrast, uses stochastic search, as described next.

Given a placement of query-sequence residues into $\beta$-strand nodes of the
Markov random field, the score can be computed exactly.
Thus, the search space is the set of all possible ways to place residues on
these nodes, as discussed in Section~\ref{mrfy-proof}.
Many stochastic search techniques rely on a gradient ascent (or descent)
approach, which makes moves (or refines guesses) along the steepest gradient,
leading quickly to local optima; various heuristics such as simulated 
annealing~\cite{Kirkpatrick:1983wa} can then help avoid getting stuck in poor
local optima.

However, we know of no way to compute a gradient on the search space of
$\beta$-strand placements, and so we must take approaches that do not rely
on this gradient.
Instead, we must rely on a random-mutation model of search, which generates one
or more candidate solutions (guesses) from a previous solution, and then
evaluates the cost function (in our case, the augmented Viterbi score) to
determine whether those guesses are better or worse than the previous step.
This can be likened to climbing a hill in the dark, feeling one's way with one 
foot before committing to a step.
This approach is referred to as \emph{random-mutation hill
 climbing}\cite{Davis:1991vx}.

In our representation, a particular solution is represented by an ordered list
of integers, one integer per $\beta$-strand in the Markov random field.
The value of each integer indicates the index, in the query sequence, of the
residue assigned to the first position of that $\beta$-strand.
Since the alignments to the regions of the Markov random field are solved
exactly by the Viterbi algorithm, this ordered list of integers uniquely
represents a solution to a Markov random field.

While the picture we have presented for our Markov random field model is most
precisely explained by assigning residue indices to the positions of 
$\beta$-strands, it may be more intuitive to consider the equivalent problem
of ``sliding'' these $\beta$-strands along the query sequence.
We will use this analogy in the following description of initial guesses.

We explored three models for generating initial guesses for our search
techniques:

\begin{itemize}
  \item \emph{Random-placement model}.
  First, we implemented a model that uniformly positions the $\beta$-strands
  along the query sequence, under the constraint that only legal placements
  may be generated, and thus the placement of any $\beta$-strand must leave 
  room for all the other $\beta$-strands in the model.
  
  \item \emph{Placement based on secondary structure prediction}.
  Next, we implemented a model that uses the PSIPRED~\cite{McGuffin:2000wx} secondary-structure
  prediction program to determine the positions of $\beta$-strands.
  Given a PSIPRED prediction for the secondary structure of a query sequence,
  we place $\beta$-strands at the most likely locations according to this
  prediction profile, randomized by a small amount of noise.
  The difficulty with this approach is that, while PSIPRED is reasonably 
  accurate when it is allowed to perform PSI-BLAST\cite{Altschul:1997tl} 
  queries to build a sequence
  profile, this comes at a run-time cost that completely dominates the running
  time of MRFy.
  However, while non-profile-based PSIPRED predictions are computationally 
  cheap, they provide poor accuracy.
  
  \item \emph{Placement based on scaling the template}.
  Finally, we implemented a model based on the observation that true homologs
  to a structurally-derived template should have their $\beta$-strands in very
  roughly similar places, in sequence, to the proteins that made up that
  template.
  This will not always hold, but appears to provide for reasonable initial
  guesses.
  Given the position of each $\beta$-strand within a template Markov random 
  field, we scale the query sequence linearly (as it may be shorter or longer 
  than the model) and place the $\beta$-strands in scaled positions.
  Note that we do not scale the $\beta$-strands themselves; their lengths are
  preserved.
  We scale only the distances between $\beta$-strands.
  We inject a small amount of noise into the placements, so that 
  population-based models, such as multi-start simulated annealing and genetic
  algorithms, start with heterogenous solutions.
  
\end{itemize}

Since we do not know how to determine when a stochastic search process has found
a \emph{global} optimum (as opposed to a good local optimum), we must also have
some termination criterion for the search.
We implemented three alternative termination criteria:

\begin{itemize}
  \item A simple generation-counting approach, where the search terminates after
  a user-specified number of generations
  \item A time-based approach, where the search terminates after a 
  user-specified amount of time has elapsed
  \item A \emph{convergence} model, where the search terminates after the search
  has failed to improve after a user-specified number of generations
\end{itemize}
  
In practice, these criteria are easily combined, with a convergence approach
often halting searches early with good results, while the generation- or 
time-based limit ensuring that the search does not take longer than a user is
willing to wait.
We next describe the alternative heuristics that MRFy implements for stochastic
search: simulated annealing, a genetic algorithm, and a local search strategy.

\subsubsection{Simulated Annealing}

Simulated annealing~\cite{Kirkpatrick:1983wa} is a heuristic for stochastic search,
inspired by the physical process of annealing in metals.
Whereas a simple hill-climbing approach will always move downhill (if the task
is minimization) or uphill (if the task is maximization), if the search begins
near to a poor local optimum, the search will terminate at that local optimum.
Simulated annealing introduces an \emph{acceptance probability function}:
\begin{equation}
  \begin{split}
  P(e,e',T) = \begin{cases}
    1, & \text{ if } e' < e\\
    \exp(-(e'-e)/T),& \text{ otherwise }\\
\end{cases}\\
  \text{ where }
     e = E(s) \\
     e' = E(s')\\
  \end{split}
\end{equation}
which relies on some energy function $E(s)$ of the current state $s$ and
a candidate state $s'$, and a temperature function $T$ that tends towards zero
as the search progresses.
In our implementation, we used an exponentially-decaying temperature function:
\begin{equation}
  T(t) = k^{t}\times T_{0}
\end{equation}
given time $t$, initial temperature $T_{0}$, and a constant $k$.
The motivation for this decaying temperature function is that, as time 
progresses, the likelihood of being in a \emph{poor} local optimum lessens, and
thus, the closer to random hill-climbing we would like the search to behave.

Our energy function $E(s)$ is, naturally, the augmented Viterbi score of a
placement:
\begin{equation}
  E(s) = V_{m}^{\prime M}(n)
\end{equation}
where $m$ is the final residue in the query sequence and $n$ is the final
node in the Markov random field, and the $\beta$-strand placements are 
determined by $s$.

We implemented simulated annealing in MRFy according to this model.
We also implemented a \emph{multi-start} version of simulated annealing in MRFy,
where a set of independently-generated guesses is subject to simulated-annealing
random descent, in parallel.
At the termination of the search, the best solution from among all the 
candidates is chosen.

\subsubsection{Genetic Algorithm}

A genetic algorithm~\cite{Holland:1977hl} is a search heuristic inspired by 
biological evolution.
A genetic algorithm relies on the idea of \emph{selection} among a population of
varied solutions to an optimization problem.
At each of many generations, the fitter individuals in the population--those
solutions which exhibit more optimal scores--are allowed to continue into the
next generation.
Not only do they continue into the next generation, but they are allowed to
``reproduce,'' or recombine, to produce new solutions.
A particular solution to a problem, within the context of a genetic algorithm,
is called a \emph{chromosome}.
At each generation, some fraction of the fittest solutions are selected and
randomly paired with one another.
Each pair of solutions produces one or more offspring; each offspring is the
result of two steps: \emph{crossover} of the two chromosomes, followed by
random \emph{mutation} of the offspring.
The mutation is nondeterministic; the crossover may be deterministic or
nondeterministic.
The resulting offspring, along with their parents, are then evaluated according
to the objective function, and this process iterates until some termination
condition.

MRFy's genetic algorithm implementation uses the same representation for a
placement as simulated annealing: an ordered list of integers.

Let a \emph{placement} $p$ on a model with $k$ $\beta$-strands be an ordered 
set of integers $p_{i}, i \in (1..k)$.
Given two placements, $p$ and $q$, MRFy implements crossover of two chromosomes 
using the following algorithm:

\begin{enumerate}
  \item Set the new placement, $p'$, to the empty set.
  \item Repeat until all placements have been chosen:
  \begin{enumerate}
    \item Let $p'_{0} = p_{0}$
    \item Let $p'_{k} = q_{k}$
    \item Remove $p_{0}$ and $p_{k}$ from $p$
    \item Remove $q_{0}$ and $q_{k}$ from $q$
    \item $p = <p_1,...,p_{k-1}>$
    \item $q = <q_1,...,q_{k-1}>$
  \end{enumerate}
\end{enumerate}

Our actual implementation is purely functional, and simply consumes elements 
from lists.
In effect, though, this algorithm simply chooses the `left-most' elements from 
one parent and the `right-most' elements from another.
After crossover, the mutation step simply moves each element $p_{i}$ of the
placement $p$ by a small, random amount, within the constraints imposed by the
neighboring $\beta$-strands.
The motivation behind this approach is to take two solutions that are of high
fitness (recall that the worst solutions at every generation are not allowed
to contribute to the next generation), and produce a new solution that combines
one ``half'' (roughly) of one solution with one ``half'' of the other.
See Figure~\ref{crossover} for an illustration of this procedure.

\begin{figure}[htb!]
\begin{center}
  \fbox{\includegraphics[width=5in]{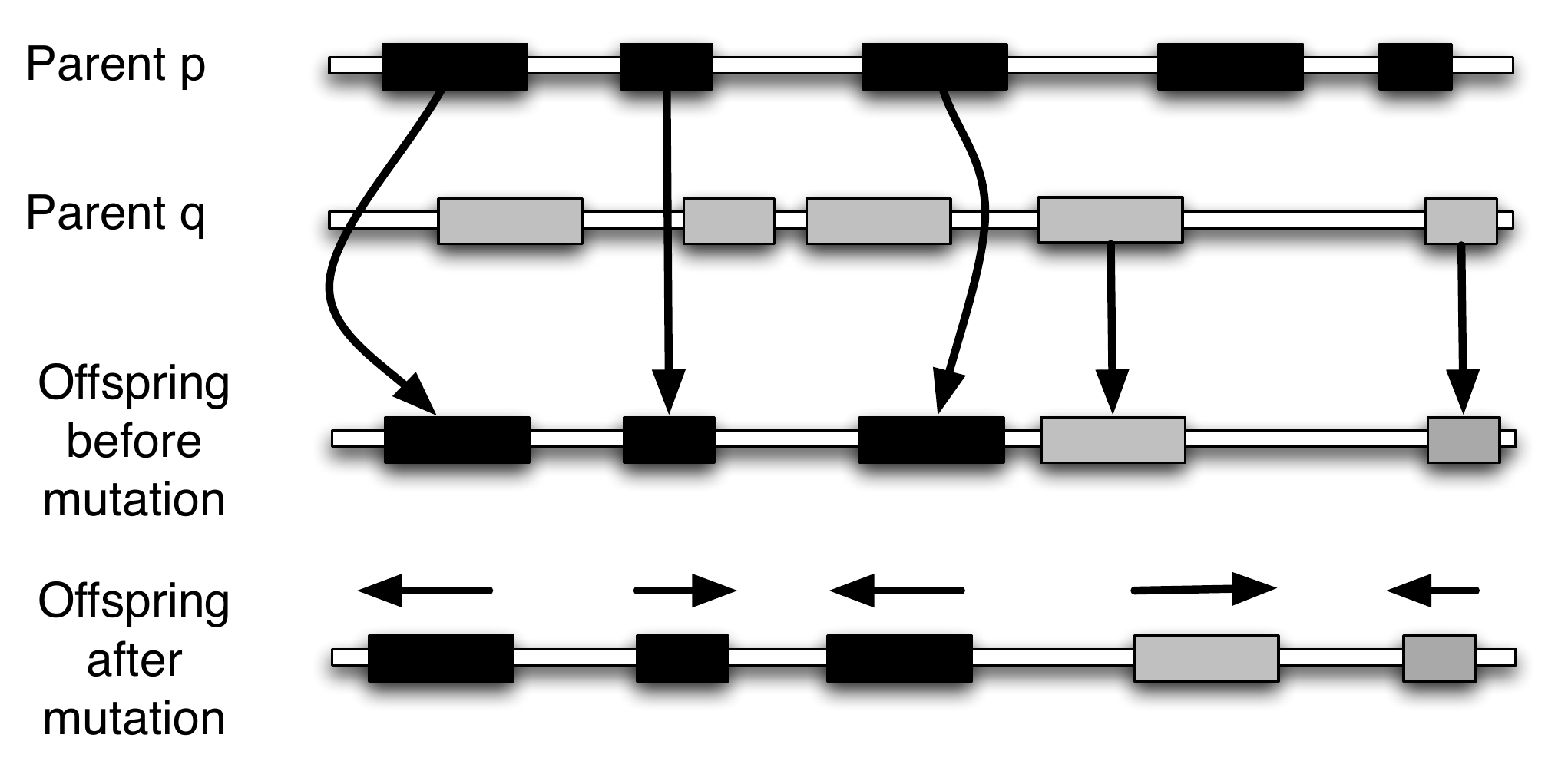}}
   \caption[The crossover and mutation process in MRFy's genetic algorithm 
   implementation.]{The crossover and mutation process in MRFy's genetic 
   algorithm implementation.
   Given parent $p$ (black) and parent $q$ (gray), alternate left and right
    placements from $p$ and $q$.
    Then, apply small random mutations to the resulting placement $p'$.}
   \label{crossover}
 \end{center}
\end{figure}

Given these operations for crossover and mutation, MRFy's genetic algorithm
implementation initializes a population of a user-specified size $P$ (typically 
one thousand placements, though we experimented with as many as ten thousand).
In parallel, each placement is scored according to the objective function.
Since scoring is far more computationally expensive than crossover and mutation,
we allow them all to reproduce, paired at random.
We then score them, and choose the $P$ best-scoring placements for the next
generation.
This process repeats until a termination condition is met, at which point the
single best placement is returned.
We note that a future enhancement to MRFy could return the $k$ best placements
for some user-specified threshold $k$, if multiple high-scoring alignments were
to be considered.

\subsubsection{Local Search}\label{localsearch}

Constraint-based local search~\cite{Hentenryck:2009vn} is a family of approaches
for exploring ``neighborhoods'' in feature space in a randomized manner, 
subject to the constraints of that solution space.
In the context of MRFy, the constraints are the previously-discussed 
restrictions that $\beta$-strands cannot overlap, and every residue must be
placed in a $\beta$-strand.
The motivation for local search is, in a particularly uneven
fitness landscape, hill climbing will often reach nearby local optima. Thus,
given a single candidate solution, local search explores the immediate 
neighborhood in great detail (perhaps, but not necessarily exhaustively).
When the local search cannot escape a local optimum, then some sort of 
\emph{non-local} move may be attempted.

This non-local move may rely on a population-based diversification approach,
in which parts of the solution may change dramatically.
In a sense, local search bears some resemblance to a genetic algorithm,
except that a population of solutions is created only when the search is stuck
in a local optima, and the best solution in that population is chosen for a new
search.

In MRFy's implementation, each step in the search consists of two phases: 
\emph{diversification} (See Figure~\ref{diversification}) and 
\emph{intensification}.
The diversification algorithm is as follows:
\begin{itemize}
\item Begin with a candidate solution $s$ (a placement), which is just
an ordered list of integers.
\item Given $s$, break the list into three sub-lists $s_{0}, s_{1}$, $s_{2}$, 
at randomly-chosen boundaries.
\item Choose one of the sub-lists $s_{i}$ at random, and mutate it into $k$ 
copies $s_{i1}$ through $s_{ik}$ at random, for some user-defined value of $k$ 
(we used $k=10$), within
the constraints imposed by the other sub-lists and the lengths of the 
$\beta$-strands.
\item Re-combine each set of lists, $(s_{1j}, s_{2j}, s_{3j})$ into a new 
placement $s^{\prime}_{j}, j \in (1..k)$.
\item Score each placement $s^{\prime}_{j}$, return the best-scoring of the $k$ new placements as a new solution.
\end{itemize}

Once diversification produces a new candidate solution, intensification brings
it toward a local minimum.
The intensification algorithm is as follows:
\begin{itemize}
\item Begin with a candidate solution $s$.
\item Repeat until no better-scoring placements are generated.
\begin{itemize}
  \item For each element $e \in s$, generate four new placements $s'_{i1}$ 
  through $s'_{i4}$ by moving $e$ up and down by 1 and two, as long as those 
  moves do not violate the constraints.
  \item Score each candidate placement $s'_{ij}$.
  \item Set $s$ to the best-scoring candidate placement $s'_{ij}$
\end{itemize}
\item Return $s$ as a new solution.
\end{itemize}

\begin{figure}[htb!]
\begin{center}
  \fbox{\includegraphics[width=4.5in]{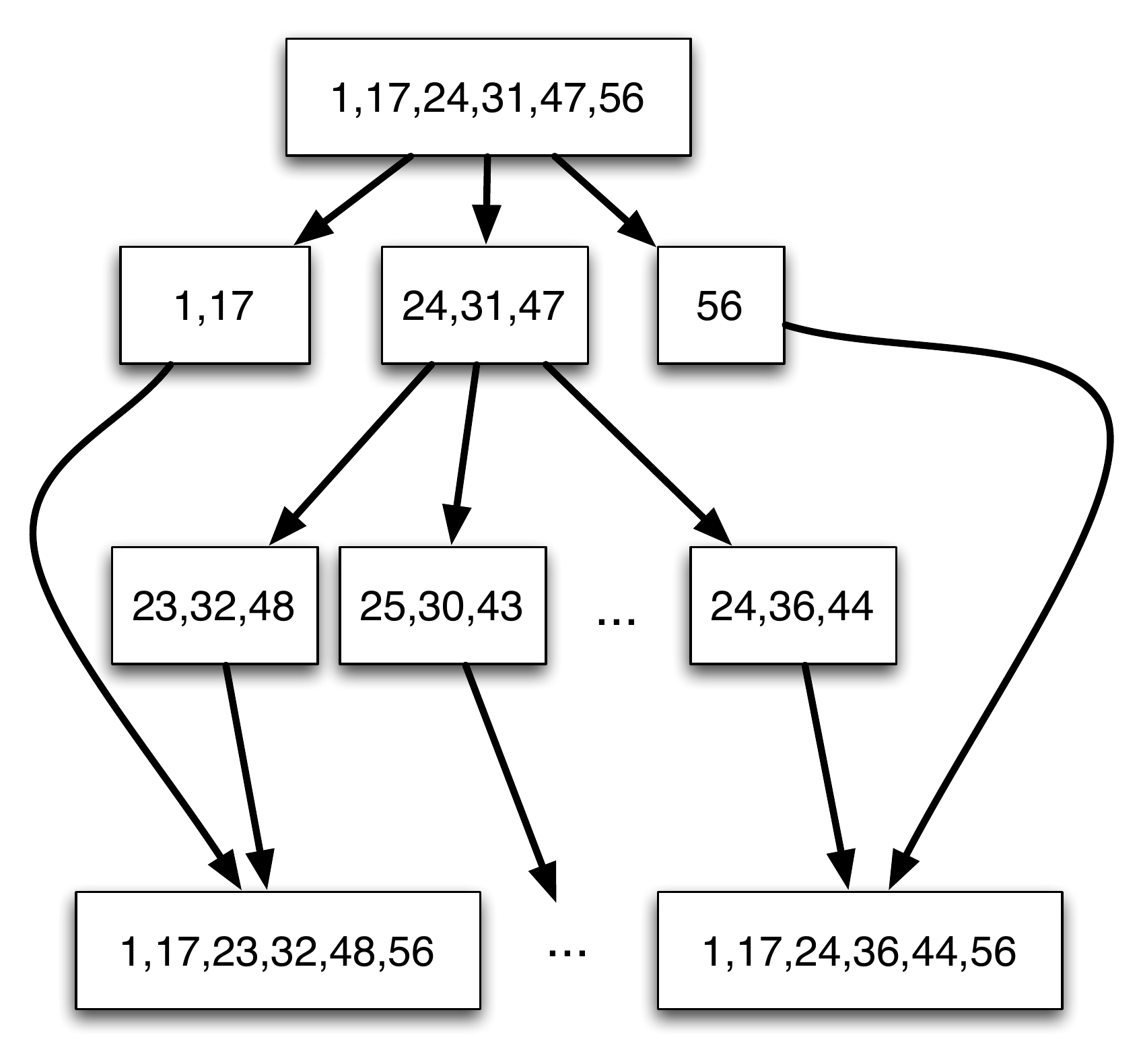}}
   \caption{The diversification step in local search.}
   \label{diversification}
 \end{center}
\end{figure}

%
%
%


\subsection{Evaluating search strategies}\label{searchstrats}

As MRFy supports three significantly different stochastic search strategies, and
a number of tunable parameters such as termination conditions and (for 
simulated annealing) the cooling schedule, we conducted a search over parameter
space using a small data set.
We built Markov random field templates from the fold ``8-bladed Beta-Propellers'', and the superfamilies ``Barwin-like endoglucanases'' (a 
$\beta$-barrel superfamily) and ``Concanavalin A-like lectins/glucanases'' (a 
$\beta$-sandwich superfamily).
We were interested in the speed of convergence for a true-positive test case,
so we tested each template with a protein sequence chosen from that fold or
superfamily: for the 8-bladed propeller, we chose ASTRAL chain d1lrwa\_
(Methanol dehydrogenase, heavy chain from  \emph{Paracoccus denitrificans}).
For the barwin-like endoglucanases, we chose ASTRAL chain d2pica1 
(Membrane-bound lytic murein transglycosylase A, MLTA from \emph{E. coli}).
For the lectins/glucanases, we chose ASTRAL chain d2sbaa\_ 
(Legume lectin from soy bean (\emph{Glycine max})).

We tested simulated annealing with a population size of 10, a maximum number of
generations of 10000, convergence periods of 200, 500, and 1000 generations, 
and a cooling factor of 0.99 (preliminary tests showed little impact from
varying the cooling factor among 0.9, 0.99, and 0.999).

We tested the genetic algorithm implementation with a population size of 1000 
and 10000, a maximum number of generations of 500, and convergence
periods of 10, 50, and 100.

Since the local search distinguishes between diversification and 
intensification, counting the number of generations is ambiguous; we used a time
limit of 10 seconds, 30 seconds and 5 minutes.
All tests were conducted on a 12-core AMD Opteron 2427 with 32GB RAM, devoting 
all 12 cores to MRFy.
For each test, we report statistics based on ten runs for each set of 
parameters.

\subsection{Simulated Evolution} \label{mrfy-simev}

In MRFy, we incorporated precisely the same ``simulated evolution'' 
implementation, as first proposed by Kumar and 
Cowen~\cite{Kumar:2009tp, Kumar:2010wv}, as we did for SMURFLite in 
Chapter~\ref{chapter:c3_smurflite}.
We added pairwise mutations based on $\beta$-strand pairings.
Unlike in Chapter~\ref{chapter:c3_smurflite}, 
here we were not attempting to mitigate
the loss of information due to simplifying the Markov random field, but rather
attempting to compensate for sparse training data.
This was motivated in part by the observation that SMURFLite benefited most
from simulated evolution when the ``number of effective families'' was low.
We use the same mutation frequencies as in Chapter~\ref{chapter:c3_smurflite}.
For each artificial sequence, we mutate at a 50\% mutation rate per length of 
the $\beta$-strands. 

\subsection{Datasets} \label{mrfy-datasets}

From SCOP (\cite{Murzin:1995uh}) version 1.75, we chose the same
$\beta$-structural superfamiles as for SMURFLite 
(Chapter~\ref{chapter:c3_smurflite}).
These superfamilies were: ``Nucleic acid-binding proteins'' (50249), 
``Translation proteins'' (50447), 
``Barwin-like endoglucanases'' (50685), ``Cyclophilin-like'' (50891), ``Sm-like 
ribonucleoproteins'' (50182), ``PDZ domain-like'' (50156), ``Prokaryotic 
SH3-related domain'' (82057), ``Tudor/PWWP/MBT'' (63748), ``Electron Transport 
accessory proteins'' (50090), ``Translation proteins SH3-like domain'' (50104), 
and ``FMN-binding split barrel'' (50475). 

\subsection{Training and testing process} \label{mrfy-training}

For the $\beta$-barrel superfamilies, we performed strict leave-family-out 
cross-validation. 
We built training templates at the superfamily level. 
For each superfamily, its constituent families were identified. 
Each family was left out, and a training set was established from the protein 
chains in the remaining families, with duplicate sequences removed. 
We built an MRF on the training set, both with and without training-data
augmentation using the same ``simulated evolution'' implementation as 
in Chapter~\ref{chapter:c3_smurflite}: we generate 150 new artificial training 
sequences from each original training sequence. 
For each artificial sequence, we mutate at a 50\% mutation rate per length of 
the $\beta$-strands.
We chose protein chains from the left-out family as positive test examples. 
Negative test examples were protein chains from all other superfamilies in SCOP 
classes 1, 2, 3 and 4 (including other barrel superfamilies), indicated as 
representatives from the nr-PDB (\cite{Berman:2000hl}) database with 
non-redundancy set to a BLAST E-value of $10^{-7}$.

We used MRFy's local search mode (see Section~\ref{localsearch}) to align each
test example to the trained MRF.
The score reported for MRFy was the combined HMM and pairwise score from the 
MRF, which is identical to the SMURF energy function.
For each training set, the scores for both methods (MRFy with and without
simulated evolution) were collected and a ROC curve (a plot of true positive 
rate versus false positive rate) computed. We report the area under the curve 
(AUC statistic) from this ROC curve (\cite{Sonego:2008uy}).

\section{Results}

\subsection{Search strategies}

For the three stochastic search approaches, we compared the raw score achieved
by each approach under a variety of conditions, as discussed in 
Section~\ref{searchstrats}.
The raw score is simply the negative log of the probability of the best path
found through the model.
Thus, raw scores are not comparable between models, but they are comparable
between query sequences for a given model.

Table~\ref{ss-propeller} indicates the performance of different stochastic
search techniques on the 8-bladed $\beta$-propeller fold.
While the simulated annealing and genetic algorithm approaches exhibit less
variance (a smaller standard deviation) from run to run, they do not approach
the minimum score of the local search approaches. Multi-start simulated 
annealing with a population of 10 and a convergence threshold of 200 generations
averages 29.3 seconds per search, but only achieves a minimum score of 2112, 
though it converged in all cases.

In contrast, local search, given 30 seconds, achieves a minimum score of 1982, 
and even in only 10 seconds achieves a minimum score of 1992.
However, the global minimum score of 1781, which is achieved by SMURF on the
8-bladed $\beta$-propeller template, is only reached by MRFy with local search
two out of ten times, and this result required local search be allowed to run
for twenty minutes.
Thus, for this problem domain, local search seems to outperform our simulated
annealing and genetic algorithm implementations.

Table~\ref{ss-barwin} indicates the performance of the stochastic
search techniques on the ``Barwin-like endoglucanases'' $\beta$-barrel 
superfamily.
These structures are less complex than the propellers, even though they are
\emph{more} computationally complex for SMURFLite (Chapter 
\ref{chapter:c3_smurflite}) if an interleave threshold greater than 2 is used. 
We see less variance than with the propellers, but once again, the local search
technique achieves a lower minimum score than simulated annealing or the genetic
algorithm.

Notably, local search achieves a minimum score of 978, which an 
\emph{exhaustive} search indicates to be a global minimum for this sequence on
this template.
With a time limit of 10 seconds, local search found this global minimum in one 
out of ten runs. 
With a time limit of 30 seconds, local search found it in two out of
ten runs, and with a time limit of 5 minutes, in four out of ten runs.

\rowcolors{2}{gray!25}{white}

\begin{small}
\begin{center}
\begin{table*}[htb]
\caption{Stochastic search performance on 8-bladed $\beta$-propeller \label{ss-propeller}}
\begin{tabular}{lllll}
\hline
             & Min Score & Mean Score & Std Score & Mean Time (s) \\
\hline
SA 200       & 2112      & 2139       & 12.2      & 29.3      \\
\hline                                                        
SA 500       & 2129      & 2146       & 9.3       & 1020      \\
\hline                                                        
SA 1000      & 2112      & 2130       & 7.8      & 3314      \\
\hline                                                        
GA 1000/10   & 2105      & 2126       & 6.6      & 285      \\
\hline                                                        
GA 1000/50   & 2094      & 2118       & 7.7      & 1239      \\
\hline                                                        
GA 1000/100  & 2107      & 2120       & \textbf{3.8}      & 548      \\
\hline                                                        
GA 10000/10  & 2087      & 2111       & 7.2      & 5809      \\
\hline                                                        
GA 10000/50  & 2094      & 2112       & 7.1      & 5174      \\
\hline                                                        
GA 10000/100 & 2079      & 2114       & 9.0      & 10226      \\
\hline                                                        
LS 10s       & 1992     & 2015      & 19.4     & \textbf{10}    \\
\hline                                                        
LS 30s       & 1982     & 1991      & 10.9     & 30    \\
\hline                                                        
LS 5m        & \textbf{1818}     & \textbf{1876}      & 37.2     & 300    \\
\hline
\end{tabular}\\
{Performance of stochastic search techniques on an 8-bladed $\beta$-propeller
template. SA is Simulated Annealing, GA is Genetic Algorithm, and LS is Local
Search. For Simulated Annealing, we show results for convergence thresholds of
200, 500, and 1000 generations. 
For the Genetic Algorithm, we show results for convergence thresholds of 10, 
50, and 100 generations, and for population sizes of 1000 and 10000.
For Local Search, we show results for time limits of 10 seconds, 30 seconds and 
five minutes, on a 12-core AMD Opteron.
MRFy never achieved the global optimum score of 1781, achieved by SMURF, on this
template, except when local search was given 20 minutes of compute time, in
which case it found the global optimum two out of ten times.}
\end{table*}
\end{center}
\end{small}

\begin{small}
\begin{center}
\begin{table*}[htb]
\caption{Stochastic search performance on ``Barwin-like'' $\beta$-barrel \label{ss-barwin}}
\begin{tabular}{llllll}
\hline
             & Min Score & Mean Score & Std Score & Mean Time (s) & Optimal \\
\hline
SA 200       & 1064      & 1071       & 3.8      & 79.5  & 0    \\
\hline                                                        
SA 500       & 1047      & 1063       & 7.6      & 104   & 0   \\
\hline                                                        
SA 1000      & 1024      & 1047       & 14.0      & 523  & 0    \\
\hline                                                        
GA 1000/10   & 1061      & 1069       & 3.6      & 232   & 0   \\
\hline                                                        
GA 1000/50   & 1059      & 1066       & 3.1      & 442    & 0  \\
\hline                                                        
GA 1000/100  & 1058      & 1069       & 4.0      & 1382   & 0   \\
\hline                                                        
GA 10000/10  & 1058      & 1063       & 2.5      & 8205   & 0   \\
\hline                                                        
GA 10000/50  & 1059      & 1061       & \textbf{2.2}      & 10306   & 0   \\
\hline                                                        
GA 10000/100 & 1057      & 1061       & \textbf{2.2}      & 16395   & 0   \\
\hline                                                        
LS 10s       & \textbf{978}     & 995      & 16.2     & \textbf{10} & 0.1   \\
\hline                                                        
LS 30s       & \textbf{978}     & 987      & 6.9     & 30  & 0.2  \\
\hline                                                        
LS 5m        & \textbf{978}     & \textbf{981}      & 2.9     & 300  & 0.4  \\
\hline
\end{tabular}\\
{Performance of stochastic search techniques on the ``Barwin-like 
endoglucanases'' $\beta$-barrel template. 
SA is Simulated Annealing, GA is Genetic Algorithm, and LS is Local Search. 
For Simulated Annealing, we show results for convergence thresholds of
200, 500, and 1000 generations. 
For the Genetic Algorithm, we show results for convergence thresholds of 10, 
50, and 100 generations, and for population sizes of 1000 and 10000.
For Local Search, we show results for time limits of 10 seconds, 30 seconds and 
five minutes, on a 12-core AMD Opteron.
The ``Optimal'' column indicates the fraction of runs for each search method
that achieved the global optimum.}
\end{table*}
\end{center}
\end{small}

Table~\ref{ss-sandwich} indicates the performance of the stochastic
search techniques on the ``Concanavalin A-like lectins/glucanases'' 
$\beta$-sandwich superfamily.
These structures are also more complex than the propellers, even though they 
are also more computationally complex for SMURFLite with an interleave 
threshold greater than 2. 
On this superfamily, there is a closer overlap between the minimum score 
achieved by simulated annealing, at 790, and the range seen by local search;
local search with a time limit of 30 seconds achieves a \emph{mean} minimum
score of 791, though its best was 740.

Notably, when given a time limit of 5 minutes, local search achieved the 
\emph{global minimum} of 554 (as determined by exhaustive search) ten out of 
ten times. Local search never found this score when given only 10 seconds or
30 seconds as a time limit.

\begin{small}
\begin{center}
\begin{table*}[!t]
\caption{Stochastic search performance on $\beta$-sandwich\label{ss-sandwich}}
\begin{tabular}{llllll}
\hline
             & Min Score & Mean Score & Std Score & Mean Time (s) & Optimal \\
\hline
SA 200       & 795      & 834       & 18.6      & 84.7    & 0  \\
\hline                                                        
SA 500       & 790      & 820       & 17.3      & 192    & 0  \\
\hline                                                        
SA 1000      & 791      & 811       & 14.7      & 493    & 0  \\
\hline                                                        
GA 1000/10   & 874      & 888       & 4.1      & 1869    & 0  \\
\hline                                                        
GA 1000/50   & 878      & 883       & \textbf{2.5}      & 1305   & 0   \\
\hline                                                        
GA 1000/100  & 865      & 878       & 5.6      & 4309    & 0  \\
\hline                                                        
GA 10000/10  & 872      & 877       & \textbf{2.5}      & 6999   & 0   \\
\hline                                                        
GA 10000/50  & 875      & 879       & 3.1      & 5317    & 0  \\
\hline                                                        
GA 10000/100 & 869      & 875       & 4.5      & 10733   & 0   \\
\hline                                                        
LS 10s       & 771     & 826      & 31.7     & \textbf{10}  & 0  \\
\hline                                                  
LS 30s       & 740     & 791      & 47.0     & 30  & 0  \\
\hline                                                  
LS 5m        & \textbf{554}     & \textbf{554}      & \textbf{0.0}     & 300 & 1.0   \\
\hline
\end{tabular}\\
{Performance of stochastic search techniques on a ``Concanavalin A-like 
lectins/glucanases'', a 12-stranded $\beta$-sandwich template. 
SA is Simulated Annealing, GA is Genetic Algorithm, and LS is Local
Search. 
For Simulated Annealing, we show results for convergence thresholds of
200, 500, and 1000 generations. 
For the Genetic Algorithm, we show results for convergence thresholds of 10, 
50, and 100 generations, and for population sizes of 1000 and 10000.
For Local Search, we show results for time limits of 10 seconds, 30 seconds and 
five minutes, on a 12-core AMD Opteron.
The ``Optimal'' column indicates the fraction of runs for each search method
that achieved the global optimum.}
\end{table*}
\end{center}
\end{small}

\rowcolors{2}{white}{white}

Our Haskell implementation made it exceedingly easy to parallelize MRFy across
multiple processing cores.
By default, MRFy will take advantage of all processing cores on a system; we
tested the parallel speedup on a system with 48 processing cores.
We measured the run-time performance of MRFy's genetic algorithm implementation
(with a fixed random seed) on the ``8-bladed $\beta$-propeller'' template.
The model has 343~nodes, of which 178 appear in 40~$\beta$-strands.
The segments between $\beta$-strands typically have at most 10~nodes.
We~used a query sequence of 592 amino acids, but each placement breaks
the sequence into 41~pieces, each of which typically has at most 20 amino
acids.
Because MRFy can solve the models between the $\beta$-strands independently,
this benchmark has a lot of parallelism.
\enlargethispage{\baselineskip}

Figure~\ref{speedup} shows speedups when using from 
 1 to 48 of the cores 
on a 48-core, 2.3GHz AMD Opteron 6176 system.
Errors are estimated from 5 runs.
After about 12 cores, where MRFy runs 6 times as fast as sequential code, 
speedup rolls off.
We note that by running 4 instances of MRFy in parallel on different searches,
we would expect to be able to use all 48 cores with about 50\% efficiency.

\begin{figure}[htb!]
\begin{center}
  \fbox{\includegraphics[width=5.25in]{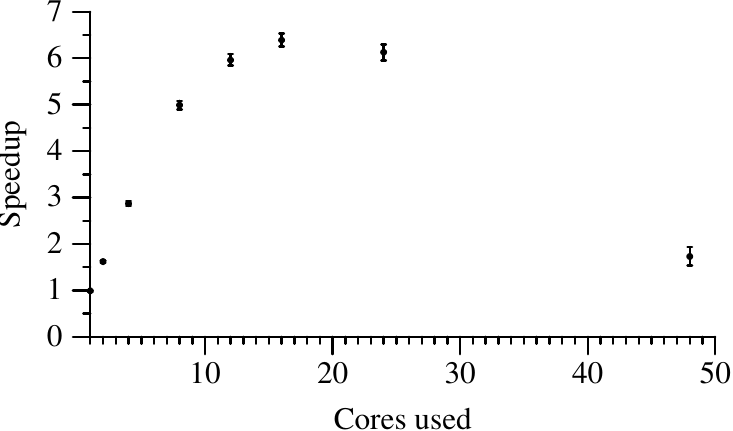}}
   \caption[MRFy's parallel speedup]{MRFy's parallel speedup on an 8-bladed 
   $\beta$-propeller, using a 48-core system.
   After about 12 cores, speedup falls off.}
   \label{speedup}
 \end{center}
\end{figure}

\subsection{Remote homology detection accuracy}

We performed cross-validation testing on 11 $\beta$-barrel superfamilies, both
with and without simulated evolution.
For MRFy, the balance between accuracy and computational efficiency is 
determined by the termination conditions, as well as the search technique 
chosen.
Because local search so dramatically outperformed simulated annealing and the
genetic algorithm, we conducted these cross-validation tests only on local 
search.
We chose 30 seconds as a balance between speed and accuracy; a 5 minute time
limit might result in better accuracy, but for high-throughput, whole-genome 
scans, 5 minutes per alignment is excessive.

We compared MRFy's performance, both with and without simulated evolution, to 
the results from Chapter~\ref{chapter:c3_smurflite}.
Table~\ref{mrfy-auc} shows the area (AUC) under the Receiver Operator 
Characteristic (ROC) curve for MRFy, the very best result from SMURFLite, and
HMMER, RAPTOR, and HHPred.
Importantly, we are choosing the best SMURFLite parameters for each superfamily,
which could not be known in advance; thus, we demonstrate improvements over the
\emph{very best} SMURFLite can perform, rather than just an average case.

We first note the ``Barwin-like endoglucanases'' superfamily highlighted in
Chapter~\ref{chapter:c3_smurflite}.
SMURFLite performed better as the interleave threshold was increased on this
superfamily, and also when simulated evolution was added.
Since MRFy discards no $\beta$-strands, we were curious how it would perform
on this superfamily.
Notably, this superfamily has exceedingly little training data; during 
cross-validation, there are at most 4 training sequences and as few as 3
when filtered at a BLAST E-value of $10^{-7}$ and the family under test is left
out.
Without simulated evolution, MRFy achieves an AUC of 0.86, outperforming 
SMURFLite
without simulated evolution (SMURFLite achieved an AUC of 0.77 with an 
interleave threshold of 2, and 0.81 with an interleave threshold of 4).
When simulated evolution is added, MRFy achieves an AUC of 0.92, outperforming
SMURFLite with an interleave threshold of 2, but falling just short of the 0.94
AUC SMURFLite demonstrates with an interleave threshold of 4 and simulated 
evolution.

MRFy outperforms SMURFLite in terms of AUC on four of the $\beta$-barrel 
superfamilies, while SMURFLite outperforms MRFy on three.
There was only one superfamily, the ``Prokaryotic SH3-related domain,'' where
SMURFLite outperformed HMMER, RAPTOR, and HHPred while MRFy did not (MRFy, with 
an AUC of 0.73, fell behind HMMER's 0.81 AUC).
Unfortunately, MRFy never produced the best performance on a superfamily that 
SMURFLite had not performed best on in our previous work.
Thus, with the exception of the ``Barwin-like endoglucanase'' superfamily, the
added $\beta$-strand information does not seem to help MRFy significantly in 
the cases where HMMER, RAPTOR, or HHPred performed best.


\rowcolors{2}{gray!25}{white}

\newcolumntype{H}[1] {%
>{\raggedright}%
p{#1}}

\begin{small}

\begin{center}
\begin{table*}[!t]
\caption{AUC on Beta-Barrel superfamilies\label{mrfy-auc}}
{\footnotesize \begin{tabular*}{\textwidth}{@{\extracolsep{\fill}}H{3.3cm}p{1.0cm}p{1.1cm}p{1.0cm}p{2.7cm}p{0.7cm}p{1.6cm}}\hline
 & HMMER & RAPTOR & HHPred & SMURF\-Lite (best) & MRFy & MRFy, SE\\
 \hline
{\bf MRFy performs best}  & & & & & & \\
\hline
Translation proteins & - & - & 0.66 & 0.93 & {\bf 0.95} & 0.91\\
Barwin-like endoglucanases & - & - & 0.75 & 0.77 & 0.86 & {\bf 0.92} \\
Tudor/PWWP/MBT & 0.78 & 0.74 & 0.67 & 0.83 & {\bf 0.86} & {\bf 0.86}\\
Nucleic acid-binding proteins & 0.75 & - & 0.67 & 0.89 & 0.75 & {\bf 0.95} \\
 \hline
 {\bf SMURFLite performs best}  & & & & & & \\
\hline
Cyclophilin-like & 0.67 & 0.61 & 0.7 & {\bf 0.85} & 0.82 & 0.80 \\ 
Sm-like ribonucleoproteins & 0.73 & 0.71 & 0.77 & {\bf 0.85} & 0.77 & 0.77 \\
Prokaryotic SH3-related domain & 0.81 & -  & - & {\bf 0.83} & 0.73 & 0.72 \\
 \hline
{\bf HHPred performs best}  & & & & & & \\
\hline
Translation proteins SH3-like & 0.83 & 0.81 & {\bf 0.86} & 0.62 & - & 0.63\\
 \hline
 {\bf RAPTOR performs best} & & & & & & \\
\hline
PDZ domain-like & 0.96 & {\bf 1.0} & 0.99 & 0.97 & 0.95 & 0.95 \\
FMN-binding split barrel & 0.62 & {\bf 0.82} & 0.61 & - & - & - \\
 \hline
{\bf HMMER performs best}  & & & & & & \\
\hline
Electron Transport accessory proteins & {\bf 0.84} & - & 0.77 & 0.66 & - & 0.68 \\
\hline
\end{tabular*}}\\{Note: for SMURFLite, value indicated is the best of all
values in Table~\ref{barreltable}. For MRFy, SimEv is simulated evolution. A dash ('-') in a result entry indicates the method failed on these structures, i.e. an AUC of less than 0.6}
\end{table*}
\end{center}

\end{small}

\rowcolors{0}{white}{white}

\section{Discussion}

We have presented MRFy, a method that uses stochastic search to find 
alignments of protein sequences to Markov random field models.
MRFy in most cases outperforms SMURFLite, but we should consider several 
possible enhancements to MRFy that might improve its performance.
As demonstrated on the $\beta$-sandwich superfamily, MRFy with local search 
achieves a globally optimal alignment when given 5 minutes of run-time, but 
fails to find a score close to this when given only 30 seconds.
It was not immediately clear how to bring convergence testing into the local
search model, but doing so might achieve results comparable to the 5 minute
results in less time.

We hope that MRFy will be useful for whole-genome annotation of newly-sequenced
organisms.
The tradeoff of time versus accuracy suggests a two-phase approach to this task:
a scan with relatively strict run-time performance requirements (perhaps no more
than ten seconds per alignment) coupled with a relatively loose $p$-value 
threshold would produce a number of candidates, many of which would likely be
false positives.
Then, MRFy could be re-run on these candidates with more computationally
demanding settings, and with a more strict $p$-value threshold.
MRFy computes $p$-values identically to SMURFLite: an extreme value 
distribution~\cite{Eddy:1998ut} is fitted to a distribution of raw scores,
and then a $p$-value is computed as $1-cdf\left( x \right)$ for any raw MRFy 
score $x$.
Computing the $p$-value accurately in the face of different search intensities
might require fitting multiple distributions, each for a different level of
search intensity.
Otherwise, if the distribution is obtained with an intensive search, then at
less-intensive search parameters, true positives may result in poor $p$-values;
similarly, if the distribution is obtained with a quick search, then 
more-intensive search parameters might result in false positives scoring 
comparatively well, and appearing to have good $p$-values.

As in Chapter~\ref{chapter:c3_smurflite}, we compared MRFy to 
HHPred~\cite{Soding:2005ff}.
As discussed, HHPred has an advantage in that it builds profiles based on all
of protein sequence space.
As a future enhancement to MRFy, we plan to introduce query profiles, so that
the MRFy alignment is to a sequence profile built from the query sequence,
rather than just the query sequence.
However, this will introduce a run-time performance hit in two ways.
First, the time to run a sequence homology search using the BLAST~\cite{Altschul:1997tl} 
family of tools can be significant, though the work on 
compressively-accelerated algorithms by Loh, et al.\cite{Loh:2012br} may reduce this 
impact.
Second, computing the Viterbi and $\beta$-pairing scores na\"{i}vely will
require time directly proportional to the number of sequences in the query
profile.
Representing these query sequences as sets of residue frequency vectors should
help, and there may be other approaches to consider.

We have demonstrated that MRFy is an improvement to SMURFLite, one that brings
the full power of a Markov random field to bear.
Thus far, only $\beta$-strand interactions lead to non-local interactions in the
MRFy Markov random field.
In the future, we will investigate fitting other secondary structural elements
(the $\alpha$-helices) into this model.
In addition, disulfide bonds, which can occur between cysteine residues and
have been shown to be highly conserved~\cite{Naamati:2009eg, Tirosh:2012iq}, 
would appear to fit easily into this model.

\chapter{Conclusion and Future Work} \label{conclusion}

\label{chapter:c5_conclusion}

\section{Contrasting Markov random field approaches}

We have explored two approaches to making the SMURF~\cite{Menke:2009, Menke:2010ti} Markov
random field model computationally tractable on all protein folds.
SMURF used multidimensional dynamic programming to exactly compute the optimal
energy function on a $\beta$-structural Markov random field, which was
computationally intractable when $\beta$-strands were highly interleaved.
In Chapter~\ref{chapter:c3_smurflite}, we demonstrated a method, SMURFLite, for 
simplifying the Markov random field itself, by removing only those nonlocal 
interactions that caused the computational complexity to grow beyond reasonable
bounds.
In addition, SMURFLite uses ``simulated evolution'' to mitigate, at least in 
part, the information loss that this simplification poses.

In Chapter~\ref{chapter:c4_mrfy}, we demonstrated an alternative method, MRFy,
for approximating a solution to the full SMURF Markov random field, without
discarding any $\beta$-strand information.
This method, too, benefits from simulated evolution, though this benefit seems
primarily confined to the protein superfamilies that have barely-adequate 
training data.

In essence, SMURFLite \emph{exactly} computes the solution to an 
\emph{approximation} of the SMURF Markov random field, while MRFy 
\emph{approximately} computes a solution to the \emph{exact} SMURF Markov random
field.

A natural extension of this work would be to combine the methodologies from
Chapters~\ref{chapter:c3_smurflite} and~\ref{chapter:c4_mrfy}.
One approach to this would be to use SMURFLite at a low interleave threshold,
such as 2, to produce an alignment that could serve as an initial guess for
MRFy's placement of $\beta$-strand residues.
Such an alignment would, at times, need to be modified to fit the 
$\beta$-strands that had been ignored by SMURFLite.

We also propose to evaluate this combined approach in comparison to the
newer threading approach, RaptorX~\cite{Peng:2011wx}, which incorporates
multiple-template alignments and solvent-accessibility information.

\section{Structurally consistent superfamilies}

The results from 
Chapters~\ref{chapter:c2_touring}~and~\ref{chapter:c3_smurflite}
suggest that a purely structural basis for remote homology detection may result
in, in some sense, an unfair test.
Some SCOP superfamilies exhibit structural inconsistency; this, as well as 
historical artifacts such as ``dustbin families'' as suggested  
by~\cite{Pethica:2012ds}, pose a considerable challenge.
In addition, the structural \emph{consistency} of the $\beta$-propeller folds 
appears to be relatively unusual at the fold level.
Given that the methods explored in this work rely upon high-quality structural
alignments that, ideally, preserve highly-conserved secondary-structural 
regions, efforts to improve these alignments would be beneficial.
Recent work in the field of structural alignment, including our 
work~\cite{Daniels:2012jr} and that of Wang, et al.~\cite{Wang:2012wq} may preserve more
sequence and secondary-structural similarity, occasionally at the expense of
small amounts of structural alignment quality.

We may also consider the task of remote homology detection when freed from the
occasional inconsistencies of SCOP; evaluating SMURFLite and MRFy on a purely
structurally-derived hierarchy, as described in 
Chapter~\ref{chapter:c2_touring}, may be worth exploring.

\section{MRFy with sequence profiles}

MRFy, along with SMURFLite, relies on computing an alignment of a single protein
sequence to a Markov random field.
As has been demonstrated by approaches such as PSI-BLAST~\cite{Altschul:1997tl} and 
BetaWrapPro~\cite{McDonnell:2005p497}, incorporating homologous profile data can improve both 
close and remote homology detection.
Given MRFy's stochastic search approach, we expect that the
less-discretized residue composition of each column of a query profile, as
compared with a single query sequence, would smooth out the fitness landscape
of MRFy's objective function (namely, the SMURF energy function) and thus
enable the local-improvement feature of MRFy's local search to be more 
efficient.
As discussed in Chapter~\ref{chapter:c4_mrfy}, two computational challenges 
arise:
how to quickly find close sequence homologs to build a profile, and how to
quickly score a profile against a Markov random field.
The work of Loh, et al.~\cite{Loh:2012br}, as well as our recent but as-yet 
unpublished work on compressively-accelerated genomic and protein sequence 
search algorithms, provides a partial solution to the first performance concern.

Regarding the second performance concern, how to quickly score a profile 
against a Markov random field, we note that HHPred~\cite{Soding:2005ff} and its relative, 
HHBlits~\cite{Remmert:2012cj} perform HMM-HMM alignment.
They solve the slightly simpler problem of aligning a hidden Markov model built
from a \emph{query} profile with a hidden Markov model built from training data.
This raises the question: could we perform MRF-MRF alignment, or at least
HMM-MRF alignment?
Our current model of Markov random fields requires a structural alignment in
order to annotate $\beta$-strands, but a hidden Markov model can be built from
a sequence alignment.
Could we build a hidden Markov model from a query sequence (and resulting 
profile), and align it to our Markov random field model?
This appears to bear further consideration.

\section{Extension to Other Protein Classes}

Both SMURFLite and MRFy differ from hidden Markov models, such as those employed
by HMMER~\cite{Eddy:1998ut}, only in that they incorporate non-local interactions between
residues participating in hydrogen bonds between $\beta$-strands.
This means that SMURFLite and MRFy are most at home in the SCOP class of 
``all beta proteins,'' and while we may also be able to show benefits in the
classes of ``alpha/beta'' and ``alpha+beta'' proteins, we would expect to
contribute little to, and in fact over-train on, the ``all alpha proteins,''
given that there are occasional $\beta$-strands in the $\alpha$-helical 
proteins, just as there are the converse (for example, the ``Barwin'' protein
discussed in Chapters~\ref{chapter:c3_smurflite}~and~\ref{chapter:c4_mrfy} has
four $\alpha$-helices in addition to its eight $\beta$-strands).

We intend to extend MRFy to incorporate $\alpha$-helix conditional 
probabilities, as explored by Cao, et al.~\cite{Cao:2012ex} within the context of 
HMMER.
At the simplest level, highly-conserved $\alpha$-helices in between 
$\beta$-strands should prevent those $\beta$-strands from occluding the helices;
incorporating the $\alpha$-helical residue propensities is an obvious extension
of the model.

\section{More Generalized Contact Maps}

Beyond specific secondary-structural elements, evolution conserves other 
non-local interactions.
Disulfide bonds, which occur between cysteine residues in proteins, anchor
certain protein structures, and are thus highly conserved.
MRFy's model of non-local interactions could easily incorporate these pairwise
bonds; the probability of a disulfide bond between two cysteine residues would
be close to 1, while the probability for any other pairing would be zero, or
close to zero.
Other, less common interactions such as peroxide and diselenide bonds may also
lend themselves to this model.

Beyond specific chemical bonds, we may consider any structural core that appears
to be highly conserved within a homologous group of proteins to be a candidate 
for our model of a Markov random field, encompassing non-local interactions.
The main prerequisite for such an extension would be the availability of 
adequate training data to build a model of conditional probabilities for the
non-local interactions in question.
%
%


\addcontentsline {toc}{chapter}{Bibliography}


\bibliographystyle{alpha}
\bibliography{thesis}

\begin{appendices}

\chapter*{Pairwise scores for $\beta$-structural proteins}

\rowcolors{2}{gray!25}{white}

\begin{small}

\begin{center}

\begin{sidewaystable*}[!t]
\caption{Pairwise scores (negative log of probability) for buried $\beta$-strands}
{\tiny 
\begin{tabular*}
{\textwidth}{@{\extracolsep{\fill}}llllllllllllllllllllll}\hline
 & A & C & D & E & F & G & H & I & K & L & M & N & P & Q & R & S & T & V & W & Y & X \\
A &   2.84 & 2.89 & 2.41 & 1.63 & 2.58 & 3.31 & 2.06 & 2.54 & 2.89 & 2.42 & 2.91 & 2.30 & 2.66 & 3.06 & 3.52 & 2.68 & 2.47 & 2.54 & 2.53 & 2.60 & 9.21 \\
C &   3.77 & 2.19 & 3.33 & 3.71 & 3.56 & 3.13 & 3.44 & 3.65 & 2.89 & 3.57 & 4.20 & 3.40 & 1.96 & 3.76 & 3.52 & 3.14 & 3.16 & 3.81 & 2.78 & 3.47 & 9.21 \\
D &   4.56 & 4.59 & 4.73 & 4.73 & 5.06 & 3.82 & 4.14 & 4.78 & 4.73 & 5.55 & 4.61 & 2.99 & 3.06 & 4.73 & 2.83 & 3.47 & 5.11 & 4.73 & 4.73 & 5.78 & 9.21 \\
E &   4.09 & 5.28 & 5.04 & 3.02 & 4.66 & 5.61 & 3.44 & 5.03 & 2.19 & 5.32 & 4.20 & 5.04 & 5.04 & 5.04 & 5.04 & 4.39 & 5.11 & 6.52 & 5.04 & 5.08 & 9.21 \\
F &   2.30 & 2.39 & 2.63 & 1.92 & 2.12 & 1.97 & 2.19 & 2.47 & 2.19 & 2.39 & 2.13 & 2.99 & 2.15 & 1.96 & 2.42 & 1.95 & 2.16 & 2.40 & 2.42 & 2.19 & 9.21 \\
G &   3.87 & 2.80 & 2.23 & 3.71 & 2.81 & 2.72 & 2.53 & 3.32 & 3.14 & 3.29 & 3.51 & 4.09 & 1.96 & 3.76 & 3.52 & 3.47 & 2.80 & 3.10 & 3.62 & 2.94 & 9.21 \\
H &   4.09 & 4.59 & 4.02 & 3.02 & 4.50 & 4.00 & 2.75 & 5.03 & 4.61 & 4.99 & 4.61 & 4.09 & 4.61 & 3.76 & 4.61 & 3.29 & 4.41 & 5.82 & 4.72 & 4.68 & 9.21 \\
I &   1.73 & 1.95 & 1.82 & 1.76 & 1.94 & 1.95 & 2.19 & 1.58 & 1.79 & 1.72 & 1.67 & 2.14 & 2.37 & 1.96 & 2.83 & 1.86 & 2.47 & 1.73 & 1.83 & 1.84 & 9.21 \\
K &   6.17 & 5.28 & 5.87 & 3.02 & 5.76 & 5.87 & 5.87 & 5.88 & 5.87 & 5.32 & 5.87 & 5.87 & 5.87 & 3.76 & 3.52 & 5.08 & 5.87 & 6.11 & 5.87 & 5.87 & 9.21 \\
L &   1.66 & 1.92 & 2.63 & 2.10 & 1.91 & 1.97 & 2.19 & 1.77 & 1.28 & 1.71 & 1.78 & 2.01 & 2.37 & 2.15 & 2.83 & 1.99 & 1.89 & 1.83 & 2.01 & 1.79 & 9.21 \\
M &   3.77 & 4.18 & 3.33 & 2.61 & 3.27 & 3.82 & 3.44 & 3.34 & 3.45 & 3.41 & 2.82 & 4.09 & 3.76 & 2.37 & 3.52 & 5.08 & 3.16 & 3.52 & 3.62 & 3.38 & 9.21 \\
N &   4.38 & 4.59 & 2.92 & 4.66 & 5.35 & 5.61 & 4.14 & 5.03 & 4.66 & 4.85 & 5.30 & 4.66 & 3.76 & 3.06 & 3.52 & 4.39 & 3.32 & 4.91 & 3.11 & 5.78 & 9.21 \\
P &   5.07 & 3.49 & 3.33 & 5.00 & 4.84 & 3.82 & 5.00 & 5.59 & 5.00 & 5.55 & 5.30 & 4.09 & 5.00 & 3.76 & 5.00 & 4.39 & 5.11 & 6.11 & 3.34 & 5.08 & 9.21 \\
Q &   5.48 & 5.28 & 5.00 & 5.00 & 4.66 & 5.61 & 4.14 & 5.19 & 2.89 & 5.32 & 3.92 & 3.40 & 3.76 & 5.00 & 3.52 & 3.98 & 5.11 & 5.01 & 4.03 & 5.78 & 9.21 \\
R &   6.17 & 5.28 & 3.33 & 5.23 & 5.35 & 5.61 & 5.23 & 6.29 & 2.89 & 6.24 & 5.30 & 4.09 & 5.23 & 3.76 & 5.23 & 5.08 & 3.72 & 5.01 & 4.03 & 4.17 & 9.21 \\
S &   3.77 & 3.34 & 2.41 & 3.02 & 3.31 & 4.00 & 2.35 & 3.76 & 2.89 & 3.84 & 5.30 & 3.40 & 3.06 & 2.66 & 3.52 & 2.44 & 5.11 & 4.12 & 3.11 & 3.83 & 9.21 \\
T &   3.53 & 3.34 & 4.02 & 3.71 & 3.51 & 3.31 & 3.44 & 4.34 & 3.64 & 3.71 & 3.36 & 2.30 & 3.76 & 3.76 & 2.14 & 5.08 & 3.72 & 3.47 & 4.03 & 3.38 & 9.21 \\
V &   1.50 & 1.88 & 1.54 & 3.02 & 1.64 & 1.50 & 2.75 & 1.50 & 1.79 & 1.55 & 1.61 & 1.79 & 2.66 & 1.56 & 1.32 & 1.99 & 1.37 & 1.36 & 2.08 & 1.62 & 9.21 \\
W &   3.97 & 3.34 & 4.03 & 4.03 & 4.15 & 4.51 & 4.14 & 4.09 & 4.03 & 4.22 & 4.20 & 2.48 & 2.37 & 3.06 & 2.83 & 3.47 & 4.41 & 4.57 & 3.34 & 3.58 & 9.21 \\
Y &   2.99 & 2.98 & 4.02 & 3.02 & 2.87 & 2.78 & 3.04 & 3.05 & 2.98 & 2.94 & 2.91 & 4.09 & 3.06 & 3.76 & 1.91 & 3.14 & 2.71 & 3.05 & 2.53 & 3.00 & 9.21 \\
X &   9.21 & 9.21 & 9.21 & 9.21 & 9.21 & 9.21 & 9.21 & 9.21 & 9.21 & 9.21 & 9.21 & 9.21 & 9.21 & 9.21 & 9.21 & 9.21 & 9.21 & 9.21 & 9.21 & 9.21 & 9.21 \\
\end{tabular*}}\\
\end{sidewaystable*}

\end{center}

\end{small}

\rowcolors{0}{white}{white}

\rowcolors{2}{gray!25}{white}

\begin{small}

\begin{center}

\begin{sidewaystable*}[!t]
\caption{Pairwise scores (negative log of probability) for exposed $\beta$-strands}
{\tiny 
\begin{tabular*}
{\textwidth}{@{\extracolsep{\fill}}llllllllllllllllllllll}\hline
 & A & C & D & E & F & G & H & I & K & L & M & N & P & Q & R & S & T & V & W & Y & X \\

A &  2.91 & 2.56 & 3.19 & 3.50 & 2.97 & 3.19 & 3.04 & 2.98 & 3.56 & 2.88 & 3.11 & 3.31 & 3.44 & 3.86 & 3.18 & 3.37 & 3.60 & 2.90 & 3.81 & 2.83 & 9.21 \\
C &  3.76 & 2.27 & 5.49 & 6.14 & 3.67 & 5.14 & 4.65 & 4.73 & 4.66 & 4.35 & 4.72 & 4.56 & 4.14 & 5.65 & 4.13 & 4.22 & 4.41 & 4.51 & 3.81 & 3.98 & 9.21 \\
D &  3.25 & 4.35 & 3.70 & 3.84 & 3.67 & 3.19 & 2.95 & 3.34 & 2.68 & 3.43 & 3.34 & 3.31 & 4.14 & 3.01 & 2.74 & 2.92 & 3.27 & 3.82 & 3.12 & 3.57 & 9.21 \\
E &  2.91 & 4.35 & 3.19 & 2.96 & 2.57 & 3.19 & 2.86 & 2.86 & 1.84 & 2.74 & 2.24 & 2.37 & 2.53 & 2.61 & 2.03 & 2.97 & 2.52 & 2.80 & 2.56 & 3.06 & 9.21 \\
F &  3.06 & 2.56 & 3.70 & 3.25 & 2.69 & 2.57 & 2.64 & 3.72 & 3.63 & 3.19 & 2.64 & 3.65 & 3.44 & 3.57 & 4.46 & 2.92 & 3.31 & 3.53 & 3.12 & 3.28 & 9.21 \\
G &  3.60 & 4.35 & 3.55 & 4.20 & 2.89 & 3.34 & 3.15 & 3.48 & 4.88 & 3.19 & 2.78 & 3.18 & 4.14 & 4.04 & 3.51 & 3.53 & 4.00 & 3.90 & 3.41 & 3.28 & 9.21 \\
H &  3.25 & 3.66 & 3.09 & 3.66 & 2.75 & 2.94 & 3.56 & 3.55 & 3.50 & 3.50 & 3.11 & 3.87 & 3.04 & 3.57 & 3.77 & 3.37 & 3.40 & 3.36 & 3.81 & 3.00 & 9.21 \\
I &  2.41 & 2.97 & 2.72 & 2.88 & 3.06 & 2.50 & 2.78 & 2.02 & 2.66 & 2.35 & 2.32 & 2.96 & 2.75 & 2.71 & 2.67 & 2.64 & 2.90 & 2.56 & 2.56 & 2.83 & 9.21 \\
K &  2.84 & 2.74 & 1.91 & 1.71 & 2.82 & 3.75 & 2.57 & 2.50 & 2.44 & 2.56 & 2.24 & 2.42 & 3.44 & 2.36 & 2.81 & 2.92 & 2.62 & 2.72 & 2.31 & 2.13 & 9.21 \\
L &  2.29 & 2.56 & 2.78 & 2.74 & 2.51 & 2.19 & 2.71 & 2.33 & 2.68 & 2.27 & 2.93 & 3.47 & 1.94 & 2.76 & 2.94 & 2.72 & 3.02 & 2.21 & 2.56 & 2.73 & 9.21 \\
M &  3.94 & 4.35 & 4.11 & 3.66 & 3.38 & 3.19 & 3.74 & 3.72 & 3.78 & 4.35 & 3.34 & 4.16 & 4.14 & 4.26 & 3.88 & 4.22 & 4.54 & 4.22 & 4.51 & 5.08 & 9.21 \\
N &  3.60 & 3.66 & 3.55 & 3.25 & 3.85 & 3.06 & 3.96 & 3.81 & 3.43 & 4.35 & 3.62 & 3.47 & 3.44 & 3.09 & 3.08 & 3.45 & 3.27 & 4.00 & 3.12 & 3.20 & 9.21 \\
P &  4.85 & 4.35 & 5.49 & 4.53 & 4.77 & 5.14 & 4.25 & 4.73 & 5.57 & 3.94 & 4.72 & 4.56 & 4.59 & 4.96 & 4.46 & 5.32 & 4.88 & 4.36 & 3.81 & 3.69 & 9.21 \\
Q &  3.76 & 4.35 & 2.85 & 3.10 & 3.38 & 3.53 & 3.27 & 3.17 & 2.97 & 3.25 & 3.34 & 2.69 & 3.44 & 2.47 & 3.36 & 3.24 & 2.82 & 3.08 & 2.71 & 3.06 & 9.21 \\
R &  2.66 & 2.41 & 2.16 & 2.10 & 3.85 & 2.57 & 3.04 & 2.71 & 3.01 & 3.00 & 2.53 & 2.26 & 2.53 & 2.94 & 3.59 & 2.68 & 2.46 & 2.49 & 2.20 & 2.83 & 9.21 \\
S &  2.91 & 2.56 & 2.40 & 3.10 & 2.37 & 2.65 & 2.71 & 2.75 & 3.18 & 2.84 & 2.93 & 2.69 & 3.44 & 2.88 & 2.74 & 2.43 & 2.27 & 3.08 & 2.90 & 2.68 & 9.21 \\
T &  2.66 & 2.27 & 2.27 & 2.17 & 2.28 & 2.65 & 2.26 & 2.53 & 2.40 & 2.67 & 2.78 & 2.04 & 2.53 & 1.99 & 2.05 & 1.79 & 1.84 & 2.17 & 3.81 & 3.13 & 9.21 \\
V &  2.15 & 2.56 & 3.01 & 2.64 & 2.69 & 2.74 & 2.40 & 2.38 & 2.68 & 2.04 & 2.64 & 2.96 & 2.19 & 2.43 & 2.27 & 2.79 & 2.36 & 2.17 & 2.31 & 2.37 & 9.21 \\
W &  4.85 & 3.66 & 4.11 & 4.20 & 4.07 & 4.04 & 4.65 & 4.17 & 4.07 & 4.19 & 4.72 & 3.87 & 3.44 & 3.86 & 3.77 & 4.40 & 5.79 & 4.10 & 3.81 & 4.38 & 9.21 \\
Y &  2.60 & 2.56 & 3.29 & 3.43 & 2.97 & 2.65 & 2.57 & 3.17 & 2.63 & 3.09 & 4.03 & 2.69 & 2.06 & 2.94 & 3.13 & 2.92 & 3.85 & 2.90 & 3.12 & 2.44 & 9.21 \\
X &  9.21 & 9.21 & 9.21 & 9.21 & 9.21 & 9.21 & 9.21 & 9.21 & 9.21 & 9.21 & 9.21 & 9.21 & 9.21 & 9.21 & 9.21 & 9.21 & 9.21 & 9.21 & 9.21 & 9.21 & 9.21
\end{tabular*}}\\
\end{sidewaystable*}

\end{center}

\end{small}

\rowcolors{0}{white}{white}

\end{appendices}


\end{document}